\documentclass[10pt]{article} 
\usepackage{latexsym}
\usepackage{epstopdf,caption,subcaption,graphicx,xcolor,hyperref}
\usepackage{tikz}
\usepackage{pgfgantt}
\usepackage[capitalise,noabbrev]{cleveref}
\usetikzlibrary{shapes,arrows,fit,calc,positioning}
\usetikzlibrary{decorations.pathreplacing}
\normalsize
\newtheorem{theorem}{Theorem}
\newtheorem{lemma}{Lemma}
\newtheorem{operation}{Operation}

\newtheorem{fact}{Fact}
\newtheorem{invariant}{Invariant}

\newtheorem{remark}{Remark}

\newtheorem{notation}{Notation}
\newtheorem{definition}{Definition}

\newenvironment{proof}{{\sc Proof. }}{\hfill$\Box$\vspace{0.1in}}
\title{An Approximation Algorithm for Covering Vertices by $4^+$-Paths}
\author{
Mingyang~Gong\thanks{Department of Computing Science, University of Alberta.  Edmonton, Alberta T6G 2E8, Canada.
	Emails: \texttt{\{mgong4, guohui\}@ualberta.ca}}
\and
Zhi-Zhong~Chen\thanks{Division of Information System Design, Tokyo Denki University. Saitama 350-0394, Japan.
  Email: \texttt{zzchen@mail.dendai.ac.jp}}
\thanks{Correspondence author.}
\and
Guohui~Lin$^*$
\and
Zhaohui~Zhan\thanks{Department of Computer Science, City University of Hong Kong. Hong Kong SAR, China.
  Email: \texttt{zhzhan3-c@my.cityu.edu.hk}}
}
\date{}
\begin{document}
\maketitle
\begin{abstract}
This paper deals with the problem of finding a collection of vertex-disjoint paths in a given graph 
$G=(V,E)$ such that each path has at least four vertices and the total number of vertices in these paths is maximized. 
The problem is NP-hard and admits an approximation algorithm which achieves a ratio of $2$ and runs in 
$O(|V|^8)$ time. The known algorithm is based on time-consuming local search, and its authors ask whether 
one can design a better approximation algorithm by a completely different approach. 
In this paper, we answer their question in the affirmative by presenting a new approximation algorithm for the problem. 
Our algorithm achieves a ratio of 1.874 and runs in $O(\min\{|E|^2|V|^2, |V|^5\})$ time. 
Unlike the previously best algorithm, ours starts with a maximum matching $M$ of $G$ and then tries to 
transform $M$ into a solution by utilizing a maximum-weight path-cycle cover in a suitably constructed graph. 

\end{abstract}

\newpage 

\section{Introduction}\label{sec:intro}
Throughout this paper, a graph always means a simple undirected graph without parallel edges or self-loops,
and an approximation algorithm always means one running in polynomial time.
Let $k$ be a positive integer. 
Given a graph $G = (V, E)$, $MPC^{k+}_v$ is the problem of finding a collection of vertex-disjoint paths 
each with at least $k$ vertices in $G$ so that the total number of vertices in these paths is maximized. 
Note that we can assume that each path in the output collection has at most $2k-1$ vertices. 
This is because we can split a path having $2k$ or more vertices into two or more paths each having at least $k$ and at most $2k-1$ vertices. 
$MPC^{k+}_v$ has numerous real-life applications such as transportation networks \cite{GFL22a}. 
In this paper, we mainly focus on $MPC^{4+}_v$. 

On one hand, $MPC^{k+}_v$ is related to many important optimization problems. 
For example, Berman and Karpinski \cite{BK06} consider {\em the maximum path cover problem}, 
which is the problem of finding a collection of vertex-disjoint paths in a given graph so that the total number of edges in the paths is maximized. 
For other related path cover problems with different objectives, the reader is referred to 
\cite{BK06, AN07, PH08, AN10, RT14, CCC18, GW20, CCL21} for more details.
On the other hand, $MPC^{k+}_v$ can be viewed as a special case of  {\em the maximum-weight $(2k - 1)$-set packing problem} 
because the former can be easily reduced to the latter as follows. 
Recall that an instance of the latter problem is a collection ${\cal C}$ of sets each having a non-negative weight and at most $2k-1$ elements. 
The objective is to select a collection of pairwise-disjoint sets in ${\cal C}$ so that the total weight of the selected sets is maximized. 
To reduce $MPC^{k+}_v$ to the maximum-weight $(2k-1)$-set packing problem,
it suffices to construct an instance ${\cal C}$ of the latter problem from a given instance graph $G$ of $MPC^{k+}_v$, 
where ${\cal C}$ is the collection of all paths of $G$ with at least $k$ and at most $2k-1$ vertices and 
the weight of each path $P$ in ${\cal C}$ is the number of vertices in $P$. 
This reduction leads to an approximation algorithm for $MPC^{k+}_v$ achieving a ratio of $k$ 
because the maximum-weight $(2k - 1)$-set packing problem can be approximated within a ratio of $k$ \cite{H83} 
or within a slightly better ratio of $k-\frac 1{63,700,992}+\epsilon$ \cite{N21} for any $\epsilon>0$.

$MPC^{k+}_v$ can be solved in polynomial time if $k\le 3$ \cite{CCL21}, but is NP-hard otherwise \cite{KLM22}.
Kobayashi {\em et al.} \cite{KLM22} design an approximation algorithm for $MPC^{4+}_v$ achieving a ratio of 4.
Afterwards, Gong et al. \cite{GFL22a} give the formal definition of $MPC^{k+}_v$ and present an approximation algorithm for $MPC^{k+}_v$ 
which achieves a ratio of $\rho(k) \le 0.4394k + 0.6576$ and runs in $O(|V|^{k+1})$ time. 
The core of their algorithm is three local improvement operations, 
each of which increases the number of vertices in the current solution by at least~1 if it is applicable. 
The algorithm stops when none of the three operations is applicable.
They employ an amortization scheme to analyze the approximation ratio of their algorithm by
assigning the vertices in the optimal solution to the vertices of the solution outputted by their algorithm. 
For the special case where $k = 4$,
they design two more local improvement operations to increase the number of vertices or the number of paths with exactly $4$ vertices in the current solution, 
and then use a more careful amortization scheme to prove that 
the approximation ratio of their algorithm is bounded by $2$ although the running time jumps to $O(|V|^8)$. 
As an open question, they ask whether one can design better approximation algorithms for the problem by completely different approaches. 

In this paper, we answer their open question in the affirmative for the case where $k = 4$. 
Motivated by the approaches in \cite{K09, CKM10, CCL21} for similar problems, one may want to design an approximation algorithm for $MPC^{k+}_v$ 
by first computing a maximum path-cycle cover ${\cal C}$ of the input graph $G$ and then transforming ${\cal C}$ into a solution for $G$. 
Unfortunately, this approach {\color{black}to maximizing the number of edges} does not seem to work. 
Our new idea for designing a better approximation algorithm for $MPC^{4+}_v$ is to let the algorithm start by computing a maximum matching $M$ in the input graph $G$. 
The intuition behind this idea is that the paths in an optimal solution for $G$ can cover at most $\frac 5{\color{black}2} |M|$ vertices. 
So, it suffices to find a solution for $G$ {\color{black}of} which the paths cover a large fraction of the endpoints of the edges in $M$. 
To this purpose, our algorithm then constructs a maximum-weight path-cycle cover $C$ in an auxiliary graph suitably constructed from $M$ and $G$. 
Our algorithm further tries to use the edges in $C$ to connect a large fraction of the edges of $M$ into paths with at least four vertices. 
If the algorithm fails to do so, then it will be able to reduce the problem to a smaller problem and in turn uses recursion to get a good solution. 

The rest of the paper is organized as follows.
Section~2 gives some basic definitions.
Section~3 presents the algorithm for $MPC^{4+}_v$.
Section~4 analyzes the approximation ratio of the algorithm.
Lastly, Section~5 concludes the paper with the main algorithm design ideas and some possible future research.

\section{Basic Definitions}\label{sec:def}
{\color{black}A number of symbols and terms are specified in Notations~\ref{nota01}--\ref{nota10}, and used in the rest of the paper.}

\begin{notation}
\label{nota01}
For a graph $G$, $V(G)$ denotes the vertex set of $G$ and $E(G)$ denotes the edge set of $G$.
\end{notation}

Let $G$ be a graph. For a subset $F$ of $E(G)$, we use $V(F)$ to denote 
the set $\{v\in V(G) \mid v$ is an endpoint of an edge in $F\}$.
A {\em spanning subgraph} of $G$ is a subgraph $H$ with $V(H)=V(G)$. 
For a set $F$ of edges in $G$, $G - F$ denotes the spanning subgraph $(V(G),E(G)\setminus F)$.  
In contrast, for a set $F$ of edges with $V(F)\subseteq V(G)$ and $F\cap E(G) = \emptyset$, 
$G + F$ denotes the graph $(V(G),E(G)\cup F)$.
The {\em degree} of a vertex $v$ in $G$, denoted by $d_G(v)$, is the number of edges incident to $v$ in $G$. 
A vertex $v$ of $G$ is {\em isolated} in $G$ if $d_G(v) = 0$. The {\em subgraph induced by} a subset $U$ of 
$V(G)$, denoted by $G[U]$, is the graph $(U, E_U)$, where $E_U=\{\{u,v\}\in E(G) \mid u, v\in U\}$. 
Two vertex-disjoint subgraphs of $G$ are {\em adjacent} in $G$ if $G$ has an edge between them. 

A {\em cycle} in $G$ is a connected subgraph of $G$ in which each vertex is of degree~2. 
A {\em path} in $G$ is either a single vertex of $G$ or a connected subgraph of $G$ 
in which exactly two vertices (called the {\em endpoints}) are of degree~$1$ and 
the others (called the {\em internal vertices}) are of degree~$2$. 
A {\em path component} of $G$ is a connected component of $G$ that is a path. 
If a path component is an edge, then it is called an {\em edge component}. 
The {\em order} of a cycle or path $P$, denoted by $|P|$, is the number of vertices in $P$. 
A {\em $k$-path} of $G$ is a path of order $k$ in $G$, while a {\em $k^+$-path} of $G$ is a path of order $k$ or more in $G$. 
A {\em triangle} of $G$ is a cycle of order~$3$ in $G$. 
A {\em matching} of $G$ is a (possibly empty) set of edges of $G$ in which no two edges share an endpoint. 
A {\em maximum matching} of $G$ is a matching of $G$ whose size is maximized over all matchings of $G$. 
A {\em path-cycle cover} of $G$ is a set $F$ of edges in $G$ such that 
in the spanning subgraph $(V(G), F)$, the degree of each vertex is at most~$2$.
A {\em star} is a connected graph in which at most one vertex is of degree $\ge 2$ and each of the remaining vertices is of degree~$1$.
The vertex of degree $\ge 2$ is called the {\em center}, while the other vertices are the {\em satellites} of the star.
{\color{black}Note that a single edge is not readily a star,
but becomes so after one vertex is chosen to be the center and accordingly the other becomes the satellite.}

\begin{notation}
\label{nota02}
For a graph $G$, 
\begin{itemize}
\parskip=0pt
\item
	$OPT(G)$ denotes an optimal solution for the instance graph $G$ of $MPC^{4+}_v$,
	and $opt(G)$ denotes the total number of vertices in $OPT(G)$;
\item
	$ALG(G)$ denotes the solution for $G$ outputted by a specific algorithm,
	and $alg(G)$ denotes the total number of vertices in $ALG(G)$.
\end{itemize}
\end{notation}

As aforementioned, we have the following fact:

\begin{fact}
\label{fact01}
The order of each path in any feasible solution is between $4$ and $7$.
\end{fact}

\section{The Algorithm for $MPC^{4+}_v$}\label{sec:1stAlgo}
Throughout the remainder of this paper, we fix an instance $G$ of $MPC^{4+}_v$ for discussion. 
Let $n = |V(G)|$ and $m = |E(G)|$. 

Our algorithm for $MPC^{4+}_v$ consists of multiple phases.
In the first phase, it computes a maximum matching $M$ in $G$ in $O(\sqrt{n}m)$ time \cite{MV80},
initializes a subgraph $H = (V(M), M)$, then repeatedly modifies $H$ and $M$ (cf. Section~\ref{subsec:k=4}) 
in such a way that $M$ always remains to be a maximum matching of $G$ with $M\subseteq E(H)$ 
and $H$ eventually becomes a graph in which 
each connected component is an edge in $M$, a triangle with one edge in $M$, a star with one edge in $M$, 
or a 5-path with two edges in $M$. 

\begin{lemma}
\label{lemma01}
$|V(M)| \ge \frac{4}{5} opt(G)$. 
\end{lemma}
\begin{proof}
Consider an arbitrary path $P$ in $OPT(G)$.
Let $e_1$, \ldots, $e_\ell$ be the edges of $P$ and suppose that they appear in $P$ in this order from one endpoint to the other.
Obviously, $M_o = \{e_i \mid i \mbox{ is odd}\}$ is a matching. If $\ell$ is odd, $V(P) = V(M_o)$; otherwise, exactly one vertex of $P$ is not in $V(M_o)$. 
We claim that $|V(M_o)| \ge \frac{4}{5} |V(P)|$. 
This is clearly true if $\ell$ is odd. So, we may assume below that $\ell$ is even. 
Then, $\ell\ge 4$ because $\ell+1=|V(P)|\ge 4$ and both $\ell$ and $4$ are even. 
Now, since $|V(M_o)| \ge \frac{\ell}{\ell+1} |V(P)| $, we have $|V(M_o)| \ge \frac{4}{5} |V(P)|$.
Note that  {\color{black}$\cup_P M_o$ is a matching and} $opt(G) = \sum_P |V(P)|$, where $P$ ranges over all paths in $OPT(G)$.
So, by the claim, $|V(M)| \ge \frac{4}{5} opt(G)$.
\end{proof}

\subsection{Modifying $H$ and $M$}\label{subsec:k=4}
We here describe a process for modifying $H$ and $M$ iteratively. 
The process consists of two steps. 
During the first step, the following will be an invariant.

\begin{invariant}
\label{inva01}
$M \subseteq E(H)$ and each connected component $K$ of $H$ is an edge or a $5$-path.
Moreover, if $K$ is an edge, then this edge is in $M$;
if $K$ is a $5$-path, then the two edges of $E(K)$ incident to the endpoints of $K$ are in $M$.
\end{invariant}

Initially, Invariant~\ref{inva01} clearly holds.

\begin{definition}
\label{def01}
An {\em augmenting triple} with respect to $H$ is a triple $(u_0, e_0 = \{v_0, w_0\}, e_1 = \{v_1, w_1\})$ such that
$u_0 \in V(G) \setminus V(H)$, both $e_0$ and $e_1$ are edge components of $H$, and one of the following two conditions holds:
\begin{itemize}
\parskip=0pt
\item[C1.]
	$\{u_0, v_0\}, \{u_0, v_1\} \in E(G)$.
\item[C2.]
	$\{u_0, v_0\}, \{w_0, v_1\} \in E(G)$.
\end{itemize}
\end{definition}

\begin{definition}
\label{def02}
{\em Modifying $H$ and $M$ with an augmenting triple $(u_0, e_0, e_1)$ w.r.t. $H$} is the operation of modifying $H$ and $M$ as follows:
\begin{description}
\parskip=0pt
\item[\rm Case~1:]
	Condition~C1 holds. In this case, add $u_0$ and the edges $\{u_0, v_0\}$, $\{u_0, v_1\}$ to $H$.
\item[\rm Case~2:]
	Condition~C2 holds. In this case, add $u_0$ and the edges $\{u_0, v_0\}$, $\{w_0, v_1\}$ to $H$ and then modify $M$ by replacing $e_0$ with $\{u_0, v_0\}$. 
\end{description}
\end{definition}

It is possible that an augmenting triple $(u_0, e_0, e_1)$ satisfies both Conditions~C1 and~C2. 
If this happens, then we prefer Condition~C1, i.e., we modify $H$ as in Case~1. 

The first step of the modification process is as follows.
\begin{description}
\parskip=0pt
\item[Step 1.1]
	Repeatedly modify $H$ and $M$ with an augmenting triple until no such triple exists. 
\end{description}

Modifying $M$ and $H$ with a given augmenting triple takes $O(1)$ time and produces one more 5-path in $H$. 
So, there are at most $O(n)$ repetitions in Step~1.1. 
To decide whether there is an augmenting triple satisfying Condition~C1, 
it suffices to check, for each vertex $u \in V(G)\setminus V(H)$, whether two edges incident to $u$ in $G$ can be used to connect two edge components of $H$ into a 5-path of $G$. 
So, this takes $O(\min\{m^2, n^3\})$ time. 
Similarly, to decide whether there is an augmenting triple satisfying Condition~C2, it suffices to check, 
for each edge component $e = \{v, w\}$ of $H$, whether $v$ has a neighbor $u \in V(G)\setminus V(H)$ and $w$ has a neighbor in $V(M \setminus \{e\})$. 
So, this takes $O(\min\{m^2, n^3\})$ time, too. In total, Step~1.1 takes $O(\min\{m^2n, n^4\})$ time.
We have the following lemma on $H$ and $M$ when Step 1.1 terminates.

\begin{lemma}
\label{lemma02}
When Step 1.1 terminates, the following statements on $H$ and $M$ hold:
\begin{enumerate}
\parskip=0pt
\item
	$M$ is a maximum matching in $G$. 
\item
	Each connected component of $H$ is a $5$-path or an edge of $M$.
	Moreover, if it is a $5$-path, then the two edges incident to the endpoints are in $M$.
\item
	If $e = \{v, w\}$ is an edge component of $H$ and $u$ is a vertex of $V(G)\setminus V(H)$ such that $\{u, v\} \in E(G)$,
	then an edge of $G$ can connect $u$ ($w$, respectively) only to $w$ ($u$, respectively) or the internal but not 
	the middle vertices of $5$-paths in $H$.
\end{enumerate}
\end{lemma}
\begin{proof}
The first statement is true {\color{black}since $M$ is modified by edge swapping only during Step 1.1.

The second statement is obvious too,
since when an augmenting triple with respect to the current $H$ and $M$ is identified,
two edge components and an outside vertex are merged into a $5$-path component, which stays untouched till Step 1.1 terminates.

We next prove the third statement.
Consider the vertex $u \in V(G)\setminus V(H)$. No edge of $G$ can connect $u$ to a vertex not in $V(M)$ or an endpoint of a $5$-path, 
due to $M$ being a maximum matching. No edge of $G$ can connect $u$ to another edge component of $H$ than $e$ either, 
since otherwise Step 1.1 would still be applicable.
It follows that the only possible neighbors of $u$ are the internal but not the middle vertices of $5$-paths, besides the vertices $v$ and $w$ of $e$.

The third statement in the lemma holds for the vertex $w$ for the same reasons,
as one can swap the edge $\{u, v\}$ with the edge $e$ of $H$ to obtain an essentially equivalent graph.}
\end{proof}

By Lemma~\ref{lemma02}, we continue to modify $H$ (but not $M$) in the next step to add edges connecting the outside vertices and the edge components.

\begin{description} 
\parskip=0pt
\item[Step 1.2]
	Add all those edges $\{u, v\} \in E(G)$ such that $u \in V(G)\setminus V(H)$ and $v$ is an endpoint of an edge component of $H$,
	as well as their endpoints $u$, to $H$. 
\end{description}

Step~1.2 is done in $O(m)$ time.
We have the next lemma on $H$ and $M$ at the end of Step 1.2.

\begin{lemma}
\label{lemma03}
Suppose $H$ and $M$ have been modified as in the above Steps 1.1--1.2 in $O(\min\{m^2 n, n^4\})$ time.
The following statements hold.
\begin{enumerate}
\parskip=0pt
\item
	$M$ is a maximum matching in $G$. 
\item
	Each connected component $K$ of $H$ is a $5$-path, an edge, a triangle, or a star.
	Moreover, if $K$ is a $5$-path, then the two edges incident to the endpoints are in $M$;
	otherwise, exactly one edge of $E(K)$ is in $M$.
\item
	If $K_1$ and $K_2$ are two different connected components of $H$ such that 
	there is an edge $\{v_1, v_2\} \in E(G)$ with $v_i \in V(K_i)$, and either $K_1$ is a star and $v_1$ is its satellite or $K_1$ is a triangle, 
	then $K_2$ is a $5$-path and $v_2$ is an internal but not the middle vertex of $K_2$. 
\item
	For each vertex $u \in V(G)\setminus V(H)$, every neighbor of $u$ in $G$ is an internal but not the middle vertex of a $5$-path in $H$.
\end{enumerate}
\end{lemma}
\begin{proof}
The first statement follows from Lemma~\ref{lemma02}, since $M$ stays untouched during Step 1.2.

The second statement follows from Lemma~\ref{lemma02} and the fact that when Step 1.1 terminates,
the two endpoints of an edge component of $H$ cannot be adjacent to two distinct outside vertices due to $M$ being a maximum matching in $G$.
That is, if the two endpoints of an edge component of $H$ are not adjacent to any outside vertex, then it remains as an edge component at the end of Step 1.2;
if the two endpoints of an edge component of $H$ are adjacent to a common outside vertex, then it becomes a triangle component at the end of Step 1.2;
otherwise, exactly one of the two endpoints of an edge component of $H$ is adjacent to one or more outside vertices,
and then it becomes a star component at the end of Step 1.2, with the endpoint being the center.
Note that all these possibilities are originated from an edge component and thus they contain exactly one edge of $M$.

For the third statement, one sees that we may assume without loss of generality that $v_1$ is an outside vertex of $H$ when Step 1.1 terminates,
i.e., the same as the vertex $u$ in Lemma~\ref{lemma02};
therefore, by Lemma~\ref{lemma02} again $v_2$ is an internal but not the middle vertex of a $5$-path.
The last statement holds again due to $M$ being a maximum matching in $G$.
\end{proof}

\subsection{Bad components and rescuing them}\label{subsec:rescue}
We consider the subgraph $H$ and the maximum matching $M$ at the end of Step 1.2.

\begin{definition}
\label{def03}
A {\em bad component} of $H$ is a connected component that is not a $5$-path.
\end{definition}

{\color{black}In the sequel, a component always means a {\em connected} component.}
In the second phase, we {\em rescue} as many bad components of $H$ as possible, by performing three steps of operations.

\begin{description}
\parskip=0pt
\item[Step 2.1] 
Construct a spanning subgraph $G_1$ of $G$ of which the edge set consists of all the edges $\{v_1, v_2\}$ of $G$ 
such that $v_1$ and $v_2$ appear in different components of $H$ and at least one of the components is bad. 
\end{description}

\begin{definition}
\label{def04}
A set $F$ of edges in $G_1$ {\em saturates a bad component} $K$ of $H$ if at least one edge in $F$ is incident to a vertex of $K$.
The {\em weight} of $F$ is the number of bad components saturated by $F$.
\end{definition}

\begin{lemma}
\label{lemma04}
A maximum-weighted path-cycle cover in $G_1$ can be computed in $O(\min\{nm,n^2\}\log n)$ time.
\end{lemma}
\begin{proof}
The proof is a reduction to the maximum-weight $[f, g]$-factor problem. 
Recall that for two functions $f$ and $g$ mapping each vertex $v$ of a{\color{black}n edge-weighted} graph $G'$ to
two non-negative integers $f(v), g(v)$ with $f(v) \le g(v)$,
an {\em $[f,g]$-factor} of $G'$ is a set $F$ of edges in $G'$ such that in the spanning subgraph $(V(G'), F)$, 
the degree of each vertex $v$ is at least $f(v)$ and at most $g(v)$. 
The {\em weight} of an $[f,g]$-factor $F$ of $G'$ is the total weight of the edges in $F$. 
Given $G'$, $f$, and $g$, a maximum-weight $[f,g]$-factor of $G'$ can be computed in $O(m'n'\log n')$ time~\cite{Gab83}, 
where $m'=|E(G')|$ and $n'=|V(G')|$. 

Let $B_1, B_2, \ldots, B_h$ be the bad components of $H$.
We construct an auxiliary edge-weighted graph $G' = (V(G) \cup X, E(G_1) \cup F_1 \cup F_2)$ as follows:
\begin{itemize}
\parskip=0pt
\item $X = \{x_i, y_i, z_i \mid 1 \le i \le h \}$. 

\item $F_1 = \{\{x_i, v\}, \{y_i, v\} \mid v \in V(B_i), 1 \le i \le h\}$ and 
$F_2 = \{\{x_i, z_i\}, \{y_i, z_i\} \mid 1 \le i \le h\}$.

\item The weight of each edge in $E(G_1) \cup F_1$ is $0$ while the weight of each edge in $F_2$ is $1$.

\item For each vertex $v \in \bigcup_{i = 1}^h V(B_i)$, let $f(v) = g(v) = 2$. 

\item For each $v \in V(G) - \bigcup_{i = 1}^h V(B_i)$, let $f(v) = 0$ and $g(v) = 2$.

\item For each $i \in\{ 1, 2, \ldots, h\}$, $f(x_i) = f(y_i) = f(z_i) = 0$ and $g(x_i) = g(y_i) = |V(B_i)|$, $g(z_i) = 1$.
\end{itemize}

We next prove that the maximum weight of an $[f,g]$-factor of $G'$ equals the maximum weight of a path-cycle cover of $G_1$. 

Given a maximum-weight path-cycle cover $C$ of $G_1$, we can obtain an $[f, g]$-factor $F$ for $G'$ as follows: 
Initially, we set $F = C$. 
Then, for each bad component $B_i$ and each vertex $v$ in $B_i$, 
we perform one of the following according to the degree of $v$ in the graph $(V(G), F)$. 
\begin{itemize}
\parskip=0pt
\item If the degree of $v$ in the graph $(V(G), F)$ is~0, 
then add the edges $\{v, x_i\}, \{v, y_i\}$ to $F$.

\item If the degree of $v$ in the graph $(V(G), F)$ is~1, then add the edge $\{v, x_i\}$ to $F$, 
and further add the edge $\{y_i, z_i\}$ to $F$ if it has not been added to $F$.

\item If the degree of $v$ in the graph $(V(G), F)$ is~2, then add the edge $\{y_i, z_i\}$ to $F$ 
if it has not been added to $F$.
\end{itemize}
Clearly, $F$ is an $[f,g]$-factor of $G'$. 
We claim that the weight of $F$ is no less than that of $C$.
To see this, consider a bad component $B_i$ saturated by $C$. 
Then, there exists a vertex $v$ in $B_i$ such that $C$ contains an edge incident to $v$. 
Hence, by the construction of $F$, $F$ contains $\{y_i, z_i\}$. 
Since the weight of $\{y_i, z_i\}$ is $1$, the claim holds. 

Conversely, given a maximum-weight $[f, g]$-factor $F$ of $G'$, we obtain a subset $C$ of $E(G_1)$ with $C = E(G_1) \cap F$. 
Since $g(v) = 2$ for each vertex $v \in V(G_1)$, $C$ is a path-cycle cover of $G_1$.
We claim that the weight of $C$ is no less than that of $F$. 
To see this, consider a bad component $B_i$ such that $\{x_i, z_i\}$ or $\{y_i, z_i\}$ is in $F$.
Since $g(z_i)=1$, exactly one of $\{x_i, z_i\}$ and $\{y_i, z_i\}$ is in $F$. 
Without loss of generality, we assume $\{x_i, z_i\}$ is in $F$. 
Then, there exists a vertex $v$ in $B_i$ such that the edge $\{v, x_i\}$ is not in $F$. 
Since $f(v) = g(v) = 2$, $C$ contains an edge incident to $v$ and hence $C$ saturates $B_i$. 
So, the claim holds. 

By the above two claims,  the maximum weight of an $[f,g]$-factor of $G'$ equals the maximum weight 
of a path-cycle cover of $G_1$. 
Now, since $|V(G')|\le 4n$ and $|E(G')|\le m+4n$, the running time is bounded by $O(\min\{mn, n^2\}\log n)$.
So, the lemma holds.
\end{proof}

%
\begin{description}
\parskip=0pt
\item[Step 2.2.]
	Compute a maximum-weight path-cycle cover $C$ of $G_1$ (as in the proof of Lemma~\ref{lemma04}).
\item[Step 2.3.]
	As long as $C$ contains an edge $e$ such that $C \setminus e$ has the same weight as $C$, 
	repeatedly remove $e$ from $C$.
\end{description}

\begin{notation}
\label{nota04}
\ \\
\vspace{-5mm}
\begin{itemize}
\parskip=0pt
\item
	$G_1$ denotes the spanning subgraph of $G$ constructed in Step 2.1.
\item
	$C$ denotes the maximum weight path-cycle cover of $G_1$ computed at the end of Step 2.3.
\item
	$M_C$ denotes the subset of the maximum matching $M$ containing those edges in $5$-paths of $H$ 
	or in bad components of $H$ saturated by $C$. 
\end{itemize}
\end{notation}

The next lemma will be crucial for analyzing the approximation ratio of our algorithm.

\begin{lemma}
\label{lemma05}
We have $|V(M_C)| \ge \frac 45 opt(G)$.
\end{lemma}
\begin{proof}
Let $m_b$ be the total number of edges of $M$ contained in the bad components of $H$. 
Let $B_1, \ldots, B_h$ be the bad components such that no edge in $OPT(G)$ is incident to any vertex of $B_i$, $i = 1, \ldots, h$.
Note that $E(G_1) \cap E(OPT(G))$ is a path-cycle cover of $G_1$ with weight $m_b - h$.

Let $\ell$ be the number of bad components not saturated by $C$. 
Then, $|M| = |M_C| + \ell$ since each bad component has exactly one edge in $M$ by Lemma~\ref{lemma03}.
Moreover, the weight of the path-cycle cover $C$ is $m_b -\ell$.
Since we compute a maximum-weight path-cycle cover of $G_1$ by Lemma \ref{lemma04},
$m_b - h \le m_b - \ell$ and in turn $h \ge \ell$. 

A crucial point is that for each bad component $B_{i}$ with $1\le i\le h$, 
no vertex of $B_{i}$ can appear in $OPT(G)$ because $|opt(B_{i})| = 0$ and 
$OPT(G)$ has no edge connecting $B_i$ to the outside of $B_i$. 
By this point, $OPT(G)$ is actually an optimal solution for the graph $G_o$ obtained from $G$ by removing the vertices of $B_i$ for every $i\in\{1,\ldots, h\}$. 
So, by Lemma~\ref{lemma01}, $|V(M_o)| \ge \frac 45 opt(G)$, where $M_o$ is a maximum matching in $G_o$. 

Note that $M_o\bigcup \left(\bigcup^{h}_{i=1} E(B_i) \cap M\right)$ is a matching of $G$ and 
its size is $|M_o| + h$ because $|E(B_i) \cap M|=1$ by Lemma~\ref{lemma03}.
Since $M$ is a maximum matching of $G$, $|M_o| + h \le |M|$. 
Recall that $h \ge \ell$ and $|M| = |M_C| + \ell$. 
Hence, $|M_o| \le |M_C|$ and $|V(M_C)| \ge \frac 45 opt(G)$.
\end{proof}

\subsection{Structure of composite components of $H+C$}
By Lemma~\ref{lemma05}, $|V(M_C)|$ is relatively large compared to $opt(G)$. 
Intuitively speaking, in order to obtain a good approximate solution for $G$, it suffices to focus on $M_C$ {\color{black}instead of its superset $M$}.
That is, we may ignore the edges of $M$ in the bad components {\em not} saturated by $C$. 

\begin{notation}
\label{nota05}
\ \\
\vspace{-5mm}
\begin{itemize}
\parskip=0pt
\item
	$H+C$ denotes the spanning subgraph $(V(G), E(H)\cup C)$.
	In the sequel, we use $K$ to refer to a component in $H+C$.
\item
	$(H+C)_m$ denotes the graph obtained from $H+C$ by contracting each component of $H$ into a single node. 
	In other words, the nodes of $(H+C)_m$ one-to-one correspond to the components of $H$ and
	two nodes are adjacent in $(H+C)_m$ if and only if $C$ contains an edge between the two corresponding components.

	We use $(K)_m$ to refer to the component of $(H+C)_m$ corresponding to the component $K$ in $H+C$.
\end{itemize}
\end{notation}

\begin{definition}
\label{def05}
A {\em composite component} $K$ of $H+C$ is one that contains two or more components of $H$, which are connected through the edges of $C$.

In contrast, an {\em isolated component} $K$ of $H+C$ is one that contains exactly one component of $H$.
\end{definition}

\begin{lemma}
\label{lemma06}
For each component $(K)_m$ of $(H+C)_m$ (see Notation~\ref{nota05}), the following statements hold:
\begin{enumerate}
\parskip=0pt
\item $(K)_m$ is an isolated node, an edge, or a star.

\item If $(K)_m$ is an edge, then at least one endpoint of $(K)_m$ corresponds to a bad component of $H$.

\item If $(K)_m$ is a star, then each satellite of $(K)_m$ corresponds to a bad component of $H$.
\end{enumerate}
\end{lemma}
\begin{proof}
If $K$ is isolated in $H+C$, then $(K)_m$ is an isolated node in $(H+C)_m$.
Otherwise, $K$ is a composite component of $H+C$.
Suppose $(K)_m$ contains a $4^+$-path.
Let $v_1, v_2, v_3, v_4$ be the first four nodes of such a $4^+$-path from one endpoint to the other.
By Step~2.3, we can remove the edge of $C$ corresponding to $\{v_2, v_3\}$ such that the weight of $C$ is unchanged since no node of $(K)_m$ becomes isolated.
Such a contradiction shows that there is no $4^+$-path, and similarly no cycle, in $(K)_m$.
It follows that $(K)_m$ is either an edge or a star.
This proves the first statement.

The other two statements follow from the construction of $G_1$ in Step 2.1, the computation of $C$ in Steps 2.2 and 2.3, and the definition of $(H+C)_m$.
\end{proof}

By the second statement in Lemma~\ref{lemma06},
when $(K)_m$ is an edge, we choose an endpoint corresponding to a bad component of $H$ as the satellite, while the other endpoint as the center.
This way, an edge becomes a star.

\begin{definition}
\label{def06}
For each composite component $K$ of $H+C$, its {\em center element} is the component of $H$ corresponding to the center of $(K)_m$,
and it is denoted as $K_c$ in the sequel;
the other components of $H$ contained in $K$ are the {\em satellite elements} of $K$.

A {\em center} ({\em satellite}, respectively) {\em element} of $H+C$ is a center (satellite, respectively) element of some composite component of $H+C$;
an isolated $5$-path of $H+C$ is also called a center element.
\end{definition}

We remark that an isolated bad component of $H+C$ is not saturated by $C$, it contains an edge of $M \setminus M_C$,
and it is ignored from further discussion.

\begin{lemma}
\label{lemma07} 
The following statements hold:
\begin{enumerate}
\parskip=0pt
\item
	Each center element $K_c$ of $H+C$ is a $5$-path, an edge or a star but not a triangle of $H$;
	each satellite element $S$ of $H+C$ is an edge, a star or a triangle but not a $5$-path of $H$.
\item
	Suppose $v$ is a vertex of $K_c$ and some satellite element $S$ is adjacent to $v$ in $H+C$.
	If $K_c$ is a star, then $v$ is the center vertex of $K_c$ and thus $v\in V(M)$;
	if $S$ is a triangle component of $H$, then $K_c$ is a $5$-path and $v$ is an internal but not the middle vertex of $K_c$. 
\end{enumerate}
\end{lemma}
\begin{proof}
Suppose $K_c$ is a triangle of $H$. Then by the third statement in Lemma~\ref{lemma03}, 
$K_c$ is the center of a composite component $K$ of which each satellite element is a $5$-path.
The third statement in Lemma~\ref{lemma06} says that $(K)_m$ is an edge, and thus the $5$-path should be the center element, a contradiction.
Next, suppose $S$ is a satellite element of $K$, and suppose to the contrary $S$ is a $5$-path.
Then, $(K)_m$ is an edge by the third statement of Lemma~\ref{lemma06}, and thus $S$ should be the center element, again a contradiction.
This proves the first statement.

For the second statement, if $K_c$ is a star and $v$ is a satellite of $K_c$,
then the third statement in Lemma~\ref{lemma03} implies that $S$ is a $5$-path, and subsequently 
the third statement in Lemma~\ref{lemma06} implies that $(K)_m$ is an edge,
again leading to a contradiction that $S$ should be the center element.
If $S$ is a triangle of $H$, then by the third statement of Lemma~\ref{lemma03} $K_c$ is a $5$-path and 
$v$ is an internal but not the middle vertex of $K_c$. 
This proves the lemma.
\end{proof}

We define the following for the vertices of a center element $K_c$.

\begin{definition}
\label{def07}
A vertex $v$ of a center element $K_c$ is an {\em anchor} of $H+C$ if $K_c$ is a $5$-path or an edge, or $K_c$ is a star and $v$ is the center vertex of $K_c$.
The edge connecting $v$ to a satellite element $S$ in $C$ is called the {\em rescue-edge} for $S$ and $v$ is called the {\em supporting} anchor for $S$. 
For a nonnegative integer $j$, an anchor $v$ is a {\em $j$-anchor} if $v$ is the supporting anchor for exactly $j$ satellite elements of $H+C$.
\end{definition}

We note that, if $K_c$ is a star component of $H$, then the second statement in Lemma~\ref{lemma07} implies that 
each satellite of $K_c$ cannot be adjacent to any satellite element of $H+C$ and thus is excluded from the above definition of anchors.
Since $C$ is a path-cycle cover of $G_1$ obtained in Step 2.3, each satellite element $S$ of $H+C$ is adjacent to a unique anchor,
and each anchor is a $0$-, $1$-, or $2$-anchor.

\begin{notation}
\label{nota06}
For each component $K$ of $H+C$, let $s(K)$ denote the number of vertices in both $K$ and $V(M_C)$, i.e., $s(K) = |V(K)\cap V(M_C)|$.

If the center $K_c$ of $K$ is a $5$-path, then let $v_1$, \ldots, $v_5$ be the anchors of $K$ ordered from one endpoint to the other on $K_c$;
if $K_c$ is an edge, then let $v_1, v_2$ be the anchors of $K$;
otherwise, $K_c$ is a star and let $v_1$ be the unique anchor (which is the center vertex of $K_c$) of $K$
and let $v_2$ be the satellite vertex of $K_c$ such that $\{v_1, v_2\}\in M$.
\end{notation}

\begin{lemma}
\label{lemma08}
For each component $K$ of $H+C$, an $OPT(K)$ can be computed in $O(1)$ time.
\end{lemma}
\begin{proof}
Let ${\cal S}$ be the collection of the satellite elements of $K$.
By Lemma~\ref{lemma07}, the center element $K_c$ is a $5$-path, an edge or a star, and thus we distinguish three cases.

Firstly, if $K_c$ is a $5$-path, then by Lemma~\ref{lemma07} $S$ is a triangle, an edge or a star, for each $S\in {\cal S}$. 
If $S$ is a triangle or an edge, then $|V(S)|\le 3$.
If $S$ is a star and the center vertex of $S$ is incident with the rescue-edge for $S$,
then we can remove all but one satellite vertex of $S$ to keep $opt(K)$ unchanged, which leads to $|V(S)|\le 2$.
If one satellite vertex $v$ of $S$ is incident with the rescue-edge for $S$,
then we remove all satellites of $S$ except $v$ and one satellite not incident with the rescue-edge for $S$ to keep $opt(K)$ unchanged, which leads to $|V(S)|\le 3$.
In conclusion, $|V(S)|\le 3$ for each $S \in {\cal S}$.
Recall that each anchor is a $0$-, $1$-, or $2$-anchor.
It follows that $|{\cal S}|\le 10$ and after the vertex removal for each $S \in {\cal S}$,
$|V(K)| \le 10 \times 3 + 5 = 35$ and hence we can compute an $OPT(K)$ in constant time.

Secondly, if $K_c$ is an edge, then by Lemmas~\ref{lemma07} and \ref{lemma03},
each $S\in {\cal S}$ is an edge or a star and if $S$ is a star then the center vertex of $S$ must be in the rescue-edge of $S$.
Similarly as in the first case, if $S$ is an edge then $|V(S)| = 2$;
if $S$ is a star, then removing all but one satellite vertex of $S$ from $K$ does not decrease $opt(K)$, and so $|V(S)| = 2$ now.
Note that $|{\cal S}| \le 4$.
After the vertex removal for each $S$, we have $|V(K)| \le 4 \times 2 + 2 = 10$ and hence we can compute an $OPT(K)$ in constant time.

Lastly, $K_c$ is a star.
By Lemma~\ref{lemma07}, each $S \in {\cal S}$ must be adjacent to the center vertex of $K_c$, and thus $|{\cal S}| \le 2$.
Also, we can remove all but one satellite vertex of $K_c$ to keep $opt(K)$ unchanged, i.e., $|V(K_c)| = 2$.
Similarly to the second case above, for each $S \in {\cal S}$ we have $|V(S)| = 2$ after removing some vertices of $S$ if necessary.
It follows that $|V(K)|\le 2 \times 2 + 2 = 6$ and hence an $OPT(K)$ can be computed in constant time.
This completes the proof.
\end{proof}

Generally speaking, by computing an $OPT(K)$ for every $K$ of $H+C$ and outputting their union as an approximate solution for $G$,
we obtain an approximation algorithm for $MPC^{4+}_v$ achieving a ratio of $\frac 54 \max_{K} \frac {s(K)}{opt(K)}$ because of Lemma~\ref{lemma05}, 
unless $K$ is {\em critical} and {\em responsible}, to be defined later.
If $K$ is an isolated $5$-path, then by Lemma~\ref{lemma03} we have $\frac {s(K)}{opt(K)} = \frac 45$. 
But if $K$ is a composite component, $\frac {s(K)}{opt(K)}$ is not necessarily small (smaller than our target value which is about $1.4992$). 
This motivates the next definition of {\em critical component}.

\begin{definition}
\label{def08}
A {\em critical component} of $H+C$ is a component $K$ with $\frac {s(K)}{opt(K)} \ge \frac {14}{11}$. 
\end{definition}


\begin{notation}
\label{nota07}
Let $v$ be an anchor of $K$. 
\begin{itemize}
\parskip=0pt
\item If $v$ is a $0$-anchor, then let $Q_v$ be the vertex $v$;
	otherwise, $Q_v$ denotes the longest path among those paths in $K$ each starts with $v$ followed by an edge of $C$ incident to $v$. 
\item If $v$ is a $2$-anchor, then we use $P_v$ to denote the longest path among those paths in $K$ each contains $v$ and the two edges of $C$ incident to $v$. 
\end{itemize}
\end{notation}

\begin{remark}
\label{remark01}
When $v$ is a $2$-anchor, $Q_v$ can be a part of $P_v$.
By Lemma~\ref{lemma07}, each satellite element of $K$ is either a triangle, an edge or a star.
Therefore, if $v$ is not a $0$-anchor, then $Q_v$ is a $3$- or $4$-path; 
if $v$ is a $2$-anchor, then $P_v$ is a $5$-, $6$-, or $7$-path.
\end{remark}

\begin{lemma}
\label{lemma09}
Suppose $K_c$ is a $5$-path.
Then, the following statements hold:
\begin{enumerate}
\parskip=0pt
\item Suppose the total number of $1$- and $2$-anchors is $5$.
Then, $opt(K)\ge 17$ if one of $v_1,v_3,v_5$ is a $2$-anchor; otherwise, $opt(K)\ge 13$.

\item Suppose the total number of $1$- and $2$-anchors is $4$.
Then, $opt(K)\ge 15$ if one of $v_1,v_3,v_5$ is a $2$-anchor; otherwise, $opt(K)\ge 12$.
\end{enumerate}
\end{lemma}
\begin{proof}
Suppose the total number of $1$- and $2$-anchors is $5$.
So, there is no $0$-anchor in $K$.
The first case is that one of $v_1,v_3, v_5$ is a $2$-anchor.
We assume $v_1$ is a $2$-anchor and the case when $v_3$ or $v_5$ is a $2$-anchor can be discussed similarly.
Then, we can construct two vertex-disjoint $6^+$-paths by connecting $\{v_j,v_{j+1}\}$ and $Q_{v_j},Q_{v_{j+1}}$, $j=2,4$ and a $5^+$-path $P_{v_1}$.
So, $opt(K)\ge 2\times 6+5=17$.
We can assume $v_1, v_3, v_5$ are not $2$-anchors. 
Then, we can construct a $6^+$-path by using $\{v_1,v_2\}$ to connect $Q_{v_1},Q_{v_2}$
and a $7^+$-path by connecting $Q_{v_3}, Q_{v_5}$ with the $3$-path $v_3$-$v_4$-$v_5$.
So, $opt(K)\ge 13$.

Suppose the total number of $1$- and $2$-anchors is $4$ and one of $v_1, v_3, v_5$ is a $2$-anchor.
We assume $v_1$ is a $2$-anchor and the case when $v_3$ or $v_5$ is a $2$-anchor can be discussed similarly. 
Then, we can construct a $6^+$-path and a $4^+$-path by connecting $Q_{v_j},Q_{v_{j+1}}$ with $\{v_j,v_{j+1}\}$, $j=2,4$ and a $5^+$-path $P_{v_1}$.
So, $opt(K)\ge 15$.

Now, we can assume $v_1, v_3, v_5$ are not $2$-anchors.
Recall that $K$ has exactly one $0$-anchor.
Suppose one of $v_2, v_4$ is a $0$-anchor.
Without loss of generality, we assume $v_2$ is a $0$-anchor.
Then, we can construct a $6^+$-path by using $\{v_4,v_5\}$ to connect $Q_{v_4},Q_{v_5}$ 
and a $7^+$-path by using the $3$-path $v_1$-$v_2$-$v_3$ to connect $Q_{v_1},Q_{v_3}$.
So, $opt(K)\ge 13$.
Then, we can assume one of $v_1, v_3, v_5$ is the unique $0$-anchor in $K$.
We can only discuss the case when $v_1$ is the unique $0$-anchor 
and we can analyze the case when one of $v_3$ or $v_5$ is the $0$-anchor similarly.
Then, we can construct two vertex-disjoint $6^+$-path by using $\{v_j, v_{j+1}\}$ to connect $Q_{v_j},Q_{v_{j+1}}$, $j=2,4$.
So, in this case, $opt(K)\ge 12$, which completes the proof.
\end{proof}

\begin{lemma}
\label{lemma10} 
Suppose that $K$ has no $2$-anchor. 
Then, $\frac {s(K)}{opt(K)} < \frac {14}{11}$ and hence $K$ is not critical.
\end{lemma}
\begin{proof}
First, consider the case where $K_c$ is an edge. 
Then, either both $v_1$ and $v_2$ are 1-anchors of $K$, or exactly one of $v_1$ and $v_2$ is a 1-anchor of $K$. 
In the former case, $Q_{v_1}$ and $Q_{v_2}$ are connected with $\{v_1,v_2\}$ into a $6^+$-path in $K$.
In the latter case, without loss of generality, we assume $v_1$ is a $1$-anchor.
So, $Q_{v_1}$ can be extended to a $4^+$-path with $\{v_1,v_2\}$. 
In conclusion, in either case, $\frac{s(K)}{opt(K)} \le 1$ and hence $K$ is not critical. 

If $K_c$ is a star, then $v_1$ is the unique $1$-anchor where $v_1$ is the center vertex of $K_c$.
So, $s(K)=4$ and we can construct a $4^+$-path by connecting $Q_{v_1}$ with $\{v_1,v_2\}$ where $v_2$ is the vertex in $V(M)$.
It follows that $opt(K)\ge 4$ and thus $\frac {s(K)}{opt(K)}\le 1$ and $K$ is not critical.

We next consider the case where $K_c$ is a $5$-path. 
Then, $opt(K) \ge 5$ because of the $5$-path. 
So, we may assume that $s(K) > 6$ because otherwise $\frac{s(K)}{opt(K)} \le \frac{6}{5} < \frac {14}{11}$ and we are done. 
Since $K$ has no $2$-anchor, $s(K) \le 14$. 
Thus, $8\le s(K)\le 14$ because $s(K)$ is even. 
If $K$ has at least four $1$-anchors, then by Lemma~\ref{lemma09}, $opt(K) \ge 12$,
implying that $\frac{s(K)}{opt(K)} \le \frac{14}{12}< \frac {14}{11}$ and we are done. 
Hence, we may assume that $K$ has at most three $1$-anchors, i.e., $s(K) \le 10$. 
It remains to distinguish two cases as follows.

\medskip
{\em Case 1:} $s(K)=8$. 
In this case, $K$ has exactly two $1$-anchors. 
Either at least one endpoint of $K_c$ is a $1$-anchor, or at least two internal vertices of $K_c$ are $1$-anchors. 
In the former case, without loss of generality, we assume $v_1$ is a $1$-anchor.
Then, $opt(K) \ge 7$ because $K_c$ and $Q_{v_1}$ are connected into a $7^+$-path in $K$. 
In the latter case, $opt(K)\ge 8$ because we can construct two vertex-disjoint $4^+$-paths in $K$, 
each of which is obtained by choosing two $1$-anchors $v_j, v_k$ of $K$ and connecting $Q_{v_j}, Q_{v_k}$ with a subpath of $K_c$, respectively. 
So, in both cases, $\frac{s(K)}{opt(K)} \le \frac{8}{7}< \frac {14}{11}$. 

\medskip
{\em Case 2:} $s(K)=10$. 
In this case, $K$ has exactly three $1$-anchors. 
Either both endpoints of $K_c$ are $1$-anchors, or at least two internal vertices of $K$ are $1$-anchors. 
In the former case, $opt(K) \ge 9$ 
because $K_c$, $Q_{v_1}$, and $Q_{v_5}$ are connected into a $9^+$-path in $K$.
In the latter case, similarly to Case 1, $opt(K)\ge 8$.
So, in both cases, $\frac{s(K)}{opt(K)} \le \frac{10}{8}< \frac {14}{11}$. 
\end{proof}

\begin{lemma}
\label{lemma11}
Suppose that $K$ has at least three $2$-anchors. 
Then, the following statements hold and hence $K$ is not critical.
\begin{enumerate}
\parskip=0pt
\item If $K$ has five $2$-anchors, then $s(K) = 24$ and $opt(K)\ge 25$.
\item If $K$ has four $2$-anchors, then $s(K) \le 22$ and $opt(K)\ge 20$.
\item If $K$ has three $2$-anchors, then either $s(K)=20$ and $opt(K) \ge 17$, 
or $s(K)\le 18$ and $opt(K) \ge {15}$.
\end{enumerate}
\end{lemma}
\begin{proof}
Since $K$ has at least three $2$-anchors, $K_c$ is a $5$-path. 
Obviously, $s(K)$ is even and $s(K) \le 24$.
If $K$ has at least four $2$-anchors, then because of the vertex-disjoint $5^+$-paths $P_{v_j}$ for each $2$-anchor $v_j$, the first two statements in the lemma hold. 
Thus, we may assume that $K$ has exactly three $2$-anchors. 
Then, $opt(K) \ge 3\times 5 = 15$ and $s(K) \le 20$. 

If $s(K) \le 18$, then $\frac {s(K)}{opt(K)} \le \frac {18}{15}$ and we are done. 
So, we may assume $s(K) = 20$.
Then, three vertices of $K_c$ are $2$-anchors and the other two vertices of $K_c$ are 1-anchors. 
It follows that at least one of $v_1, v_3, v_5$ is a $2$-anchor.
By the first statement of Lemma~\ref{lemma09}, $opt(K) \ge 17$ and $K$ is not critical, which completes the proof.
\end{proof}

\begin{lemma}
\label{lemma12}
Suppose that $K$ has exactly two $2$-anchors. Then, $K_c$ is a $5$-path or an edge and the following statements hold:
\begin{enumerate}
\parskip=0pt
\item If $K_c$ is an edge, then $s(K) = 10$ and $opt(K)\ge 10$, and hence $K$ is not critical.

\item If $K_c$ is a $5$-path and $s(K) \le 12$, then $opt(K) \ge 10$ 
and hence $K$ is not critical.

\item Suppose that $K_c$ is a $5$-path and $s(K) = 18$.
 If $K$ is critical, then $opt(K) \in \{13, 14\}$ 
and the two non-middle internal vertices of $K_c$ are 2-anchors in $K$; 
otherwise, $opt(K) \ge 15$.

\item Suppose that $K_c$ is a $5$-path and $s(K) = 16$. 
If $K$ is critical, then $opt(K) = 12$ 
and the two non-middle internal vertices of $K_c$ are 2-anchors in $K$; 
otherwise, $opt(K) \ge 13$.

\item
Suppose that $K_c$ is a $5$-path and $s(K) = 14$. 
If $K$ is critical, then $opt(K) = 11$ 
and $K$ has only five different structures as shown in Figure~\ref{fig01}; 
otherwise, $opt(K) \ge 12$. 
\end{enumerate}
\end{lemma}
\begin{proof}
Note that if $K_c$ is a star, then it has exactly one anchor.
So, by the first statement of Lemma~\ref{lemma07}, $K_c$ is a $5$-path or an edge.
If $K_c$ is an edge, then clearly both $v_1,v_2$ of $K_c$ are $2$-anchors and 
hence $opt(K) \ge 2\times 5=10$ because of the two vertex-disjoint $5^+$-paths $P_{v_1}$ and $P_{v_2}$.

Suppose $K_c$ is a $5$-path. 
Since $K$ has exactly two $2$-anchors, $opt(K) \ge 2\times 5=10$ and $s(K) \le 18$. 
Now, if $s(K) \le 12$, then we are done. 
Thus, we can further assume that $s(K) > 12$ and in turn $s(K) \ge 14$ because $s(K)$ is even. 
Therefore, $s(K) \in \{14, 16, 18\}$. 

\medskip
{\em Case 1:} $s(K) = 18$. 
In this case, $K$ has exactly seven satellite elements in total. 
Since $K$ has only two $2$-anchors, each vertex of $K_c$ is a $1$- or $2$-anchor in $K$. 
So, the total number of $1$- and $2$-anchors is $5$.
By Lemma~\ref{lemma09}, if one of $v_1,v_3,v_5$ is a $2$-anchor, then $opt(K)\ge 17$ and $\frac {s(K)}{opt(K)}\le \frac {18}{17}$.
Hence $K$ is not critical and then we can assume none of $v_1,v_3,v_5$ is a $2$-anchor.
That is, $v_2,v_4$ are both $2$-anchors and $v_1,v_2,v_3$ are both $1$-anchors.
By Lemma~\ref{lemma09} again, $opt(K)\ge 13$.
Obviously, if $opt(K) \ge 15$, then $\frac {s(K)}{opt(K)} \le \frac {18}{15} < \frac{14}{11}$ and hence $K$ is not critical. 
Otherwise,  $opt(K) = 13$ or $14$, and $\frac {s(K)}{opt(K)} = \frac {18}{13}$ or $\frac {18}{14}$. 

\medskip
{\em Case 2:} $s(K) = 16$. 
In this case, $K$ has exactly six satellite elements in total. 
Since $K$ has exactly two $2$-anchors, the total number of $1$- and $2$-anchors is $4$.
By Lemma~\ref{lemma09}, if one of $v_1,v_3,v_5$ is a $2$-anchor, then $opt(K)\ge 15$ and $K$ is not critical.
Similarly to Case~1, we can assume $v_2,v_4$ are both $2$-anchors and two of $v_1,v_2,v_3$ are $1$-anchors.
By Lemma~\ref{lemma09} again, $opt(K)\ge 12$.
Obviously, if $opt(K) \ge 13$, then $\frac {s(K)}{opt(K)} \le \frac {16}{13} < \frac{14}{11}$ and hence $K$ is not critical. 
Otherwise,  $opt(K) = 12$ and $\frac {s(K)}{opt(K)} = \frac {16}{12}> \frac {14}{11}$.  

\medskip
{\em Case 3:} $s(K) = 14$. 
In this case, $K$ has exactly five satellite elements in total. 
If both $v_1$ and $v_2$ are $2$-anchors in $K$, then $K$ contains three vertex-disjoint $4^+$-paths 
(namely, $P_{v_1}$, $P_{v_2}$, and a $4^+$-path obtained by connecting an edge of $K_c$ to $Q_{v_i}$ for some $i\in\{3,4,5\}$ where $v_i$ is a $1$-anchor).
Hence $opt(K) \ge 2\times 5+4 = 14$, implying that $K$ is not critical. 
Similarly, if both $v_4,v_5$ or $v_1,v_5$ are $2$-anchors in $K$, then $opt(K) \ge 14$, implying that $K$ is not critical. 
So, it remains to consider the following three subcases.

\medskip
{\em Case 3.1:} $v_3$ is a $2$-anchor in $K$. 
In this case, either some $v_i\in\{v_1,v_2\}$ or some $v_i\in\{v_4,v_5\}$ is the other $2$-anchor of $K$. 
Without loss of generality, we assume the former case. 
If some $v_j\in\{v_4,v_5\}$ is a $1$-anchor in $K$, then $K$ contains three vertex-disjoint $4^+$-paths 
(namely, $P_{v_3}$, $P_{v_i}$, and a $4^+$-path obtained by connecting the edge $\{v_4,v_5\}$ to $P_{v_j}$), 
and hence $opt(K) \ge 2\times 5+4 = 12$, implying that $K$ is not critical. 
Thus, we may assume that one of $v_1$ and $v_2$ is a $2$-anchor and the other is a $1$-anchor in $K$. 
Then, $K$ contains a $6^+$-path obtained by using the edge $\{v_1,v_2\}$ to connect $Q_{v_1}$ and $Q_{v_2}$. 
Since this $6^+$-path and $P_{v_3}$ are vertex-disjoint, $opt(K)\ge 5+6 = 11$ and hence $\frac {s(K)}{opt(K)} \le \frac {14}{11}$. 
If $opt(K)\ge 12$, then $\frac {s(K)}{opt(K)} \le \frac {14}{12} < \frac {14}{11}$ and therefore $K$ is not critical. 
Otherwise, $opt(K)=11$ and $K$ is critical. 

\medskip
{\em Case 3.2:} Both $v_2$ and $v_4$ are $2$-anchors in $K$. 
In this case, $v_1$, $v_3$, or $v_5$ is a $1$-anchor in $K$. 
We assume that $v_1$ is a $1$-anchor in $K$; the other two cases can be similarly discussed. 
Then, besides $P_{v_4}$, $K$ contains a $6^+$-path obtained by using the edge $\{v_1,v_2\}$ to connect $Q_{v_1}$ and $Q_{v_2}$. 
So, $opt(K)\ge 5+6=11$. If $opt(K) \ge 12$, then $\frac {s(K)}{opt(K)} \le \frac {14}{12}$ and 
$K$ is not critical. 
Thus, $K$ can be critical only when $opt(K) = 11$. 

\medskip
{\em Case 3.3:} 
Either both $v_1$ and $v_4$ are $2$-anchors in $K$, or both $v_2$ and $v_5$ are $2$-anchors in $K$. 
By symmetry, we may assume the former case. 
If $v_2$ is a $1$-anchor in $K$, then besides $P_{v_1}$ and $P_{v_4}$, $K$ contains a $4^+$-path obtained by using the path $\{v_2,v_3\}$ to connect $Q_{v_2}$, 
implying that $opt(K) \ge 2\times 5+4 = 14$ and hence $K$ is not critical. 
Similarly, if $v_3$ is a $1$-anchor in $K$, then $opt(K) \ge 14$ and hence $K$ is not critical. 
Hence, we may assume that $v_5$ is a $1$-anchor in $K$. 
Then, similarly to Case~3.2, $K$ can be critical only when $opt(K) = 11$. 
\end{proof}

\begin{lemma}
\label{lemma13}
Suppose that $K$ has exactly one $2$-anchor.  
Then, the following statements hold:
\begin{enumerate}
\parskip=0pt
\item
Suppose $K_c$ is an edge. 
If $K$ is critical, then $s(K)=8$ and $opt(K) = 6$; 
otherwise, either $s(K)=6$ and $opt(K)\ge 5$, or $s(K)=8$ and $opt(K) \ge 7$. 

\item 
Suppose $K_c$ is a star. 
Then, $s(K)=6$, $opt(K)\ge 5$ and hence $K$ is not critical.

\item Suppose $K_c$ is a $5$-path and $s(K) \ge 12$.
Then, $s(K)=12$ and $opt(K) \ge 10$; 
$s(K)=14$ and $opt(K) \ge 12$; 
$s(K)=16$ and $opt(K) \ge 13$. 
Hence $K$ is not critical. 

\item Suppose that $K_c$ is a $5$-path and $s(K) = 10$. 
If $K$ is not critical, then $opt(K) \ge 8$; 
otherwise, $opt(K) = 7$ and there is a pair $(v_i,v_j)\in\{(v_2,v_1), (v_4,v_5)\}$ 
such that $v_i$ is a $2$-anchor and $v_j$ is a $1$-anchor in $K$. 

\item Suppose that $K_c$ is a $5$-path and $s(K) = 8$. 
If $K$ is not critical, then $opt(K) \ge 7$; otherwise, either $opt(K)=5$ or $opt(K)=6$. 
Moreover, if $opt(K) = 6$, then one of the non-middle internal vertex of $K$ is a 2-anchor in $K$;
if $opt(K) = 5$, then the middle vertex of $K_c$ is a 2-anchor in $K$. 
\end{enumerate}
\end{lemma}
\begin{proof}
We first prove the first statement.
Without loss of generality, let $v_1$ be the unique $2$-anchor in $K$. 
If $K$ contains no $1$-anchor, then $s(K)=6$ and $opt(K)\ge 5$ because of $P_{v_1}$, implying that $K$ is not critical. 
So, we may assume that $v_2$ is a $1$-anchor.
Thus, $s(K)=8$, and $opt(K)\ge 6$ because we can obtain a $6^+$-path by using the edge $\{v_1,v_2\}$ to connect $Q_{v_1}$ and $Q_{v_2}$. 
If $opt(K) \ge 7$, $\frac {s(K)}{opt(K)}\le\frac 87$ and hence $K$ is not critical. 
Otherwise, $opt(K) = 6$, and $K$ is critical and has the first structure in Figure~\ref{fig01}.  

If $K_c$ is a star, then $v_1$ is a $2$-anchor and $s(K)=6$.
We can obtain a $5^+$-path by choosing $P_{v_1}$.
So, $opt(K)\ge 5$ and $K$ is not critical.

To prove the other statements, we assume that $K_c$ is a $5$-path. 
Since $K$ has exactly one $2$-anchors and $s(K)$ is even, $s(K) \in \{8, 10, 12, 14, 16\}$.  

\medskip
{\em Case 1:} $s(K) = 16$. 
In this case, one vertex of $K_c$ is a $2$-anchor in $K$ and each other vertex of $K_c$ is a $1$-anchor in $K$. 
So, by the first statement of Lemma~\ref{lemma09}, $opt(K)\ge 13$ and $\frac{s(K)}{opt(K)} \le \frac{16}{13}$. 
Therefore, $K$ is not critical. 

\medskip
{\em Case 2:} $s(K) = 14$. 
In this case, $K$ has exactly three $1$-anchors and one $2$-anchor.
So, by the second statement of Lemma~\ref{lemma09}, $opt(K)\ge 12$ and $K$ is not critical.

\medskip
{\em Case 3:} $s(K) = 12$. In this case, we distinguish three subcases as follows.

\medskip
{\em Case 3.1:} $v_3$ is a $2$-anchor in $K$. 
In this case, if both $v_i$ and $v_{i+1}$ are $1$-anchors in $K$ for some $i\in\{1,4\}$, 
then besides $P_{v_3}$, $K$ contains a $6^+$-path obtained by using the edge $\{v_i,v_{i+1}\}$ to connect $Q_{v_i}$ and $Q_{v_{i+1}}$.
So, $opt(K) \ge 5+6=11$ and $K$ is not critical. 
Otherwise, exactly one of $v_1$ and $v_2$ is a $1$-anchor in $K$ and so is exactly one of $v_4$ and $v_5$, 
and besides $P_{v_3}$, $K$ contains two $4^+$-paths obtained by connecting the edge $\{v_j,v_{j+1}\}$ to $Q_{v_j}$ or $Q_{v_{j+1}}$ for each $j\in\{1,4\}$.
Therefore $opt(K) \ge 2\times 4 + 5 = 13$ and $K$ is not critical. 

\medskip
{\em Case 3.2:} $v_1$ or $v_5$ is a $2$-anchor in $K$. 
By symmetry, we may assume $v_1$ is a $2$-anchor in $K$. 
Then, two vertices $v_i$ and $v_j$ in $\{v_2,\ldots,v_5\}$ are $1$-anchors in $K$. 
So, besides $P_{v_1}$, $K$ has a $6^+$-path (vertex-disjoint from $P_{v_1}$) obtained by using the subpath of $K_c$ between $v_i$ and $v_j$ to connect $Q_{v_i}$ and $Q_{v_j}$. 
Thus, $opt(K) \ge 5+6=11$ and hence $K$ is not critical. 

\medskip
{\em Case 3.3:} $v_2$ or $v_4$ is a $2$-anchor in $K$. 
By symmetry, we may assume $v_2$ is a 2-anchor in $K$. 
If two vertices in $\{v_3, v_4, v_5\}$ are $1$-anchors in $K$, then as in Case~3.2, $opt(K) \ge 11$. 
So, we may assume that $v_1$ is a $1$-anchor in $K$. 
Then, $K$ has a $6^+$-path obtained by using the edge $\{v_1,v_2\}$ to connect $Q_{v_1}$ and $Q_{v_2}$. 
Moreover, $K$ has a $4^+$-path obtained by connecting the edge $\{v_i,v_{i+1}\}$ to $Q_{v_i}$ or $Q_{v_{i+1}}$ for some $i\in\{3,4\}$
since one of $v_3, v_4, v_5$ is a $1$-anchor.
Thus, $opt(K)\ge 6 + 4 = 10$ and $K$ is not critical.

\medskip
{\em Case 4:} $s(K) = 10$. 
In this case, $K$ has exactly one $2$-anchor and exactly one $1$-anchor. 

\medskip
{\em Case 4.1:} $v_i$ is a 2-anchor in $K$ for some $i\in\{1,3,5\}$. 
In this case, besides $P_{v_i}$, $K$ has a $4^+$-path (vertex-disjoint from $P_{v_i}$) obtained by
connecting the edge $\{v_j,v_{j+1}\}$ to $Q_{v_j}$ or $Q_{v_{j+1}}$ for some $j\in\{1,\ldots,5\}$. 
So, $opt(K)\ge 5+4 = 9$ and $K$ is not critical. 

\medskip
{\em Case 4.2:} $v_i$ is a $2$-anchor in $K$ for some $i\in\{2,4\}$. 
By symmetry, we may assume $v_2$ is the $2$-anchor. 
If some $v_j\in\{v_3,v_4,v_5\}$ is a $1$-anchor in $K$, then besides $P_{v_2}$, $K$ has a $4^+$-path 
(vertex-disjoint from $P_{v_i}$) obtained by connecting the edge $\{v_j,v_{j+1}\}$ to $Q_{v_j}$ or 
$Q_{v_{j+1}}$ for some $j\in\{3,4\}$; 
so, $opt(K)\ge 5 + 4 = 9$ and $K$ is not critical. 
Thus, we may further assume that $v_1$ is a $1$-anchor in $K$. 
Then, $K$ has a $7^+$-path obtained by connecting $K_c$ to $Q_{v_1}$, and hence $opt(K) \ge 7$. 
If $opt(K) \ge 8$, then $K$ is not critical. 
Otherwise, $opt(K) = 7$ and $K$ becomes critical. 

\medskip
{\em Case 5:} $s(K) = 8$. 
In this case, one vertex of $K$ is a $2$-anchor in $K$ and the other vertices of $K$ are $0$-anchors. 
If $v_i$ is a $2$-anchor in $K$ for some $i\in\{1,5\}$, then besides $P_{v_i}$, $K$ has a $4^+$-path 
(indeed, a subpath of $K_c$) disjoint from $P_{v_i}$; 
so, $opt(K)\ge 5+4 = 9$ and $K$ is not critical. 
Otherwise, either $v_3$ or some $v_i\in\{v_2,v_4\}$ is a $2$-anchor in $K$. 
In the former case, $opt(K) \ge 5$ by choosing $K_c$ and hence $\frac{s(K)}{opt(K)} \le \frac{8}{5}$.
If $opt(K)\ge 7$, then $\frac{s(K)}{opt(K)} \le \frac{8}{7} < \frac{14}{11}$ and $K$ is not critical. 
By the third statement of Lemma~\ref{lemma03}, for each satellite element of $v_3$, 
if it is a star, then the center vertex must be in its rescue-edge.
So, $P_{v_3}$ is a $5$-path and hence if $K$ is critical, then $opt(K)=5$.

In the latter case, $K$ has a $6^+$-path obtained by connecting $Q_{v_i}$ to a subpath of $K_c$ with four vertices, 
and hence $opt(K) \ge 6$ and $\frac{s(K)}{opt(K)} \le \frac 86$. 
Obviously, if $opt(K)\ge 7$, then $\frac{s(K)}{opt(K)} \le \frac 87 < \frac {14}{11}$ and $K$ is not critical. 
In summary, if $K$ is critical, then either $opt(K)=5$ or $opt(K)=6$. 

By the discussion for the case when $s(K)=8$, $v_3$ is a $2$-anchor in $K$ if $opt(K) = 5$, 
while either $v_2$ or $v_4$ is a $2$-anchor in $K$ if $opt(K) = 6$.
Then, the lemma follows.
\end{proof}

Now, Lemmas~\ref{lemma10}--\ref{lemma13} imply that a critical component $K$ of $H+C$ has one $2$-anchor or two $2$-anchors.
Furthermore, if $K$ has one $2$-anchor, then by Lemma~\ref{lemma13},
$K_c$ is an edge or a $5$-path and the possible structures for $K$ are shown in the first {\color{black}row} of Figure~\ref{fig01}.
Otherwise, $K$ has two $2$-anchors and $K_c$ is a $5$-path.
By Lemma~\ref{lemma12}, $s(K)\in \{14, 16, 18\}$ and Figure~\ref{fig01} except the first row shows all possible structures for $K$.

\begin{remark}
\label{remark02}
Even if a satellite element $S$ of $K$ can be a star or triangle, 
we almost always draw only one edge of $S$ in Figure \ref{fig01} for simplicity. 
We will keep this convention in the subsequent figures.
\end{remark}

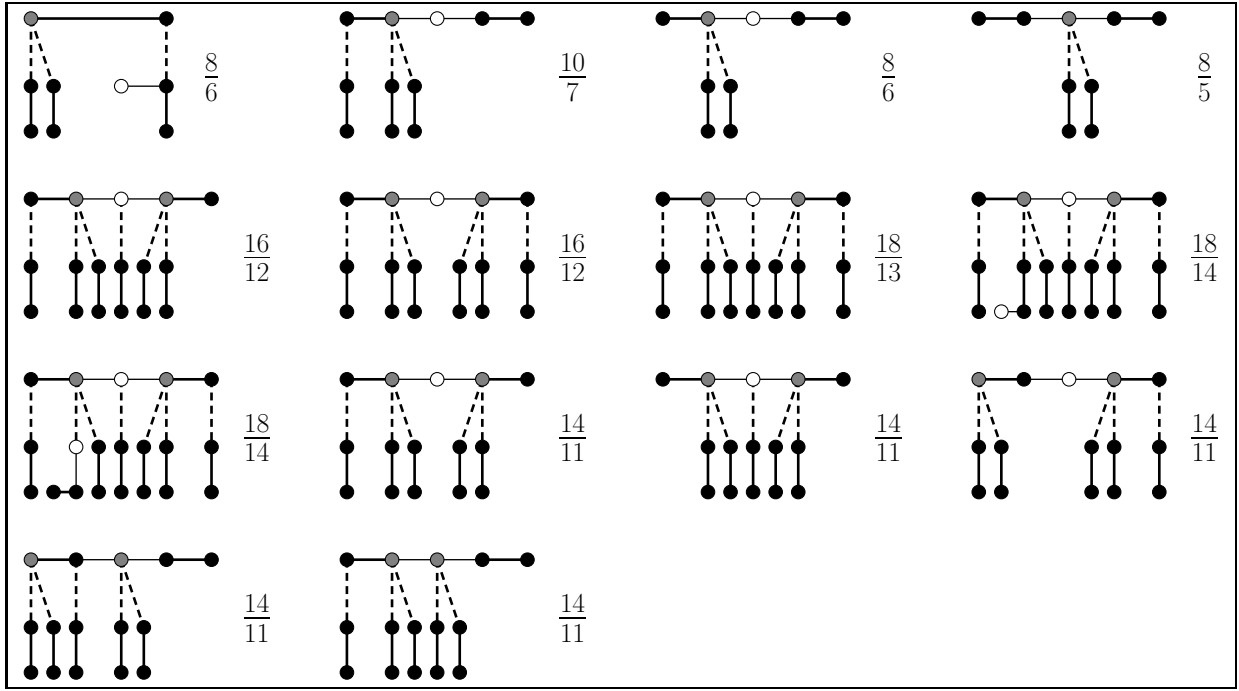
\begin{figure}[thb]
\begin{center}
\framebox{
\begin{minipage}{6.3in}
\begin{tikzpicture}[scale=0.6,transform shape]

\draw [thick, line width = 1pt] (-8, 2) -- (-5, 2);
\draw [densely dashed, line width = 1pt] (-8, 2) -- (-8, 0.5);
\draw [thick, line width = 1pt] (-8, 0.5) -- (-8, -0.5);
\draw [densely dashed, line width = 1pt] (-8, 2) -- (-7.5, 0.5);
\draw [thick, line width = 1pt] (-7.5, 0.5) -- (-7.5, -0.5);
\draw [densely dashed, line width = 1pt] (-5, 2) -- (-5, 0.5);
\draw [thick, line width = 1pt] (-5, 0.5) -- (-5, -0.5);
\draw [thin, line width = 0.5pt] (-5, 0.5) -- (-6, 0.5);
\filldraw[fill = gray] (-8, 2) circle(.15);
\filldraw (-5,2) circle(.15);
\filldraw (-8,0.5) circle(.15);
\filldraw (-8,-0.5) circle(.15);
\filldraw (-7.5,0.5) circle(.15);
\filldraw (-7.5,-0.5) circle(.15);
\filldraw (-5,0.5) circle(.15);
\filldraw (-5,-0.5) circle(.15);
\filldraw[fill = white] (-6, 0.5) circle(.15);
\node[font=\fontsize{27}{6}\selectfont] at (-4,0.7) {$\frac 86$};

\draw [thick, line width = 1pt] (-1, 2) -- (0, 2);
\draw [thick, line width = 1pt] (2, 2) -- (3, 2);
\draw [thin, line width = 0.5pt] (0, 2) -- (1, 2);
\draw [thin, line width = 0.5pt] (1, 2) -- (2, 2);
\draw [densely dashed, line width = 1pt] (-1, 2) -- (-1,0.5);
\draw [thick, line width = 1pt] (-1, 0.5) -- (-1, -0.5);
\draw [densely dashed, line width = 1pt] (0, 2) -- (0, 0.5);
\draw [thick, line width = 1pt] (0, 0.5) -- (0, -0.5);
\draw [densely dashed, line width = 1pt] (0, 2) -- (0.5,0.5);
\draw [thick, line width = 1pt] (0.5, 0.5) -- (0.5, -0.5);
\filldraw (-1, 2) circle(.15);
\filldraw[fill = gray] (0, 2) circle(.15);
\filldraw[fill = white] (1, 2) circle(.15);
\filldraw (2, 2) circle(.15);
\filldraw (3, 2) circle(.15);
\filldraw (-1,0.5) circle(.15);
\filldraw (-1,-0.5) circle(.15);
\filldraw (0,0.5) circle(.15);
\filldraw (0,-0.5) circle(.15);
\filldraw (0.5,0.5) circle(.15);
\filldraw (0.5,-0.5) circle(.15);
\node[font=\fontsize{25}{6}\selectfont] at (4,0.7) {$\frac {10}{7}$};

\draw [thick, line width = 1pt] (6, 2) -- (7, 2);
\draw [thick, line width = 1pt] (9, 2) -- (10, 2);
\draw [thin, line width = 0.5pt] (7, 2) -- (8, 2);
\draw [thin, line width = 0.5pt] (8, 2) -- (9, 2);
\draw [densely dashed, line width = 1pt] (7, 2) -- (7, 0.5);
\draw [thick, line width = 1pt] (7, 0.5) -- (7, -0.5);
\draw [densely dashed, line width = 1pt] (7, 2) -- (7.5, 0.5);
\draw [thick, line width = 1pt] (7.5, 0.5) -- (7.5, -0.5);
\filldraw (6, 2) circle(.15);
\filldraw[fill = gray] (7, 2) circle(.15);
\filldraw[fill = white] (8, 2) circle(.15);
\filldraw (9, 2) circle(.15);
\filldraw (10, 2) circle(.15);
\filldraw (7, 0.5) circle(.15);
\filldraw (7, -0.5) circle(.15);
\filldraw (7.5, 0.5) circle(.15);
\filldraw (7.5, -0.5) circle(.15);
\node[font=\fontsize{25}{6}\selectfont] at (11,0.7) {$\frac 86$};

\draw [thick, line width = 1pt] (13, 2) -- (14, 2);
\draw [thick, line width = 1pt] (16, 2) -- (17, 2);
\draw [thin, line width = 0.5pt] (14, 2) -- (15, 2);
\draw [thin, line width = 0.5pt] (15, 2) -- (16, 2);
\draw [densely dashed, line width = 1pt] (15, 2) -- (15, 0.5);
\draw [thick, line width = 1pt] (15, 0.5) -- (15, -0.5);
\draw [densely dashed, line width = 1pt] (15, 2) -- (15.5, 0.5);
\draw [thick, line width = 1pt] (15.5, 0.5) -- (15.5, -0.5);
\filldraw (13, 2) circle(.15);
\filldraw (14, 2) circle(.15);
\filldraw[fill = gray] (15, 2) circle(.15);
\filldraw (16, 2) circle(.15);
\filldraw (17, 2) circle(.15);
\filldraw (15, 0.5) circle(.15);
\filldraw (15, -0.5) circle(.15);
\filldraw (15.5, 0.5) circle(.15);
\filldraw (15.5, -0.5) circle(.15);
\node[font=\fontsize{25}{6}\selectfont] at (18,0.7) {$\frac 85$};


\draw [thick, line width = 1pt] (-8, -2) -- (-7, -2);
\draw [thick, line width = 1pt] (-5, -2) -- (-4, -2);
\draw [thin, line width = 0.5pt] (-7, -2) -- (-6, -2);
\draw [thin, line width = 0.5pt] (-6, -2) -- (-5, -2);
\draw [densely dashed, line width = 1pt] (-8, -2) -- (-8,-3.5);
\draw [thick, line width = 1pt] (-8, -3.5) -- (-8, -4.5);
\draw [densely dashed, line width = 1pt] (-7, -2) -- (-7, -3.5);
\draw [thick, line width = 1pt] (-7, -3.5) -- (-7, -4.5);
\draw [densely dashed, line width = 1pt] (-7, -2) -- (-6.5,-3.5);
\draw [thick, line width = 1pt] (-6.5, -3.5) -- (-6.5, -4.5);
\draw [densely dashed, line width = 1pt] (-6, -2) -- (-6,-3.5);
\draw [thick, line width = 1pt] (-6, -3.5) -- (-6, -4.5);
\draw [densely dashed, line width = 1pt] (-5, -2) -- (-5.5,-3.5);
\draw [thick, line width = 1pt] (-5.5, -3.5) -- (-5.5, -4.5);
\draw [densely dashed, line width = 1pt] (-5, -2) -- (-5,-3.5);
\draw [thick, line width = 1pt] (-5, -3.5) -- (-5, -4.5);
\filldraw (-8, -2) circle(.15);
\filldraw[fill = gray] (-7, -2) circle(.15);
\filldraw[fill = white] (-6, -2) circle(.15);
\filldraw[fill = gray] (-5, -2) circle(.15);
\filldraw (-4, -2) circle(.15);
\filldraw (-8,-3.5) circle(.15);
\filldraw (-8,-4.5) circle(.15);
\filldraw (-7,-3.5) circle(.15);
\filldraw (-7,-4.5) circle(.15);
\filldraw (-6.5,-3.5) circle(.15);
\filldraw (-6.5,-4.5) circle(.15);
\filldraw (-5.5,-3.5) circle(.15);
\filldraw (-5.5,-4.5) circle(.15);
\filldraw (-5,-3.5) circle(.15);
\filldraw (-5,-4.5) circle(.15);
\filldraw (-6,-3.5) circle(.15);
\filldraw (-6,-4.5) circle(.15);
\node[font=\fontsize{27}{6}\selectfont] at (-3,-3.3) {$\frac {16}{12}$};

\draw [thick, line width = 1pt] (-1, -2) -- (0, -2);
\draw [thick, line width = 1pt] (2, -2) -- (3, -2);
\draw [thin, line width = 0.5pt] (0, -2) -- (1, -2);
\draw [thin, line width = 0.5pt] (1, -2) -- (2, -2);
\draw [densely dashed, line width = 1pt] (-1, -2) -- (-1,-3.5);
\draw [thick, line width = 1pt] (-1, -3.5) -- (-1, -4.5);
\draw [densely dashed, line width = 1pt] (0, -2) -- (0, -3.5);
\draw [thick, line width = 1pt] (0, -3.5) -- (0, -4.5);
\draw [densely dashed, line width = 1pt] (0, -2) -- (0.5,-3.5);
\draw [thick, line width = 1pt] (0.5, -3.5) -- (0.5, -4.5);
\draw [densely dashed, line width = 1pt] (2, -2) -- (1.5,-3.5);
\draw [thick, line width = 1pt] (1.5, -3.5) -- (1.5, -4.5);
\draw [densely dashed, line width = 1pt] (2, -2) -- (2,-3.5);
\draw [thick, line width = 1pt] (2, -3.5) -- (2, -4.5);
\draw [densely dashed, line width = 1pt] (3, -2) -- (3,-3.5);
\draw [thick, line width = 1pt] (3, -3.5) -- (3, -4.5);
\filldraw (-1, -2) circle(.15);
\filldraw[fill = gray] (0, -2) circle(.15);
\filldraw[fill = white] (1, -2) circle(.15);
\filldraw[fill = gray] (2, -2) circle(.15);
\filldraw (3, -2) circle(.15);
\filldraw (-1,-3.5) circle(.15);
\filldraw (-1,-4.5) circle(.15);
\filldraw (0,-3.5) circle(.15);
\filldraw (0,-4.5) circle(.15);
\filldraw (0.5,-3.5) circle(.15);
\filldraw (0.5,-4.5) circle(.15);
\filldraw (1.5,-3.5) circle(.15);
\filldraw (1.5,-4.5) circle(.15);
\filldraw (2,-3.5) circle(.15);
\filldraw (2,-4.5) circle(.15);
\filldraw (3,-3.5) circle(.15);
\filldraw (3,-4.5) circle(.15);
\node[font=\fontsize{27}{6}\selectfont] at (4,-3.3) {$\frac {16}{12}$};

\draw [thick, line width = 1pt] (6, -2) -- (7, -2);
\draw [thick, line width = 1pt] (9, -2) -- (10, -2);
\draw [thin, line width = 0.5pt] (7, -2) -- (8, -2);
\draw [thin, line width = 0.5pt] (8, -2) -- (9, -2);
\draw [densely dashed, line width = 1pt] (6, -2) -- (6,-3.5);
\draw [thick, line width = 1pt] (6, -3.5) -- (6, -4.5);
\draw [densely dashed, line width = 1pt] (7, -2) -- (7, -3.5);
\draw [thick, line width = 1pt] (7, -3.5) -- (7, -4.5);
\draw [densely dashed, line width = 1pt] (7, -2) -- (7.5,-3.5);
\draw [thick, line width = 1pt] (7.5, -3.5) -- (7.5, -4.5);
\draw [densely dashed, line width = 1pt] (9, -2) -- (8.5,-3.5);
\draw [thick, line width = 1pt] (8.5, -3.5) -- (8.5, -4.5);
\draw [densely dashed, line width = 1pt] (9, -2) -- (9,-3.5);
\draw [thick, line width = 1pt] (9, -3.5) -- (9, -4.5);
\draw [densely dashed, line width = 1pt] (10, -2) -- (10,-3.5);
\draw [thick, line width = 1pt] (10, -3.5) -- (10, -4.5);
\draw [densely dashed, line width = 1pt] (8, -2) -- (8,-3.5);
\draw [thick, line width = 1pt] (8, -3.5) -- (8, -4.5);
\filldraw (6, -2) circle(.15);
\filldraw[fill = gray] (7, -2) circle(.15);
\filldraw[fill = white] (8, -2) circle(.15);
\filldraw[fill = gray] (9, -2) circle(.15);
\filldraw (10, -2) circle(.15);
\filldraw (6,-3.5) circle(.15);
\filldraw (6,-4.5) circle(.15);
\filldraw (7,-3.5) circle(.15);
\filldraw (7,-4.5) circle(.15);
\filldraw (7.5,-3.5) circle(.15);
\filldraw (7.5,-4.5) circle(.15);
\filldraw (8,-3.5) circle(.15);
\filldraw (8,-4.5) circle(.15);
\filldraw (8.5,-3.5) circle(.15);
\filldraw (8.5,-4.5) circle(.15);
\filldraw (9,-3.5) circle(.15);
\filldraw (9,-4.5) circle(.15);
\filldraw (10,-3.5) circle(.15);
\filldraw (10,-4.5) circle(.15);
\node[font=\fontsize{27}{6}\selectfont] at (11,-3.3) {$\frac {18}{13}$};

\draw [thick, line width = 1pt] (13, -2) -- (14, -2);
\draw [thick, line width = 1pt] (16, -2) -- (17, -2);
\draw [thin, line width = 0.5pt] (14, -2) -- (15, -2);
\draw [thin, line width = 0.5pt] (15, -2) -- (16, -2);
\draw [densely dashed, line width = 1pt] (13, -2) -- (13,-3.5);
\draw [thick, line width = 1pt] (13, -3.5) -- (13, -4.5);
\draw [densely dashed, line width = 1pt] (14, -2) -- (14, -3.5);
\draw [thick, line width = 1pt] (14, -3.5) -- (14, -4.5);
\draw [thin, line width = 0.5pt] (14, -4.5) -- (13.5, -4.5);
\draw [densely dashed, line width = 1pt] (14, -2) -- (14.5,-3.5);
\draw [thick, line width = 1pt] (14.5, -3.5) -- (14.5, -4.5);
\draw [densely dashed, line width = 1pt] (15, -2) -- (15, -3.5);
\draw [thick, line width = 1pt] (15, -3.5) -- (15, -4.5);
\draw [densely dashed, line width = 1pt] (16, -2) -- (15.5,-3.5);
\draw [thick, line width = 1pt] (15.5, -3.5) -- (15.5, -4.5);
\draw [densely dashed, line width = 1pt] (16, -2) -- (16,-3.5);
\draw [thick, line width = 1pt] (16, -3.5) -- (16, -4.5);
\draw [densely dashed, line width = 1pt] (17, -2) -- (17, -3.5);
\draw [thick, line width = 1pt] (17, -3.5) -- (17, -4.5);
\filldraw (13, -2) circle(.15);
\filldraw[fill = gray] (14, -2) circle(.15);
\filldraw[fill = white] (15, -2) circle(.15);
\filldraw[fill = gray] (16, -2) circle(.15);
\filldraw (17, -2) circle(.15);
\filldraw (13,-3.5) circle(.15);
\filldraw (13,-4.5) circle(.15);
\filldraw (14,-3.5) circle(.15);
\filldraw (14,-4.5) circle(.15);
\filldraw (14.5,-3.5) circle(.15);
\filldraw (14.5,-4.5) circle(.15);
\filldraw (15.5,-3.5) circle(.15);
\filldraw (15.5,-4.5) circle(.15);
\filldraw (16,-3.5) circle(.15);
\filldraw (16,-4.5) circle(.15);
\filldraw (17,-3.5) circle(.15);
\filldraw (17,-4.5) circle(.15);
\filldraw (15,-3.5) circle(.15);
\filldraw (15,-4.5) circle(.15);
\filldraw[fill = white] (13.5,-4.5) circle(.15);
\node[font=\fontsize{27}{6}\selectfont] at (18,-3.3) {$\frac {18}{14}$};


\draw [thick, line width = 1pt] (-8, -6) -- (-7, -6);
\draw [thick, line width = 1pt] (-5, -6) -- (-4, -6);
\draw [thin, line width = 0.5pt] (-7, -6) -- (-6, -6);
\draw [thin, line width = 0.5pt] (-6, -6) -- (-5, -6);
\draw [densely dashed, line width = 1pt] (-8, -6) -- (-8,-7.5);
\draw [thick, line width = 1pt] (-8, -7.5) -- (-8, -8.5);
\draw [densely dashed, line width = 1pt] (-7, -6) -- (-7, -7.5);
\draw [thin, line width = 0.5pt] (-7, -7.5) -- (-7, -8.5);
\draw [thick, line width = 1pt] (-7, -8.5) -- (-7.5, -8.5);
\draw [densely dashed, line width = 1pt] (-7, -6) -- (-6.5,-7.5);
\draw [thick, line width = 1pt] (-6.5, -7.5) -- (-6.5, -8.5);
\draw [densely dashed, line width = 1pt] (-6, -6) -- (-6,-7.5);
\draw [thick, line width = 1pt] (-6, -7.5) -- (-6, -8.5);
\draw [densely dashed, line width = 1pt] (-5, -6) -- (-5.5,-7.5);
\draw [thick, line width = 1pt] (-5.5, -7.5) -- (-5.5, -8.5);
\draw [densely dashed, line width = 1pt] (-5, -6) -- (-5,-7.5);
\draw [thick, line width = 1pt] (-5, -7.5) -- (-5, -8.5);
\draw [densely dashed, line width = 1pt] (-4, -6) -- (-4,-7.5);
\draw [thick, line width = 1pt] (-4, -7.5) -- (-4, -8.5);
\filldraw (-8, -6) circle(.15);
\filldraw[fill = gray] (-7, -6) circle(.15);
\filldraw[fill = white] (-6, -6) circle(.15);
\filldraw[fill = gray] (-5, -6) circle(.15);
\filldraw (-4, -6) circle(.15);
\filldraw (-8,-7.5) circle(.15);
\filldraw (-8,-8.5) circle(.15);
\filldraw[fill = white] (-7,-7.5) circle(.15);
\filldraw (-7,-8.5) circle(.15);
\filldraw (-7.5,-8.5) circle(.15);
\filldraw (-6.5,-7.5) circle(.15);
\filldraw (-6.5,-8.5) circle(.15);
\filldraw (-5.5,-7.5) circle(.15);
\filldraw (-5.5,-8.5) circle(.15);
\filldraw (-5,-7.5) circle(.15);
\filldraw (-5,-8.5) circle(.15);
\filldraw (-6,-7.5) circle(.15);
\filldraw (-6,-8.5) circle(.15);
\filldraw (-4,-7.5) circle(.15);
\filldraw (-4,-8.5) circle(.15);
\node[font=\fontsize{27}{6}\selectfont] at (-3,-7.3) {$\frac {18}{14}$};

\draw [thick, line width = 1pt] (-1, -6) -- (0, -6);
\draw [thick, line width = 1pt] (2, -6) -- (3, -6);
\draw [thin, line width = 0.5pt] (0, -6) -- (1, -6);
\draw [thin, line width = 0.5pt] (1, -6) -- (2, -6);
\draw [densely dashed, line width = 1pt] (-1, -6) -- (-1,-7.5);
\draw [thick, line width = 1pt] (-1, -7.5) -- (-1, -8.5);
\draw [densely dashed, line width = 1pt] (0, -6) -- (0, -7.5);
\draw [thick, line width = 1pt] (0, -7.5) -- (0, -8.5);
\draw [densely dashed, line width = 1pt] (0, -6) -- (0.5,-7.5);
\draw [thick, line width = 1pt] (0.5, -7.5) -- (0.5, -8.5);
\draw [densely dashed, line width = 1pt] (2, -6) -- (1.5,-7.5);
\draw [thick, line width = 1pt] (1.5, -7.5) -- (1.5, -8.5);
\draw [densely dashed, line width = 1pt] (2, -6) -- (2,-7.5);
\draw [thick, line width = 1pt] (2, -7.5) -- (2, -8.5);
\filldraw (-1, -6) circle(.15);
\filldraw[fill = gray] (0, -6) circle(.15);
\filldraw[fill = white] (1, -6) circle(.15);
\filldraw[fill = gray] (2, -6) circle(.15);
\filldraw (3, -6) circle(.15);
\filldraw (-1,-7.5) circle(.15);
\filldraw (-1,-8.5) circle(.15);
\filldraw (0,-7.5) circle(.15);
\filldraw (0,-8.5) circle(.15);
\filldraw (0.5,-7.5) circle(.15);
\filldraw (0.5,-8.5) circle(.15);
\filldraw (1.5,-7.5) circle(.15);
\filldraw (1.5,-8.5) circle(.15);
\filldraw (2,-7.5) circle(.15);
\filldraw (2,-8.5) circle(.15);
\node[font=\fontsize{27}{6}\selectfont] at (4,-7.3) {$\frac {14}{11}$};

\draw [thick, line width = 1pt] (6, -6) -- (7, -6);
\draw [thick, line width = 1pt] (9, -6) -- (10, -6);
\draw [thin, line width = 0.5pt] (7, -6) -- (8, -6);
\draw [thin, line width = 0.5pt] (8, -6) -- (9, -6);
\draw [densely dashed, line width = 1pt] (7, -6) -- (7, -7.5);
\draw [thick, line width = 1pt] (7, -7.5) -- (7, -8.5);
\draw [densely dashed, line width = 1pt] (7, -6) -- (7.5,-7.5);
\draw [thick, line width = 1pt] (7.5, -7.5) -- (7.5, -8.5);
\draw [densely dashed, line width = 1pt] (8, -6) -- (8,-7.5);
\draw [thick, line width = 1pt] (8, -7.5) -- (8, -8.5);
\draw [densely dashed, line width = 1pt] (9, -6) -- (8.5,-7.5);
\draw [thick, line width = 1pt] (8.5, -7.5) -- (8.5, -8.5);
\draw [densely dashed, line width = 1pt] (9, -6) -- (9,-7.5);
\draw [thick, line width = 1pt] (9, -7.5) -- (9, -8.5);
\filldraw (6, -6) circle(.15);
\filldraw[fill = gray] (7, -6) circle(.15);
\filldraw[fill = white] (8, -6) circle(.15);
\filldraw[fill = gray] (9, -6) circle(.15);
\filldraw (10, -6) circle(.15);
\filldraw (7,-8.5) circle(.15);
\filldraw (7,-7.5) circle(.15);
\filldraw (7.5,-7.5) circle(.15);
\filldraw (7.5,-8.5) circle(.15);
\filldraw (8.5,-7.5) circle(.15);
\filldraw (8.5,-8.5) circle(.15);
\filldraw (9,-7.5) circle(.15);
\filldraw (9,-8.5) circle(.15);
\filldraw (8,-7.5) circle(.15);
\filldraw (8,-8.5) circle(.15);
\node[font=\fontsize{27}{6}\selectfont] at (11,-7.3) {$\frac {14}{11}$};

\draw [thick, line width = 1pt] (13, -6) -- (14, -6);
\draw [thick, line width = 1pt] (16, -6) -- (17, -6);
\draw [thin, line width = 0.5pt] (14, -6) -- (15, -6);
\draw [thin, line width = 0.5pt] (15, -6) -- (16, -6);
\draw [densely dashed, line width = 1pt] (13, -6) -- (13, -7.5);
\draw [thick, line width = 1pt] (13, -7.5) -- (13, -8.5);
\draw [densely dashed, line width = 1pt] (13, -6) -- (13.5,-7.5);
\draw [thick, line width = 1pt] (13.5, -7.5) -- (13.5, -8.5);
\draw [densely dashed, line width = 1pt] (17, -6) -- (17,-7.5);
\draw [thick, line width = 1pt] (17, -7.5) -- (17, -8.5);
\draw [densely dashed, line width = 1pt] (16, -6) -- (15.5,-7.5);
\draw [thick, line width = 1pt] (15.5, -7.5) -- (15.5, -8.5);
\draw [densely dashed, line width = 1pt] (16, -6) -- (16,-7.5);
\draw [thick, line width = 1pt] (16, -7.5) -- (16, -8.5);
\filldraw[fill = gray] (13, -6) circle(.15);
\filldraw (14, -6) circle(.15);
\filldraw[fill = white] (15, -6) circle(.15);
\filldraw[fill = gray] (16, -6) circle(.15);
\filldraw (17, -6) circle(.15);
\filldraw (13,-7.5) circle(.15);
\filldraw (13,-8.5) circle(.15);
\filldraw (13.5,-7.5) circle(.15);
\filldraw (13.5,-8.5) circle(.15);
\filldraw (15.5,-7.5) circle(.15);
\filldraw (15.5,-8.5) circle(.15);
\filldraw (16,-7.5) circle(.15);
\filldraw (16,-8.5) circle(.15);
\filldraw (17,-7.5) circle(.15);
\filldraw (17,-8.5) circle(.15);
\node[font=\fontsize{27}{6}\selectfont] at (18,-7.3) {$\frac {14}{11}$};


\draw [thick, line width = 1pt] (-8, -10) -- (-7, -10);
\draw [thick, line width = 1pt] (-5, -10) -- (-4, -10);
\draw [thin, line width = 0.5pt] (-7, -10) -- (-6, -10);
\draw [thin, line width = 0.5pt] (-6, -10) -- (-5, -10);
\draw [densely dashed, line width = 1pt] (-8, -10) -- (-8, -11.5);
\draw [thick, line width = 1pt] (-8, -11.5) -- (-8, -12.5);
\draw [densely dashed, line width = 1pt] (-8, -10) -- (-7.5,-11.5);
\draw [thick, line width = 1pt] (-7.5, -11.5) -- (-7.5, -12.5);
\draw [densely dashed, line width = 1pt] (-6, -10) -- (-6,-11.5);
\draw [thick, line width = 1pt] (-6, -11.5) -- (-6, -12.5);
\draw [densely dashed, line width = 1pt] (-6, -10) -- (-5.5,-11.5);
\draw [thick, line width = 1pt] (-5.5, -11.5) -- (-5.5, -12.5);
\draw [densely dashed, line width = 1pt] (-7, -10) -- (-7,-11.5);
\draw [thick, line width = 1pt] (-7, -11.5) -- (-7, -12.5);
\filldraw[fill = gray] (-8, -10) circle(.15);
\filldraw (-7, -10) circle(.15);
\filldraw[fill = gray] (-6, -10) circle(.15);
\filldraw (-5, -10) circle(.15);
\filldraw (-4, -10) circle(.15);
\filldraw (-8,-11.5) circle(.15);
\filldraw (-8,-12.5) circle(.15);
\filldraw (-7.5,-11.5) circle(.15);
\filldraw (-7.5,-12.5) circle(.15);
\filldraw (-5.5,-11.5) circle(.15);
\filldraw (-5.5,-12.5) circle(.15);
\filldraw (-7,-11.5) circle(.15);
\filldraw (-7,-12.5) circle(.15);
\filldraw (-6,-11.5) circle(.15);
\filldraw (-6,-12.5) circle(.15);
\node[font=\fontsize{27}{6}\selectfont] at (-3,-11.3) {$\frac {14}{11}$};

\draw [thick, line width = 1pt] (-1, -10) -- (0, -10);
\draw [thick, line width = 1pt] (2, -10) -- (3, -10);
\draw [thin, line width = 0.5pt] (0, -10) -- (1, -10);
\draw [thin, line width = 0.5pt] (1, -10) -- (2, -10);
\draw [densely dashed, line width = 1pt] (-1, -10) -- (-1, -11.5);
\draw [thick, line width = 1pt] (-1, -11.5) -- (-1, -12.5);
\draw [densely dashed, line width = 1pt] (0, -10) -- (0,-11.5);
\draw [thick, line width = 1pt] (0, -11.5) -- (0, -12.5);
\draw [densely dashed, line width = 1pt] (0, -10) -- (0.5,-11.5);
\draw [thick, line width = 1pt] (0.5, -11.5) -- (0.5, -12.5);
\draw [densely dashed, line width = 1pt] (1, -10) -- (1,-11.5);
\draw [thick, line width = 1pt] (1, -11.5) -- (1, -12.5);
\draw [densely dashed, line width = 1pt] (1, -10) -- (1.5,-11.5);
\draw [thick, line width = 1pt] (1.5, -11.5) -- (1.5, -12.5);
\filldraw (-1, -10) circle(.15);
\filldraw[fill = gray] (0, -10) circle(.15);
\filldraw[fill = gray] (1, -10) circle(.15);
\filldraw (2, -10) circle(.15);
\filldraw (3, -10) circle(.15);
\filldraw (-1,-11.5) circle(.15);
\filldraw (-1,-12.5) circle(.15);
\filldraw (0,-11.5) circle(.15);
\filldraw (0,-12.5) circle(.15);
\filldraw (0.5,-11.5) circle(.15);
\filldraw (0.5,-12.5) circle(.15);
\filldraw (1,-11.5) circle(.15);
\filldraw (1,-12.5) circle(.15);
\filldraw (1.5,-11.5) circle(.15);
\filldraw (1.5,-12.5) circle(.15);
\node[font=\fontsize{27}{6}\selectfont] at (4,-11.3) {$\frac {14}{11}$};

\end{tikzpicture}
\end{minipage}}
\end{center}
\captionsetup{width=1.0\linewidth}
\caption{The possible structures for a critical component $K$ of $H+C$, 
where thick (respectively, dashed) edges are in the matching $M$ (respectively, the path-cycle cover $C$), 
thin edges are not in $M \cup C$, 
the filled (respectively, blank) vertices are in (respectively, not in) $V(M)$, gray vertices are $2$-anchors, 
and the fraction on the right side of each structure is $\frac{s(K)}{opt(K)}$. 
\label{fig01}}
\end{figure}

Recall that every critical component $K$ has one or two $2$-anchors.
We introduce the following definition for such two $2$-anchors.

\begin{definition}
\label{def09}
A $2$-anchor of $H+C$ is {\em critical} if it appears in a critical component of $H+C$. 
A satellite element of $H+C$ is {\em critical} if its rescue-anchor is critical in $H+C$. 
\end{definition}

\begin{definition}
\label{def10}
Suppose that $v$ is a $0$- or $1$-anchor in $H+C$ and $S$ is a satellite element in $H+C$ 
such that $S$ has a vertex $w$ with $\{v,w\}\in E(G)$. 
Then, {\em moving $S$ to $v$ in $H+C$} is the operation of modifying $C$ by replacing the rescue-edge of $S$ with the edge $\{v,w\}$. 
\end{definition}
By Figure~\ref{fig01}, we have the next fact. 

\begin{fact}
\label{fact02}
For each critical component $K$ of $H+C$ and its critical satellite element $S$, the following statements hold:
\begin{enumerate}
\parskip=0pt
\item
If we delete $S$ from $K$, then $K$ is no longer critical and will not become isolated.

\item
If $v$ is a $0$-anchor of $K$ such that $v$ is adjacent to $S$ in $G$,
then moving $S$ to $v$ in $H+C$ makes $K$ no longer critical.

\item
If $v$ is a $1$-anchor in $K$, then moving $S$ to $v$ makes $K$ remain critical only if one of the following two cases happens:

\begin{enumerate}
\parskip=0pt
\item $K$ has the first structure in Figure~\ref{fig01} and the rescue-anchor of $S$ is the unique $2$-anchor in $K$.
$G$ contains an edge $\{v, x\}$, where $x$ appears in $S$ and $v$ is the unique $1$-anchor in $K$. 

\item $K$ has the last or the second last structure in Figure~\ref{fig01} and the rescue-anchor of $S$ is the leftmost $2$-anchor in $K$.
$G$ contains an edge $\{v, x\}$, where $x$ appears in $S$ and $v$ is the unique $1$-anchor in $K$. 
\end{enumerate}
\end{enumerate}
\end{fact}
\begin{proof}
Recall that $K$ has one $2$-anchor or two $2$-anchors since $K$ is critical.
If we delete one critical satellite element $S$ of $K$, then the number of $2$-anchors is reduced by one.
If $K$ has one $2$-anchor, then after the removal of $S$, $K$ has no $2$-anchor.
So, by Lemma~\ref{lemma10}, $K$ is not critical.
If $K$ has two $2$-anchors, then $s(K)\in \{14, 16, 18\}$.
Thus, if we delete one critical satellite element of $K$, then $s(K) \in \{12, 14, 16\}$ and $K$ has one $2$-anchor.
By Lemma~\ref{lemma13}, $K$ is not critical now.
Since $K$ has at least one satellite element after the removal of $S$, then $K$ is not isolated.
So, the first statement holds.

Note that $v$ is a $0$-anchor.
So, moving $S$ to $v$ makes $K$ reduce the number of $2$-anchors by one.
By a similar proof of the first statement, the second statement holds.

If $K$ has one $2$-anchor, then by Lemma~\ref{lemma13}, either $K_c$ is an edge, $s(K)=8$ or $K_c$ is a $5$-path, $s(K)=10$.
In the former case, $K$ has a same structure by symmetry after moving $S$ to $v$. 
In the latter case, by Lemma~\ref{lemma13}, there is a pair $(v_i,v_j)\in\{(v_2,v_1), (v_4,v_5)\}$ 
such that $v_i$ is a $2$-anchor and $v_j$ is a $1$-anchor in $K$.
Without loss of generality, we assume $v_2$ is a $2$-anchor and $v_1$ is a $1$-anchor (the second structure in Figure~\ref{fig01}).
Then clearly, moving $S$ to $v_1$ makes $K$ no longer critical.

If $K$ has two $2$-anchors, $s(K)\in \{14, 16, 18\}$ and $K_c$ is a $5$-path by Lemma~\ref{lemma12}.
Furthermore, if $s(K)\in \{16, 18\}$, $v_2, v_4$ are both $2$-anchors.
Hence moving $S$ to $v$ makes $K$ no longer critical since at least one of $v_2, v_4$ is not a $2$-anchor.
The remaining case is $s(K)=14$. 
We can check the last five structures in Figure~\ref{fig01} and we find the third statement holds.
\end{proof}
 
\begin{definition}
\label{def11}
Let $K$ be a composite component of $H+C$. 
If $K$ has a $1$-anchor $v$ such that $G$ has an edge between $v$ and some critical satellite-element $S$ of $H+C$ in $G$ 
and moving $S$ to $v$ in $H+C$ makes $K$ critical in $H+C$, 
then we call $K$ a {\em responsible component} of $H+C$ and call $v$ a {\em responsible $1$-anchor} of $H+C$. 
\end{definition}
By the third statement in Fact~\ref{fact02},
a component of $H+C$ can be both critical and responsible only if it has the first or one of the last two structures in Figure~\ref{fig01}.

\begin{lemma}
\label{lemma14}
Suppose that a component $K$ of $H+C$ is both critical and responsible. 
If $K$ has the first structure in Figure~\ref{fig01}, then $s(K)=8$ and we find a feasible solution with at least $7$ vertices in $O(1)$ time; 
otherwise, $s(K)=14$ and we find a feasible solution with at least $12$ vertices in $O(1)$ time.
\end{lemma}
\begin{proof}
First, consider the case where $K$ has the first structure in Figure~\ref{fig01}. 
Without loss of generality, we assume $v_1$ and $v_2$ are the unique $2$- and $1$-anchor, respectively.
Since $K$ is responsible, $G$ contains an edge $\{v_2, x\}$, where $x$ appears in a satellite element $S$ whose rescue-anchor is $v_1$. 
If $S$ and $Q_{v_1}$ are not vertex-disjoint, then we find the satellite element $S'$ other than $S$ whose rescue-anchor is $v_1$, 
and re-define $Q_{v_1}$ to be the longest path among those paths in $K$ which starts with $v_1$ and the rescue-edge of $S'$. 
Then,  $G[V(K)]$ contains 
a $7^+$-path in which $Q_{v_2}$, the edge $\{v,x\}$, a path from $x$ to $v_1$ in $K$, and $Q_{v_1}$ appear in this order. 

Next, consider the case where $K$ has one of the last two structures in Figure~\ref{fig01}. 
We assume the last structure in Figure~\ref{fig01}; 
the following discussion also applies to the other case similarly. 
Since $K$ is responsible, $G$ contains an edge $\{v_1, x\}$, where $x$ appears in 
a satellite element $S$ whose rescue-anchor is $v_2$. 
If $S$ and $Q_{v_2}$ are not vertex-disjoint, then we find the satellite element $S'$ other than $S$ whose rescue-anchor is $v_2$, 
and re-define $Q_{v_2}$ to be the longest path among those paths in $K$ which starts with $v_2$ and the rescue-edge of $S'$.
Besides $P_{v_3}$, $G$ contains a $7^+$-path obtained by connecting $Q_{v_1}$, the edge $\{v_1,x\}$, a path from $x$ to $v_2$ in $K$, and $Q_{v_2}$. 
So, we can find a solution with at least $12$ vertices, which completes the proof. 
\end{proof}

By the above lemma, we know for each critical and responsible component $K$,
we can find a feasible solution for $K$ in constant time, which is still denoted as $OPT(K)$ for ease of presentation,
with $\frac {s(K)}{opt(K)}< \frac {14}{11}$.
Now, we can regard each critical and responsible component $K$ as a non-critical component.
So, any critical component cannot be responsible or vice versa.
Hereafter, a critical component always refers to a critical but not responsible component
and a responsible component always refers to a responsible but not critical component.

By Definition~\ref{def11}, the structure for a responsible component of $H+C$ can only be obtained 
by deleting a critical satellite-element from one of the structures in Figure~\ref{fig01}. 
So, by Figure~\ref{fig01}, we can list all possible structures for responsible components of $H+C$,
which are shown in Figure~\ref{fig02}.

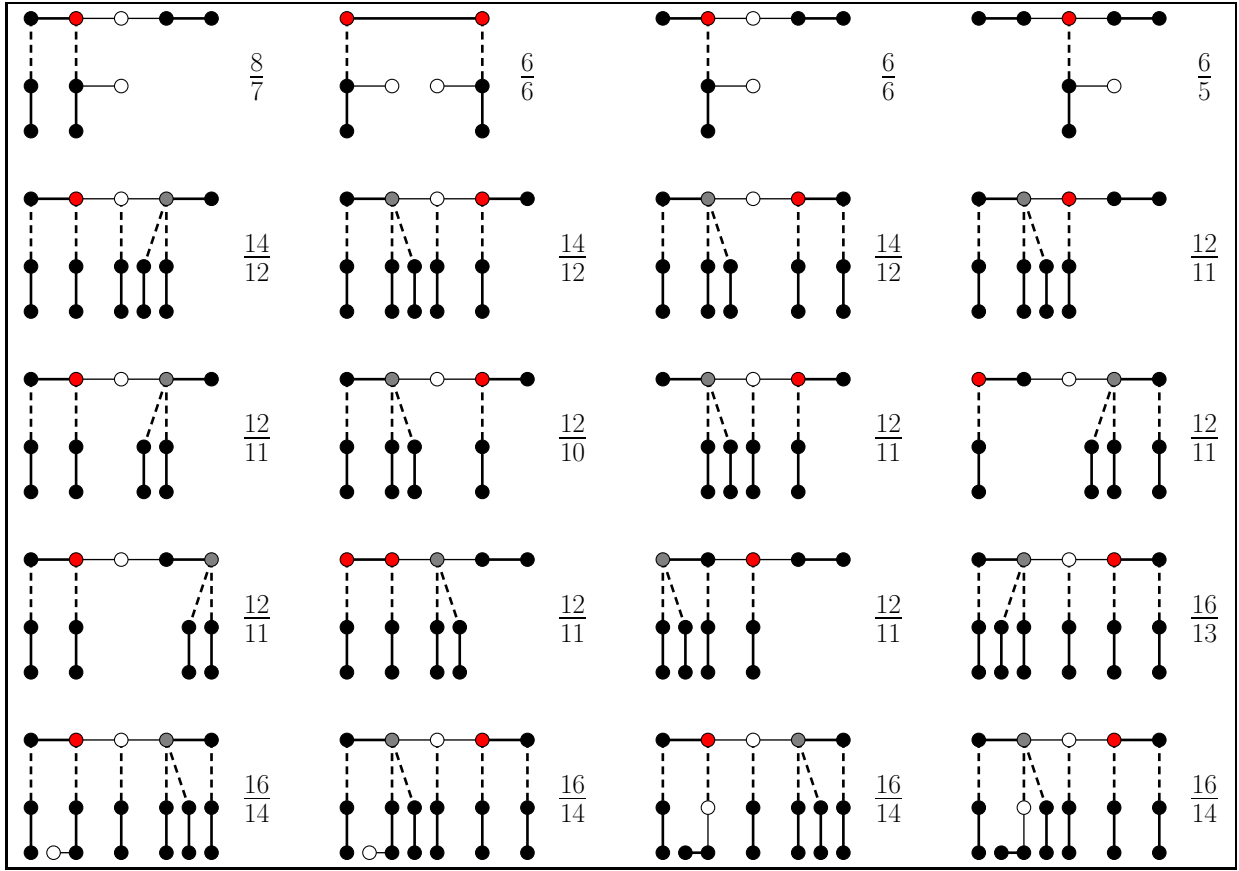
\begin{figure}[thb]
\begin{center}
\framebox{
\begin{minipage}{6.3in}
\begin{tikzpicture}[scale=0.6,transform shape]

\draw [thick, line width = 1pt] (-8, 8) -- (-7, 8);
\draw [thick, line width = 1pt] (-5, 8) -- (-4, 8);
\draw [thin, line width = 0.5pt] (-7, 8) -- (-6, 8);
\draw [thin, line width = 0.5pt] (-6,8) -- (-5,8);
\draw [densely dashed, line width = 1pt] (-8, 8) -- (-8, 6.5);
\draw [thick, line width = 1pt] (-8, 6.5) -- (-8, 5.5);
\draw [densely dashed, line width = 1pt] (-7, 8) -- (-7, 6.5);
\draw [thick, line width = 1pt] (-7, 6.5) -- (-7, 5.5);
\draw [thin, line width = 0.5pt] (-7, 6.5) -- (-6, 6.5);
\filldraw (-8, 8) circle(.15);
\filldraw[fill = red] (-7, 8) circle(.15);
\filldraw[fill = white] (-6, 8) circle(.15);
\filldraw (-5, 8) circle(.15);
\filldraw (-4, 8) circle(.15);
\filldraw (-8,6.5) circle(.15);
\filldraw (-8,5.5) circle(.15);
\filldraw (-7,6.5) circle(.15);
\filldraw (-7,5.5) circle(.15);
\filldraw[fill = white] (-6, 6.5) circle(.15);
\node[font=\fontsize{27}{6}\selectfont] at (-3,6.7) {$\frac 87$};

\draw [thick, line width = 1pt] (-1, 8) -- (2, 8);
\draw [densely dashed, line width = 1pt] (-1, 8) -- (-1, 6.5);
\draw [thick, line width = 1pt] (-1, 6.5) -- (-1, 5.5);
\draw [thin, line width = 0.5pt] (-1, 6.5) -- (0, 6.5);
\draw [densely dashed, line width = 1pt] (2, 8) -- (2, 6.5);
\draw [thick, line width = 1pt] (2, 6.5) -- (2, 5.5);
\draw [thin, line width = 0.5pt] (2, 6.5) -- (1, 6.5);
\filldraw[fill = red] (-1, 8) circle(.15);
\filldraw[fill = red] (2, 8) circle(.15);
\filldraw (-1, 6.5) circle(.15);
\filldraw (-1, 5.5) circle(.15);
\filldraw[fill = white] (0, 6.5) circle(.15);
\filldraw (2, 6.5) circle(.15);
\filldraw (2, 5.5) circle(.15);
\filldraw[fill = white] (1, 6.5) circle(.15);
\node[font=\fontsize{27}{6}\selectfont] at (3,6.7) {$\frac 66$};

\draw [thick, line width = 1pt] (6, 8) -- (7, 8);
\draw [thick, line width = 1pt] (9,8) -- (10,8);
\draw [thin, line width = 0.5pt] (7,8) -- (8,8);
\draw [thin, line width = 0.5pt] (8,8) -- (9,8);
\draw [densely dashed, line width = 1pt] (7, 8) -- (7, 6.5);
\draw [thick, line width = 1pt] (7, 6.5) -- (7, 5.5);
\draw [thin, line width = 0.5pt] (7, 6.5) -- (8, 6.5);
\filldraw (6, 8) circle(.15);
\filldraw[fill = red] (7, 8) circle(.15);
\filldraw[fill = white] (8, 8) circle(.15);
\filldraw (9, 8) circle(.15);
\filldraw (10, 8) circle(.15);
\filldraw (7,6.5) circle(.15);
\filldraw (7,5.5) circle(.15);
\filldraw[fill = white] (8, 6.5) circle(.15);
\node[font=\fontsize{27}{6}\selectfont] at (11,6.7) {$\frac 66$};

\draw [thick, line width = 1pt] (13, 8) -- (14, 8);
\draw [thick, line width = 1pt] (16,8) -- (17,8);
\draw [thin, line width = 0.5pt] (14, 8) -- (15,8);
\draw [thin, line width = 0.5pt] (15, 8) -- (16,8);
\draw [densely dashed, line width = 1pt] (15, 8) -- (15, 6.5);
\draw [thick, line width = 1pt] (15, 6.5) -- (15, 5.5);
\draw [thin, line width = 0.5pt] (15, 6.5) -- (16, 6.5);
\filldraw (13, 8) circle(.15);
\filldraw (14, 8) circle(.15);
\filldraw[fill = red] (15, 8) circle(.15);
\filldraw (16, 8) circle(.15);
\filldraw (17, 8) circle(.15);
\filldraw (15,6.5) circle(.15);
\filldraw (15,5.5) circle(.15);
\filldraw[fill = white] (16, 6.5) circle(.15);
\node[font=\fontsize{27}{6}\selectfont] at (18,6.7) {$\frac 65$};


\draw [thick, line width = 1pt] (-8, 4) -- (-7, 4);
\draw [thick, line width = 1pt] (-5, 4) -- (-4, 4);
\draw [thin, line width = 0.5pt] (-7, 4) -- (-6, 4);
\draw [thin, line width = 0.5pt] (-6, 4) -- (-5, 4);
\draw [densely dashed, line width = 1pt] (-8, 4) -- (-8,2.5);
\draw [thick, line width = 1pt] (-8, 2.5) -- (-8, 1.5);
\draw [densely dashed, line width = 1pt] (-7, 4) -- (-7, 2.5);
\draw [thick, line width = 1pt] (-7, 2.5) -- (-7, 1.5);
\draw [densely dashed, line width = 1pt] (-6, 4) -- (-6,2.5);
\draw [thick, line width = 1pt] (-6, 2.5) -- (-6, 1.5);
\draw [densely dashed, line width = 1pt] (-5, 4) -- (-5.5,2.5);
\draw [thick, line width = 1pt] (-5.5, 2.5) -- (-5.5, 1.5);
\draw [densely dashed, line width = 1pt] (-5, 4) -- (-5,2.5);
\draw [thick, line width = 1pt] (-5, 2.5) -- (-5, 1.5);
\filldraw (-8, 4) circle(.15);
\filldraw[fill = red] (-7, 4) circle(.15);
\filldraw[fill = white] (-6, 4) circle(.15);
\filldraw[fill = gray] (-5, 4) circle(.15);
\filldraw (-4, 4) circle(.15);
\filldraw (-8,2.5) circle(.15);
\filldraw (-8,1.5) circle(.15);
\filldraw (-7,2.5) circle(.15);
\filldraw (-7,1.5) circle(.15);
\filldraw (-5.5,2.5) circle(.15);
\filldraw (-5.5,1.5) circle(.15);
\filldraw (-5,2.5) circle(.15);
\filldraw (-5,1.5) circle(.15);
\filldraw (-6,2.5) circle(.15);
\filldraw (-6,1.5) circle(.15);
\node[font=\fontsize{25}{6}\selectfont] at (-3,2.7) {$\frac {14}{12}$};

\draw [thick, line width = 1pt] (-1, 4) -- (0, 4);
\draw [thick, line width = 1pt] (2, 4) -- (3, 4);
\draw [thin, line width = 0.5pt] (0, 4) -- (1, 4);
\draw [thin, line width = 0.5pt] (1, 4) -- (2, 4);
\draw [densely dashed, line width = 1pt] (-1, 4) -- (-1,2.5);
\draw [thick, line width = 1pt] (-1, 2.5) -- (-1, 1.5);
\draw [densely dashed, line width = 1pt] (0, 4) -- (0, 2.5);
\draw [thick, line width = 1pt] (0, 2.5) -- (0, 1.5);
\draw [densely dashed, line width = 1pt] (1, 4) -- (1,2.5);
\draw [thick, line width = 1pt] (1, 2.5) -- (1, 1.5);
\draw [densely dashed, line width = 1pt] (0, 4) -- (0.5,2.5);
\draw [thick, line width = 1pt] (0.5, 2.5) -- (0.5, 1.5);
\draw [densely dashed, line width = 1pt] (2, 4) -- (2,2.5);
\draw [thick, line width = 1pt] (2, 2.5) -- (2, 1.5);
\filldraw (-1, 4) circle(.15);
\filldraw[fill = gray] (0, 4) circle(.15);
\filldraw[fill = white] (1, 4) circle(.15);
\filldraw[fill = red] (2, 4) circle(.15);
\filldraw (3, 4) circle(.15);
\filldraw (-1,2.5) circle(.15);
\filldraw (-1,1.5) circle(.15);
\filldraw (0,2.5) circle(.15);
\filldraw (0,1.5) circle(.15);
\filldraw (2,2.5) circle(.15);
\filldraw (2,1.5) circle(.15);
\filldraw (0.5,2.5) circle(.15);
\filldraw (0.5,1.5) circle(.15);
\filldraw (1,2.5) circle(.15);
\filldraw (1,1.5) circle(.15);
\node[font=\fontsize{25}{6}\selectfont] at (4,2.7) {$\frac {14}{12}$};

\draw [thick, line width = 1pt] (6, 4) -- (7,4);
\draw [thick, line width = 1pt] (9, 4) -- (10, 4);
\draw [thin, line width = 0.5pt] (7, 4) -- (8, 4);
\draw [thin, line width = 0.5pt] (8, 4) -- (9, 4);
\draw [densely dashed, line width = 1pt] (6, 4) -- (6,2.5);
\draw [thick, line width = 1pt] (6, 2.5) -- (6, 1.5);
\draw [densely dashed, line width = 1pt] (7, 4) -- (7, 2.5);
\draw [thick, line width = 1pt] (7, 2.5) -- (7, 1.5);
\draw [densely dashed, line width = 1pt] (7, 4) -- (7.5,2.5);
\draw [thick, line width = 1pt] (7.5, 2.5) -- (7.5, 1.5);
\draw [densely dashed, line width = 1pt] (9, 4) -- (9,2.5);
\draw [thick, line width = 1pt] (9, 2.5) -- (9, 1.5);
\draw [densely dashed, line width = 1pt] (10, 4) -- (10,2.5);
\draw [thick, line width = 1pt] (10, 2.5) -- (10, 1.5);
\filldraw (6, 4) circle(.15);
\filldraw[fill = gray] (7, 4) circle(.15);
\filldraw[fill = white] (8, 4) circle(.15);
\filldraw[fill = red] (9, 4) circle(.15);
\filldraw (10, 4) circle(.15);
\filldraw (6,2.5) circle(.15);
\filldraw (6,1.5) circle(.15);
\filldraw (7,2.5) circle(.15);
\filldraw (7,1.5) circle(.15);
\filldraw (7.5,2.5) circle(.15);
\filldraw (7.5,1.5) circle(.15);
\filldraw (9,2.5) circle(.15);
\filldraw (9,1.5) circle(.15);
\filldraw (10,2.5) circle(.15);
\filldraw (10,1.5) circle(.15);
\node[font=\fontsize{25}{6}\selectfont] at (11,2.7) {$\frac {14}{12}$};

\draw [thick, line width = 1pt] (13, 4) -- (14,4);
\draw [thick, line width = 1pt] (16, 4) -- (17, 4);
\draw [thin, line width = 0.5pt] (14, 4) -- (15, 4);
\draw [thin, line width = 0.5pt] (15, 4) -- (16, 4);
\draw [densely dashed, line width = 1pt] (13, 4) -- (13,2.5);
\draw [thick, line width = 1pt] (13, 2.5) -- (13, 1.5);
\draw [densely dashed, line width = 1pt] (14, 4) -- (14, 2.5);
\draw [thick, line width = 1pt] (14, 2.5) -- (14, 1.5);
\draw [densely dashed, line width = 1pt] (14, 4) -- (14.5,2.5);
\draw [thick, line width = 1pt] (14.5, 2.5) -- (14.5, 1.5);
\draw [densely dashed, line width = 1pt] (15, 4) -- (15,2.5);
\draw [thick, line width = 1pt] (15, 2.5) -- (15, 1.5);
\filldraw (13, 4) circle(.15);
\filldraw[fill = gray] (14, 4) circle(.15);
\filldraw[fill = red] (15, 4) circle(.15);
\filldraw (16, 4) circle(.15);
\filldraw (17, 4) circle(.15);
\filldraw (13,2.5) circle(.15);
\filldraw (13,1.5) circle(.15);
\filldraw (14,2.5) circle(.15);
\filldraw (14,1.5) circle(.15);
\filldraw (14.5,2.5) circle(.15);
\filldraw (14.5,1.5) circle(.15);
\filldraw (15,2.5) circle(.15);
\filldraw (15,1.5) circle(.15);
\node[font=\fontsize{25}{6}\selectfont] at (18,2.7) {$\frac {12}{11}$};


\draw [thick, line width = 1pt] (-8, 0) -- (-7, 0);
\draw [thick, line width = 1pt] (-5, 0) -- (-4, 0);
\draw [thin, line width = 0.5pt] (-7, 0) -- (-6, 0);
\draw [thin, line width = 0.5pt] (-6, 0) -- (-5, 0);
\draw [densely dashed, line width = 1pt] (-8, 0) -- (-8,-1.5);
\draw [thick, line width = 1pt] (-8, -1.5) -- (-8, -2.5);
\draw [densely dashed, line width = 1pt] (-7, 0) -- (-7, -1.5);
\draw [thick, line width = 1pt] (-7, -1.5) -- (-7, -2.5);
\draw [densely dashed, line width = 1pt] (-5, 0) -- (-5.5,-1.5);
\draw [thick, line width = 1pt] (-5.5, -1.5) -- (-5.5, -2.5);
\draw [densely dashed, line width = 1pt] (-5, 0) -- (-5,-1.5);
\draw [thick, line width = 1pt] (-5, -1.5) -- (-5, -2.5);
\filldraw (-8, 0) circle(.15);
\filldraw[fill = red] (-7, 0) circle(.15);
\filldraw[fill = white] (-6, 0) circle(.15);
\filldraw[fill = gray] (-5, 0) circle(.15);
\filldraw (-4, 0) circle(.15);
\filldraw (-8,-1.5) circle(.15);
\filldraw (-8,-2.5) circle(.15);
\filldraw (-7,-1.5) circle(.15);
\filldraw (-7,-2.5) circle(.15);
\filldraw (-5.5,-1.5) circle(.15);
\filldraw (-5.5,-2.5) circle(.15);
\filldraw (-5,-1.5) circle(.15);
\filldraw (-5,-2.5) circle(.15);
\node[font=\fontsize{27}{6}\selectfont] at (-3,-1.3) {$\frac {12}{11}$};

\draw [thick, line width = 1pt] (-1, 0) -- (0, 0);
\draw [thick, line width = 1pt] (2, 0) -- (3, 0);
\draw [thin, line width = 0.5pt] (0, 0) -- (1, 0);
\draw [thin, line width = 0.5pt] (1, 0) -- (2, 0);
\draw [densely dashed, line width = 1pt] (-1, 0) -- (-1,-1.5);
\draw [thick, line width = 1pt] (-1, -1.5) -- (-1, -2.5);
\draw [densely dashed, line width = 1pt] (0, 0) -- (0, -1.5);
\draw [thick, line width = 1pt] (0, -1.5) -- (0, -2.5);
\draw [densely dashed, line width = 1pt] (0, 0) -- (0.5,-1.5);
\draw [thick, line width = 1pt] (0.5, -1.5) -- (0.5, -2.5);
\draw [densely dashed, line width = 1pt] (2, 0) -- (2,-1.5);
\draw [thick, line width = 1pt] (2, -1.5) -- (2, -2.5);
\filldraw (-1, 0) circle(.15);
\filldraw[fill = gray] (0, 0) circle(.15);
\filldraw[fill = white] (1, 0) circle(.15);
\filldraw[fill = red] (2, 0) circle(.15);
\filldraw (3, 0) circle(.15);
\filldraw (-1,-1.5) circle(.15);
\filldraw (-1,-2.5) circle(.15);
\filldraw (0,-1.5) circle(.15);
\filldraw (0,-2.5) circle(.15);
\filldraw (0.5,-1.5) circle(.15);
\filldraw (0.5,-2.5) circle(.15);
\filldraw (2,-1.5) circle(.15);
\filldraw (2,-2.5) circle(.15);
\node[font=\fontsize{27}{6}\selectfont] at (4,-1.3) {$\frac {12}{10}$};

\draw [thick, line width = 1pt] (6, 0) -- (7, 0);
\draw [thick, line width = 1pt] (9, 0) -- (10, 0);
\draw [thin, line width = 0.5pt] (7, 0) -- (8, 0);
\draw [thin, line width = 0.5pt] (8, 0) -- (9, 0);
\draw [densely dashed, line width = 1pt] (8, 0) -- (8,-1.5);
\draw [thick, line width = 1pt] (8, -1.5) -- (8, -2.5);
\draw [densely dashed, line width = 1pt] (7, 0) -- (7, -1.5);
\draw [thick, line width = 1pt] (7, -1.5) -- (7, -2.5);
\draw [densely dashed, line width = 1pt] (7, 0) -- (7.5,-1.5);
\draw [thick, line width = 1pt] (7.5, -1.5) -- (7.5, -2.5);
\draw [densely dashed, line width = 1pt] (9, 0) -- (9,-1.5);
\draw [thick, line width = 1pt] (9, -1.5) -- (9, -2.5);
\filldraw (6, 0) circle(.15);
\filldraw[fill = gray] (7, 0) circle(.15);
\filldraw[fill = white] (8, 0) circle(.15);
\filldraw[fill = red] (9, 0) circle(.15);
\filldraw (10, 0) circle(.15);
\filldraw (8,-1.5) circle(.15);
\filldraw (8,-2.5) circle(.15);
\filldraw (7,-1.5) circle(.15);
\filldraw (7,-2.5) circle(.15);
\filldraw (7.5,-1.5) circle(.15);
\filldraw (7.5,-2.5) circle(.15);
\filldraw (9,-1.5) circle(.15);
\filldraw (9,-2.5) circle(.15);
\node[font=\fontsize{27}{6}\selectfont] at (11,-1.3) {$\frac {12}{11}$};

\draw [thick, line width = 1pt] (13, 0) -- (14, 0);
\draw [thick, line width = 1pt] (16, 0) -- (17, 0);
\draw [thin, line width = 0.5pt] (14, 0) -- (15, 0);
\draw [thin, line width = 0.5pt] (15, 0) -- (16, 0);
\draw [densely dashed, line width = 1pt] (13, 0) -- (13,-1.5);
\draw [thick, line width = 1pt] (13, -1.5) -- (13, -2.5);
\draw [densely dashed, line width = 1pt] (16, 0) -- (16, -1.5);
\draw [thick, line width = 1pt] (16, -1.5) -- (16, -2.5);
\draw [densely dashed, line width = 1pt] (16, 0) -- (15.5,-1.5);
\draw [thick, line width = 1pt] (15.5, -1.5) -- (15.5, -2.5);
\draw [densely dashed, line width = 1pt] (17, 0) -- (17,-1.5);
\draw [thick, line width = 1pt] (17, -1.5) -- (17, -2.5);
\filldraw[fill = red] (13, 0) circle(.15);
\filldraw (14, 0) circle(.15);
\filldraw[fill = white] (15, 0) circle(.15);
\filldraw[fill = gray] (16, 0) circle(.15);
\filldraw (17, 0) circle(.15);
\filldraw (13,-1.5) circle(.15);
\filldraw (13,-2.5) circle(.15);
\filldraw (16,-1.5) circle(.15);
\filldraw (16,-2.5) circle(.15);
\filldraw (15.5,-1.5) circle(.15);
\filldraw (15.5,-2.5) circle(.15);
\filldraw (17,-1.5) circle(.15);
\filldraw (17,-2.5) circle(.15);
\node[font=\fontsize{27}{6}\selectfont] at (18,-1.3) {$\frac {12}{11}$};


\draw [thick, line width = 1pt] (-8, -4) -- (-7, -4);
\draw [thick, line width = 1pt] (-5, -4) -- (-4, -4);
\draw [thin, line width = 0.5pt] (-7, -4) -- (-6, -4);
\draw [thin, line width = 0.5pt] (-6, -4) -- (-5, -4);
\draw [densely dashed, line width = 1pt] (-8, -4) -- (-8,-5.5);
\draw [thick, line width = 1pt] (-8, -5.5) -- (-8, -6.5);
\draw [densely dashed, line width = 1pt] (-7, -4) -- (-7, -5.5);
\draw [thick, line width = 1pt] (-7, -5.5) -- (-7, -6.5);
\draw [densely dashed, line width = 1pt] (-4, -4) -- (-4.5,-5.5);
\draw [thick, line width = 1pt] (-4.5, -5.5) -- (-4.5, -6.5);
\draw [densely dashed, line width = 1pt] (-4, -4) -- (-4,-5.5);
\draw [thick, line width = 1pt] (-4, -5.5) -- (-4, -6.5);
\filldraw (-8, -4) circle(.15);
\filldraw[fill = red] (-7, -4) circle(.15);
\filldraw[fill = white] (-6, -4) circle(.15);
\filldraw (-5, -4) circle(.15);
\filldraw[fill = gray] (-4, -4) circle(.15);
\filldraw (-8,-5.5) circle(.15);
\filldraw (-8,-6.5) circle(.15);
\filldraw (-7,-5.5) circle(.15);
\filldraw (-7,-6.5) circle(.15);
\filldraw (-4.5,-5.5) circle(.15);
\filldraw (-4.5,-6.5) circle(.15);
\filldraw (-4,-5.5) circle(.15);
\filldraw (-4,-6.5) circle(.15);
\node[font=\fontsize{27}{6}\selectfont] at (-3,-5.3) {$\frac {12}{11}$};

\draw [thick, line width = 1pt] (-1, -4) -- (0, -4);
\draw [thick, line width = 1pt] (2, -4) -- (3, -4);
\draw [thin, line width = 0.5pt] (0, -4) -- (1, -4);
\draw [thin, line width = 0.5pt] (1, -4) -- (2, -4);
\draw [densely dashed, line width = 1pt] (-1, -4) -- (-1,-5.5);
\draw [thick, line width = 1pt] (-1, -5.5) -- (-1, -6.5);
\draw [densely dashed, line width = 1pt] (0, -4) -- (0, -5.5);
\draw [thick, line width = 1pt] (0, -5.5) -- (0, -6.5);
\draw [densely dashed, line width = 1pt] (1, -4) -- (1.5,-5.5);
\draw [thick, line width = 1pt] (1.5, -5.5) -- (1.5, -6.5);
\draw [densely dashed, line width = 1pt] (1, -4) -- (1,-5.5);
\draw [thick, line width = 1pt] (1, -5.5) -- (1, -6.5);
\filldraw[fill = red] (-1, -4) circle(.15);
\filldraw[fill = red] (0, -4) circle(.15);
\filldraw[fill = gray] (1, -4) circle(.15);
\filldraw (2, -4) circle(.15);
\filldraw (3, -4) circle(.15);
\filldraw (-1,-5.5) circle(.15);
\filldraw (-1,-6.5) circle(.15);
\filldraw (0,-5.5) circle(.15);
\filldraw (0,-6.5) circle(.15);
\filldraw (1.5,-5.5) circle(.15);
\filldraw (1.5,-6.5) circle(.15);
\filldraw (1,-5.5) circle(.15);
\filldraw (1,-6.5) circle(.15);
\node[font=\fontsize{27}{6}\selectfont] at (4,-5.3) {$\frac {12}{11}$};

\draw [thick, line width = 1pt] (6, -4) -- (7, -4);
\draw [thick, line width = 1pt] (9, -4) -- (10, -4);
\draw [thin, line width = 0.5pt] (7, -4) -- (8, -4);
\draw [thin, line width = 0.5pt] (8, -4) -- (9, -4);
\draw [densely dashed, line width = 1pt] (6, -4) -- (6,-5.5);
\draw [thick, line width = 1pt] (6, -5.5) -- (6, -6.5);
\draw [densely dashed, line width = 1pt] (6, -4) -- (6.5, -5.5);
\draw [thick, line width = 1pt] (6.5, -5.5) -- (6.5, -6.5);
\draw [densely dashed, line width = 1pt] (7, -4) -- (7,-5.5);
\draw [thick, line width = 1pt] (7, -5.5) -- (7, -6.5);
\draw [densely dashed, line width = 1pt] (8, -4) -- (8,-5.5);
\draw [thick, line width = 1pt] (8, -5.5) -- (8, -6.5);
\filldraw[fill = gray] (6, -4) circle(.15);
\filldraw (7, -4) circle(.15);
\filldraw[fill = red] (8, -4) circle(.15);
\filldraw (9, -4) circle(.15);
\filldraw (10, -4) circle(.15);
\filldraw (6,-5.5) circle(.15);
\filldraw (6,-6.5) circle(.15);
\filldraw (6.5,-5.5) circle(.15);
\filldraw (6.5,-6.5) circle(.15);
\filldraw (7,-5.5) circle(.15);
\filldraw (7,-6.5) circle(.15);
\filldraw (8,-5.5) circle(.15);
\filldraw (8,-6.5) circle(.15);
\node[font=\fontsize{27}{6}\selectfont] at (11,-5.3) {$\frac {12}{11}$};

\draw [thick, line width = 1pt] (13, -4) -- (14, -4);
\draw [thick, line width = 1pt] (16, -4) -- (17, -4);
\draw [thin, line width = 0.5pt] (14, -4) -- (15, -4);
\draw [thin, line width = 0.5pt] (15, -4) -- (16, -4);
\draw [densely dashed, line width = 1pt] (13, -4) -- (13,-5.5);
\draw [thick, line width = 1pt] (13, -5.5) -- (13, -6.5);
\draw [densely dashed, line width = 1pt] (14, -4) -- (13.5, -5.5);
\draw [thick, line width = 1pt] (13.5, -5.5) -- (13.5, -6.5);
\draw [densely dashed, line width = 1pt] (14, -4) -- (14,-5.5);
\draw [thick, line width = 1pt] (14, -5.5) -- (14, -6.5);
\draw [densely dashed, line width = 1pt] (15, -4) -- (15,-5.5);
\draw [thick, line width = 1pt] (15, -5.5) -- (15, -6.5);
\draw [densely dashed, line width = 1pt] (16, -4) -- (16,-5.5);
\draw [thick, line width = 1pt] (16, -5.5) -- (16, -6.5);
\draw [densely dashed, line width = 1pt] (17, -4) -- (17,-5.5);
\draw [thick, line width = 1pt] (17, -5.5) -- (17, -6.5);
\filldraw (13, -4) circle(.15);
\filldraw[fill = gray] (14, -4) circle(.15);
\filldraw[fill = white] (15, -4) circle(.15);
\filldraw[fill = red] (16, -4) circle(.15);
\filldraw (17, -4) circle(.15);
\filldraw (13,-5.5) circle(.15);
\filldraw (13,-6.5) circle(.15);
\filldraw (13.5,-5.5) circle(.15);
\filldraw (13.5,-6.5) circle(.15);
\filldraw (14,-5.5) circle(.15);
\filldraw (14,-6.5) circle(.15);
\filldraw (15,-5.5) circle(.15);
\filldraw (15,-6.5) circle(.15);
\filldraw (16,-5.5) circle(.15);
\filldraw (16,-6.5) circle(.15);
\filldraw (17,-5.5) circle(.15);
\filldraw (17,-6.5) circle(.15);
\node[font=\fontsize{27}{6}\selectfont] at (18,-5.3) {$\frac {16}{13}$};


\draw [thick, line width = 1pt] (-8, -8) -- (-7, -8);
\draw [thick, line width = 1pt] (-5, -8) -- (-4, -8);
\draw [thin, line width = 0.5pt] (-7, -8) -- (-6, -8);
\draw [thin, line width = 0.5pt] (-6, -8) -- (-5, -8);
\draw [densely dashed, line width = 1pt] (-8, -8) -- (-8,-9.5);
\draw [thick, line width = 1pt] (-8, -9.5) -- (-8, -10.5);
\draw [densely dashed, line width = 1pt] (-5, -8) -- (-4.5, -9.5);
\draw [thick, line width = 1pt] (-4.5, -9.5) -- (-4.5, -10.5);
\draw [densely dashed, line width = 1pt] (-7, -8) -- (-7,-9.5);
\draw [thick, line width = 1pt] (-7, -9.5) -- (-7, -10.5);
\draw [thin, line width = 0.5pt] (-7, -10.5) -- (-7.5, -10.5);
\draw [densely dashed, line width = 1pt] (-6, -8) -- (-6,-9.5);
\draw [thick, line width = 1pt] (-6, -9.5) -- (-6, -10.5);
\draw [densely dashed, line width = 1pt] (-5, -8) -- (-5,-9.5);
\draw [thick, line width = 1pt] (-5, -9.5) -- (-5, -10.5);
\draw [densely dashed, line width = 1pt] (-4, -8) -- (-4,-9.5);
\draw [thick, line width = 1pt] (-4, -9.5) -- (-4, -10.5);
\filldraw (-8, -8) circle(.15);
\filldraw[fill = red] (-7, -8) circle(.15);
\filldraw[fill = white] (-6, -8) circle(.15);
\filldraw[fill = gray] (-5, -8) circle(.15);
\filldraw (-4, -8) circle(.15);
\filldraw (-8,-9.5) circle(.15);
\filldraw (-8,-10.5) circle(.15);
\filldraw[fill = white] (-7.5,-10.5) circle(.15);
\filldraw (-4.5,-9.5) circle(.15);
\filldraw (-4.5,-10.5) circle(.15);
\filldraw (-7,-9.5) circle(.15);
\filldraw (-7,-10.5) circle(.15);
\filldraw (-6,-9.5) circle(.15);
\filldraw (-6,-10.5) circle(.15);
\filldraw (-5,-9.5) circle(.15);
\filldraw (-5,-10.5) circle(.15);
\filldraw (-4,-9.5) circle(.15);
\filldraw (-4,-10.5) circle(.15);
\node[font=\fontsize{27}{6}\selectfont] at (-3,-9.3) {$\frac {16}{14}$};

\draw [thick, line width = 1pt] (-1, -8) -- (0, -8);
\draw [thick, line width = 1pt] (2, -8) -- (3, -8);
\draw [thin, line width = 0.5pt] (0, -8) -- (1, -8);
\draw [thin, line width = 0.5pt] (1, -8) -- (2, -8);
\draw [densely dashed, line width = 1pt] (-1, -8) -- (-1,-9.5);
\draw [thick, line width = 1pt] (-1, -9.5) -- (-1, -10.5);
\draw [densely dashed, line width = 1pt] (0, -8) -- (0.5, -9.5);
\draw [thick, line width = 1pt] (0.5, -9.5) -- (0.5, -10.5);
\draw [densely dashed, line width = 1pt] (0, -8) -- (0,-9.5);
\draw [thick, line width = 1pt] (0, -9.5) -- (0, -10.5);
\draw [thin, line width = 0.5pt] (0, -10.5) -- (-0.5, -10.5);
\draw [densely dashed, line width = 1pt] (1, -8) -- (1,-9.5);
\draw [thick, line width = 1pt] (1, -9.5) -- (1, -10.5);
\draw [densely dashed, line width = 1pt] (2, -8) -- (2,-9.5);
\draw [thick, line width = 1pt] (2, -9.5) -- (2, -10.5);
\draw [densely dashed, line width = 1pt] (3, -8) -- (3,-9.5);
\draw [thick, line width = 1pt] (3, -9.5) -- (3, -10.5);
\filldraw (-1, -8) circle(.15);
\filldraw[fill = gray] (0, -8) circle(.15);
\filldraw[fill = white] (1, -8) circle(.15);
\filldraw[fill = red] (2, -8) circle(.15);
\filldraw (3, -8) circle(.15);
\filldraw (-1,-9.5) circle(.15);
\filldraw (-1,-10.5) circle(.15);
\filldraw[fill = white] (-0.5,-10.5) circle(.15);
\filldraw (0.5,-9.5) circle(.15);
\filldraw (0.5,-10.5) circle(.15);
\filldraw (0,-9.5) circle(.15);
\filldraw (0,-10.5) circle(.15);
\filldraw (1,-9.5) circle(.15);
\filldraw (1,-10.5) circle(.15);
\filldraw (2,-9.5) circle(.15);
\filldraw (2,-10.5) circle(.15);
\filldraw (3,-9.5) circle(.15);
\filldraw (3,-10.5) circle(.15);
\node[font=\fontsize{27}{6}\selectfont] at (4,-9.3) {$\frac {16}{14}$};

\draw [thick, line width = 1pt] (6, -8) -- (7, -8);
\draw [thick, line width = 1pt] (9, -8) -- (10, -8);
\draw [thin, line width = 0.5pt] (7, -8) -- (8, -8);
\draw [thin, line width = 0.5pt] (8, -8) -- (9, -8);
\draw [densely dashed, line width = 1pt] (6, -8) -- (6,-9.5);
\draw [thick, line width = 1pt] (6, -9.5) -- (6, -10.5);
\draw [densely dashed, line width = 1pt] (9, -8) -- (9.5, -9.5);
\draw [thick, line width = 1pt] (9.5, -9.5) -- (9.5, -10.5);
\draw [densely dashed, line width = 1pt] (7, -8) -- (7,-9.5);
\draw [thin, line width = 0.5pt] (7, -9.5) -- (7, -10.5);
\draw [thick, line width = 1pt] (7, -10.5) -- (6.5, -10.5);
\draw [densely dashed, line width = 1pt] (8, -8) -- (8,-9.5);
\draw [thick, line width = 1pt] (8, -9.5) -- (8, -10.5);
\draw [densely dashed, line width = 1pt] (9, -8) -- (9,-9.5);
\draw [thick, line width = 1pt] (9, -9.5) -- (9, -10.5);
\draw [densely dashed, line width = 1pt] (10, -8) -- (10,-9.5);
\draw [thick, line width = 1pt] (10, -9.5) -- (10, -10.5);
\filldraw (6, -8) circle(.15);
\filldraw[fill = red] (7, -8) circle(.15);
\filldraw[fill = white] (8, -8) circle(.15);
\filldraw[fill = gray] (9, -8) circle(.15);
\filldraw (10, -8) circle(.15);
\filldraw (6,-9.5) circle(.15);
\filldraw (6,-10.5) circle(.15);
\filldraw (6.5,-10.5) circle(.15);
\filldraw (9.5,-9.5) circle(.15);
\filldraw (9.5,-10.5) circle(.15);
\filldraw[fill = white] (7,-9.5) circle(.15);
\filldraw (7,-10.5) circle(.15);
\filldraw (8,-9.5) circle(.15);
\filldraw (8,-10.5) circle(.15);
\filldraw (9,-9.5) circle(.15);
\filldraw (9,-10.5) circle(.15);
\filldraw (10,-9.5) circle(.15);
\filldraw (10,-10.5) circle(.15);
\node[font=\fontsize{27}{6}\selectfont] at (11,-9.3) {$\frac {16}{14}$};

\draw [thick, line width = 1pt] (13, -8) -- (14, -8);
\draw [thick, line width = 1pt] (16, -8) -- (17, -8);
\draw [thin, line width = 0.5pt] (14, -8) -- (15, -8);
\draw [thin, line width = 0.5pt] (15, -8) -- (16, -8);
\draw [densely dashed, line width = 1pt] (13, -8) -- (13,-9.5);
\draw [thick, line width = 1pt] (13, -9.5) -- (13, -10.5);
\draw [densely dashed, line width = 1pt] (14, -8) -- (14.5, -9.5);
\draw [thick, line width = 1pt] (14.5, -9.5) -- (14.5, -10.5);
\draw [densely dashed, line width = 1pt] (14, -8) -- (14,-9.5);
\draw [thin, line width = 0.5pt] (14, -9.5) -- (14, -10.5);
\draw [thick, line width = 1pt] (14, -10.5) -- (13.5, -10.5);
\draw [densely dashed, line width = 1pt] (15, -8) -- (15,-9.5);
\draw [thick, line width = 1pt] (15, -9.5) -- (15, -10.5);
\draw [densely dashed, line width = 1pt] (16, -8) -- (16,-9.5);
\draw [thick, line width = 1pt] (16, -9.5) -- (16, -10.5);
\draw [densely dashed, line width = 1pt] (17, -8) -- (17,-9.5);
\draw [thick, line width = 1pt] (17, -9.5) -- (17, -10.5);
\filldraw (13, -8) circle(.15);
\filldraw[fill = gray] (14, -8) circle(.15);
\filldraw[fill = white] (15, -8) circle(.15);
\filldraw[fill = red] (16, -8) circle(.15);
\filldraw (17, -8) circle(.15);
\filldraw (13,-9.5) circle(.15);
\filldraw (13,-10.5) circle(.15);
\filldraw(13.5,-10.5) circle(.15);
\filldraw (14.5,-9.5) circle(.15);
\filldraw (14.5,-10.5) circle(.15);
\filldraw[fill = white] (14,-9.5) circle(.15);
\filldraw (14,-10.5) circle(.15);
\filldraw (15,-9.5) circle(.15);
\filldraw (15,-10.5) circle(.15);
\filldraw (16,-9.5) circle(.15);
\filldraw (16,-10.5) circle(.15);
\filldraw (17,-9.5) circle(.15);
\filldraw (17,-10.5) circle(.15);
\node[font=\fontsize{27}{6}\selectfont] at (18,-9.3) {$\frac {16}{14}$};

\end{tikzpicture}
\end{minipage}}
\end{center}
\captionsetup{width=.9\linewidth}
\caption{The possible structures for a responsible but not critical component $K$ of $H+C$, where thick 
(respectively, dashed) edges are in the matching $M$ (respectively, the path-cycle cover $C$), 
thin edges are not in $M \cup C$, 
the filled (respectively, blank) vertices are in (respectively, not in) $V(M)$, 
gray vertices are $2$-anchors, red vertices are responsible $1$-anchors (cf. Definition~\ref{def02}), and the fraction on the right side of
each structure is $\frac{s(K)}{opt(K)}$. 
\label{fig02}}
\end{figure}

\subsection{Operations for modifying critical components}\label{subsec:op}
In this subsection, we define three operations for modifying $C$ 
(and accordingly one or more critical components of $H+C$) 
so that after the modification, $H+C$ will hopefully have fewer critical components. 
Let $v$ be a vertex of a satellite element $S$ in a critical component $K$ and $v'$ be a vertex of $K'$ in $H+C$.
We remark that $K$ and $K'$ may be the same.
Suppose $\{v,v'\} \in E(G)\setminus C$ and we design the following three operations.

\begin{operation}   
\label{op01}     
Suppose one the following two conditions is satisfied:
\begin{itemize}
\parskip=0pt
\item $v'$ is a $0$-anchor of $K$ or

\item $v'$ is a $1$-anchor of $K$ and modifying $C$ by replacing the rescue-edge of $S$ with the edge $\{v,v'\}$ decreases the number of critical components in $H+C$. 
\end{itemize}
Then, the operation modifies $C$ by replacing the rescue-edges of $S$ with the edge $\{v,v'\}$.  
(cf. Figure~\ref{fig03})
\end{operation}

Clearly, Operation~\ref{op01} does not change the weight of $C$ by the first statements in Fact~\ref{fact02}. 
Suppose $v'$ is a $0$-anchor.
If $K=K'$, then after Operation~\ref{op01}, $K$ is no longer critical by the second statement of Fact~\ref{fact02}. 
Then, we suppose $K\ne K'$.
Obviously, $K$ is no longer critical but $K'$ may become critical after Operation~\ref{op01}.
So, Operation~\ref{op01} may not necessarily decrease but does not increase the number of critical components in $H+C$. 
Fortunately, Operation~\ref{op01} changes $v'$ from a $0$-anchor to a $1$-anchor.
So, Operation~\ref{op01} decreases the number of $0$-anchors in $H+C$ by~$1$ or the number of critical components in $H+C$ by~$1$.
Obviously, Operation~\ref{op01} does not change the number of components in $H+C$.

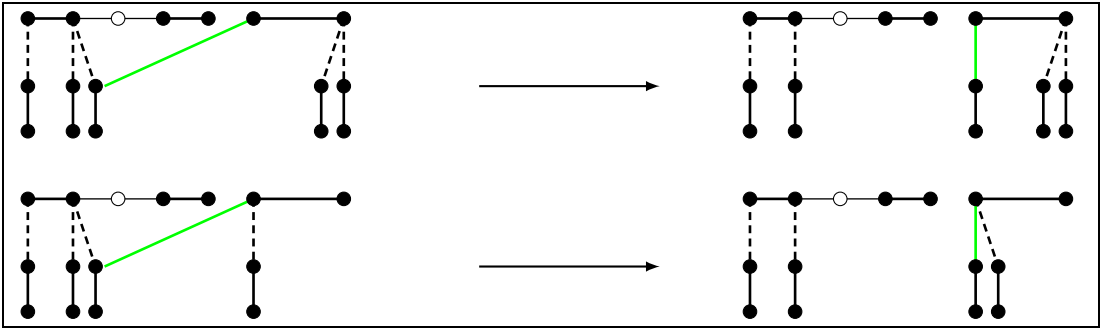
\begin{figure}[thb]
\begin{center}
\framebox{
\begin{minipage}{5.6in}
\begin{tikzpicture}[scale=0.6,transform shape]


\draw [thick, line width = 1pt] (-8, -6) -- (-7, -6);
\draw [thick, line width = 1pt] (-5, -6) -- (-4, -6);
\draw [thin, line width = 0.5pt] (-7, -6) -- (-6, -6);
\draw [thin, line width = 0.5pt] (-6, -6) -- (-5, -6);
\draw [densely dashed, line width = 1pt] (-8, -6) -- (-8,-7.5);
\draw [thick, line width = 1pt] (-8, -7.5) -- (-8, -8.5);
\draw [densely dashed, line width = 1pt] (-7, -6) -- (-7, -7.5);
\draw [thick, line width = 1pt] (-7, -7.5) -- (-7, -8.5);
\draw [densely dashed, line width = 1pt] (-7, -6) -- (-6.5,-7.5);
\draw [thick, line width = 1pt] (-6.5, -7.5) -- (-6.5, -8.5);
\filldraw (-8, -6) circle(.15);
\filldraw (-7, -6) circle(.15);
\filldraw[fill = white] (-6, -6) circle(.15);
\filldraw (-5, -6) circle(.15);
\filldraw (-4, -6) circle(.15);
\filldraw (-8,-7.5) circle(.15);
\filldraw (-8,-8.5) circle(.15);
\filldraw (-7,-7.5) circle(.15);
\filldraw (-7,-8.5) circle(.15);
\filldraw (-6.5,-7.5) circle(.15);
\filldraw (-6.5,-8.5) circle(.15);

\draw [thick, line width = 1pt, color=green] (-6.3, -7.5) -- (-3, -6);

\draw [thick, line width = 1pt] (-3, -6) -- (-1, -6);
\draw [densely dashed, line width = 1pt] (-1, -6) -- (-1,-7.5);
\draw [thick, line width = 1pt] (-1, -7.5) -- (-1, -8.5);
\draw [densely dashed, line width = 1pt] (-1, -6) -- (-1.5,-7.5);
\draw [thick, line width = 1pt] (-1.5, -7.5) -- (-1.5, -8.5);
\filldraw (-3, -6) circle(.15);
\filldraw (-1, -6) circle(.15);
\filldraw (-1,-7.5) circle(.15);
\filldraw (-1,-8.5) circle(.15);
\filldraw (-1.5,-7.5) circle(.15);
\filldraw (-1.5,-8.5) circle(.15);

\draw [-latex, thick] (2, -7.5) to (6, -7.5);

\draw [thick, line width = 1pt] (8, -6) -- (9, -6);
\draw [thick, line width = 1pt] (11, -6) -- (12,-6);
\draw [thin, line width = 0.5pt] (9, -6) -- (10, -6);
\draw [thin, line width = 0.5pt] (10, -6) -- (11, -6);
\draw [densely dashed, line width = 1pt] (8, -6) -- (8, -7.5);
\draw [thick, line width = 1pt] (8, -7.5) -- (8, -8.5);
\draw [densely dashed, line width = 1pt] (9, -6) -- (9, -7.5);
\draw [thick, line width = 1pt] (9, -7.5) -- (9, -8.5);
\filldraw (8, -6) circle(.15);
\filldraw (9, -6) circle(.15);
\filldraw[fill = white] (10, -6) circle(.15);
\filldraw (11, -6) circle(.15);
\filldraw (12, -6) circle(.15);
\filldraw (8, -7.5) circle(.15);
\filldraw (8, -8.5) circle(.15);
\filldraw (9, -7.5) circle(.15);
\filldraw (9, -8.5) circle(.15);

\draw [thick, line width = 1pt] (13, -6) -- (15, -6);
\draw [thick, line width = 1pt, color=green] (13, -6) -- (13,-7.5);
\draw [thick, line width = 1pt] (13, -7.5) -- (13, -8.5);
\draw [densely dashed, line width = 1pt] (15, -6) -- (15,-7.5);
\draw [thick, line width = 1pt] (15, -7.5) -- (15, -8.5);
\draw [densely dashed, line width = 1pt] (15, -6) -- (14.5,-7.5);
\draw [thick, line width = 1pt] (14.5, -7.5) -- (14.5, -8.5);
\filldraw (13, -6) circle(.15);
\filldraw (15, -6) circle(.15);
\filldraw (13,-7.5) circle(.15);
\filldraw (13,-8.5) circle(.15);
\filldraw (14.5,-7.5) circle(.15);
\filldraw (14.5,-8.5) circle(.15);
\filldraw (15,-7.5) circle(.15);
\filldraw (15,-8.5) circle(.15);


\draw [thick, line width = 1pt] (-8, -10) -- (-7, -10);
\draw [thick, line width = 1pt] (-5, -10) -- (-4, -10);
\draw [thin, line width = 0.5pt] (-7, -10) -- (-6, -10);
\draw [thin, line width = 0.5pt] (-6, -10) -- (-5, -10);
\draw [densely dashed, line width = 1pt] (-8, -10) -- (-8,-11.5);
\draw [thick, line width = 1pt] (-8, -11.5) -- (-8, -12.5);
\draw [densely dashed, line width = 1pt] (-7, -10) -- (-7, -11.5);
\draw [thick, line width = 1pt] (-7, -11.5) -- (-7, -12.5);
\draw [densely dashed, line width = 1pt] (-7, -10) -- (-6.5,-11.5);
\draw [thick, line width = 1pt] (-6.5, -11.5) -- (-6.5, -12.5);
\filldraw (-8, -10) circle(.15);
\filldraw (-7, -10) circle(.15);
\filldraw[fill = white] (-6, -10) circle(.15);
\filldraw (-5, -10) circle(.15);
\filldraw (-4, -10) circle(.15);
\filldraw (-8,-11.5) circle(.15);
\filldraw (-8,-12.5) circle(.15);
\filldraw (-7,-11.5) circle(.15);
\filldraw (-7,-12.5) circle(.15);
\filldraw (-6.5,-11.5) circle(.15);
\filldraw (-6.5,-12.5) circle(.15);

\draw [thick, line width = 1pt, color=green] (-6.3, -11.5) -- (-3, -10);

\draw [thick, line width = 1pt] (-3, -10) -- (-1, -10);
\draw [densely dashed, line width = 1pt] (-3, -10) -- (-3,-11.5);
\draw [thick, line width = 1pt] (-3, -11.5) -- (-3, -12.5);
\filldraw (-3, -10) circle(.15);
\filldraw (-1, -10) circle(.15);
\filldraw (-3,-11.5) circle(.15);
\filldraw (-3,-12.5) circle(.15);

\draw [-latex, thick] (2, -11.5) to (6, -11.5);

\draw [thick, line width = 1pt] (8, -10) -- (9, -10);
\draw [thick, line width = 1pt] (11, -10) -- (12,-10);
\draw [thin, line width = 0.5pt] (9, -10) -- (10, -10);
\draw [thin, line width = 0.5pt] (10, -10) -- (11, -10);
\draw [densely dashed, line width = 1pt] (8, -10) -- (8, -11.5);
\draw [thick, line width = 1pt] (8, -11.5) -- (8, -12.5);
\draw [densely dashed, line width = 1pt] (9, -10) -- (9, -11.5);
\draw [thick, line width = 1pt] (9, -11.5) -- (9, -12.5);
\filldraw (8, -10) circle(.15);
\filldraw (9, -10) circle(.15);
\filldraw[fill = white] (10, -10) circle(.15);
\filldraw (11, -10) circle(.15);
\filldraw (12, -10) circle(.15);
\filldraw (8, -11.5) circle(.15);
\filldraw (8, -12.5) circle(.15);
\filldraw (9, -11.5) circle(.15);
\filldraw (9, -12.5) circle(.15);

\draw [thick, line width = 1pt] (13, -10) -- (15, -10);
\draw [thick, line width = 1pt, color=green] (13, -10) -- (13,-11.5);
\draw [thick, line width = 1pt] (13, -11.5) -- (13, -12.5);
\draw [densely dashed, line width = 1pt] (13, -10) -- (13.5,-11.5);
\draw [thick, line width = 1pt] (13.5, -11.5) -- (13.5, -12.5);
\filldraw (13, -10) circle(.15);
\filldraw (15, -10) circle(.15);
\filldraw (13,-11.5) circle(.15);
\filldraw (13,-12.5) circle(.15);
\filldraw (13.5,-11.5) circle(.15);
\filldraw (13.5,-12.5) circle(.15);

\end{tikzpicture}
\end{minipage}}
\end{center}
\captionsetup{width=1.0\linewidth}
\caption{Two representative possible cases in Operation~\ref{op01}, 
where the first (respectively, second) satisfies $v$ is a $0$-anchor (respectively, $1$-anchor) and the green edge is the edge $\{v,v'\}$. 
\label{fig03}}
\end{figure}

\begin{operation}
\label{op02}         
Suppose $v'$ is in a satellite-element $S'$ of $K'$ and 
the center element $K'_c$ of $K'$ is an edge or a star to which no satellite element other than $S'$ is adjacent in $H+C$.
Then (cf. Figure~\ref{fig04}), the operation modifies $C$ by replacing the rescue-edge of $S$ with the edge $\{v,v'\}$. 
\end{operation}

Obviously, Operation~\ref{op02} does not change the weight of $C$ by the first statement of Fact~\ref{fact02}.
Note that $K'$ has no $2$-anchor and hence $K'$ is not critical by Lemma~\ref{lemma10}.
So, $K\ne K'$ since $K$ is critical.
Moreover, after Operation~\ref{op02}, $S'$ becomes the center element of $K'$ and hence Lemma~\ref{lemma06} still holds. 
Furthermore, by the first statement in Fact~\ref{fact02} and Lemma~\ref{lemma12}, 
$K, K'$ are not critical after Operation~\ref{op02} and thus Operation~\ref{op02} decreases the number of critical components in $H+C$ by~$1$. 
Clearly, Operation~\ref{op02} does not change the number of components in $H+C$. 
Before Operation~\ref{op02}, $K'$ may have one $0$-anchor $x$.
After Operation~\ref{op02}, $x$ will be in a satellite element of $H+C$ and hence will not be a $0$-anchor, 
but $S'$ will become a center element with two satellite elements adjacent to it in $H+C$, 
implying that one vertex of $S'$ may become a $0$-anchor in $H+C$ (or not an anchor, if $S'$ is a star).
In summary, Operation~\ref{op02} does not increase the number of $0$-anchors in $H+C$.

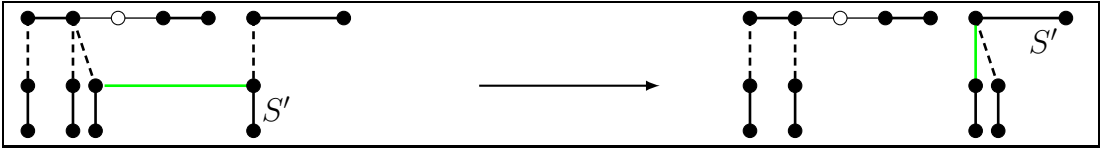
\begin{figure}[thb]
\begin{center}
\framebox{
\begin{minipage}{5.6in}
\begin{tikzpicture}[scale=0.6,transform shape]

\draw [thick, line width = 1pt] (-8, -2) -- (-7, -2);
\draw [thick, line width = 1pt] (-5, -2) -- (-4, -2);
\draw [thin, line width = 0.5pt] (-7, -2) -- (-6, -2);
\draw [thin, line width = 0.5pt] (-6, -2) -- (-5, -2);
\draw [densely dashed, line width = 1pt] (-8, -2) -- (-8,-3.5);
\draw [thick, line width = 1pt] (-8, -3.5) -- (-8, -4.5);
\draw [densely dashed, line width = 1pt] (-7, -2) -- (-7, -3.5);
\draw [thick, line width = 1pt] (-7, -3.5) -- (-7, -4.5);
\draw [densely dashed, line width = 1pt] (-7, -2) -- (-6.5,-3.5);
\draw [thick, line width = 1pt] (-6.5, -3.5) -- (-6.5, -4.5);
\filldraw (-8, -2) circle(.15);
\filldraw (-7, -2) circle(.15);
\filldraw[fill = white] (-6, -2) circle(.15);
\filldraw (-5, -2) circle(.15);
\filldraw (-4, -2) circle(.15);
\filldraw (-8,-3.5) circle(.15);
\filldraw (-8,-4.5) circle(.15);
\filldraw (-7,-3.5) circle(.15);
\filldraw (-7,-4.5) circle(.15);
\filldraw (-6.5,-3.5) circle(.15);
\filldraw (-6.5,-4.5) circle(.15);

\draw [thick, line width = 1pt, color=green] (-6.3, -3.5) -- (-3, -3.5);

\draw [thick, line width = 1pt] (-3, -2) -- (-1, -2);
\draw [densely dashed, line width = 1pt] (-3, -2) -- (-3,-3.5);
\draw [thick, line width = 1pt] (-3, -3.5) -- (-3, -4.5);
\filldraw (-3, -2) circle(.15);
\filldraw (-1, -2) circle(.15);
\filldraw (-3,-3.5) circle(.15);
\filldraw (-3,-4.5) circle(.15);
\node[font=\fontsize{20}{6}\selectfont] at (-2.5,-4) {$S'$};

\draw [-latex, thick] (2, -3.5) to (6, -3.5);

\draw [thick, line width = 1pt] (8, -2) -- (9, -2);
\draw [thick, line width = 1pt] (11, -2) -- (12,-2);
\draw [thin, line width = 0.5pt] (9, -2) -- (10, -2);
\draw [thin, line width = 0.5pt] (10, -2) -- (11, -2);
\draw [densely dashed, line width = 1pt] (8, -2) -- (8, -3.5);
\draw [thick, line width = 1pt] (8, -3.5) -- (8, -4.5);
\draw [densely dashed, line width = 1pt] (9, -2) -- (9, -3.5);
\draw [thick, line width = 1pt] (9, -3.5) -- (9, -4.5);
\filldraw (8, -2) circle(.15);
\filldraw (9, -2) circle(.15);
\filldraw[fill = white] (10, -2) circle(.15);
\filldraw (11, -2) circle(.15);
\filldraw (12, -2) circle(.15);
\filldraw (8, -3.5) circle(.15);
\filldraw (8, -4.5) circle(.15);
\filldraw (9, -3.5) circle(.15);
\filldraw (9, -4.5) circle(.15);

\draw [thick, line width = 1pt] (13, -2) -- (15, -2);
\draw [thick, line width = 1pt, color=green] (13, -2) -- (13,-3.5);
\draw [thick, line width = 1pt] (13, -3.5) -- (13, -4.5);
\draw [densely dashed, line width = 1pt] (13, -2) -- (13.5,-3.5);
\draw [thick, line width = 1pt] (13.5, -3.5) -- (13.5, -4.5);
\filldraw (13, -2) circle(.15);
\filldraw (15, -2) circle(.15);
\filldraw (13,-3.5) circle(.15);
\filldraw (13,-4.5) circle(.15);
\filldraw (13.5,-3.5) circle(.15);
\filldraw (13.5,-4.5) circle(.15);
\node[font=\fontsize{20}{6}\selectfont] at (14.5,-2.5) {$S'$};

\end{tikzpicture}
\end{minipage}}
\end{center}
\captionsetup{width=1.0\linewidth}
\caption{A representative possible case in Operation~\ref{op02} and the green edge is the edge $\{v,v'\}$. 
\label{fig04}}
\end{figure}

\begin{operation}   
\label{op03}     
Suppose $v'$ appears in a satellite-element $S'$ of $K'$ and $K'_c$ is a $5$-path or
$K'_c$ is an edge or a star to which 
at least one more satellite element other than $S'$ is adjacent in $H+C$.
Then (cf. Figure~\ref{fig05}), the operation modifies $C$ by replacing the rescue-edges of $S$ and $S'$ with the edge $\{v,v'\}$. 
\end{operation}

By the first statement of Fact~\ref{fact02}, Operation~\ref{op03} does not change the weight of $C$ since $K, K'$ will not be an isolated bad component of $H$.

Operation~\ref{op03} uses the edge $\{v,v'\}$ to connect $S$ and $S'$ into a new composite component $K_{new}$ of $H+C$. 
Since both $S$ and $S'$ are not $5$-paths by the first statement of Lemma~\ref{lemma07}, 
either of them can be treated as the center element of $K_{new}$ and the other becomes the satellite element of $K_{new}$. 
Note that $K_{new}$ has at most one $0$-anchor, and the rescue-anchor of $S'$ may be a $1$-anchor before Operation~\ref{op03}. 
So, Operation~\ref{op03} increases the number of $0$-anchors in $H+C$ by at most $2$. 

By Lemma~\ref{lemma10}, $K_{new}$ is not critical. 
If $K=K'$, then clearly Operation~\ref{op03} does not increase the number of critical components in $H+C$.
Otherwise, Operation~\ref{op03} makes $K$ not critical because of the first statement in  Fact~\ref{fact02}, 
but it is possible that Operation~\ref{op03} makes $K'$ critical. 
In any case, Operation~\ref{op03} does not increase the number of critical components in $H+C$. 
Luckily, Operation~\ref{op03} always increases the number of components in $H+C$ by~$1$.

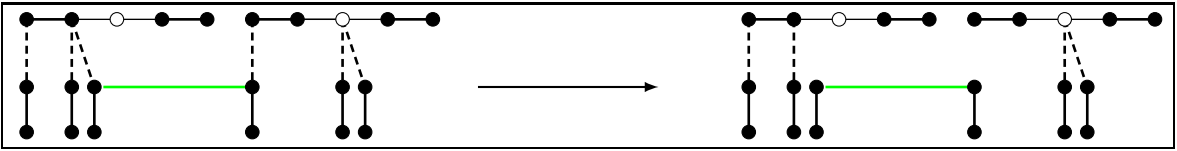
\begin{figure}[thb]
\begin{center}
\framebox{
\begin{minipage}{6in}
\begin{tikzpicture}[scale=0.6,transform shape]

\draw [thick, line width = 1pt] (-8, 2) -- (-7, 2);
\draw [thick, line width = 1pt] (-5, 2) -- (-4, 2);
\draw [thin, line width = 0.5pt] (-7, 2) -- (-6, 2);
\draw [thin, line width = 0.5pt] (-6, 2) -- (-5, 2);
\draw [densely dashed, line width = 1pt] (-8, 2) -- (-8,0.5);
\draw [thick, line width = 1pt] (-8, 0.5) -- (-8, -0.5);
\draw [densely dashed, line width = 1pt] (-7, 2) -- (-7, 0.5);
\draw [thick, line width = 1pt] (-7, 0.5) -- (-7, -0.5);
\draw [densely dashed, line width = 1pt] (-7, 2) -- (-6.5,0.5);
\draw [thick, line width = 1pt] (-6.5, 0.5) -- (-6.5, -0.5);
\filldraw (-8, 2) circle(.15);
\filldraw (-7, 2) circle(.15);
\filldraw[fill = white] (-6, 2) circle(.15);
\filldraw (-5, 2) circle(.15);
\filldraw (-4, 2) circle(.15);
\filldraw (-8,0.5) circle(.15);
\filldraw (-8,-0.5) circle(.15);
\filldraw (-7,0.5) circle(.15);
\filldraw (-7,-0.5) circle(.15);
\filldraw (-6.5,0.5) circle(.15);
\filldraw (-6.5,-0.5) circle(.15);

\draw [thick, line width = 1pt, color=green] (-6.3, 0.5) -- (-3, 0.5);

\draw [thick, line width = 1pt] (-3, 2) -- (-2, 2);
\draw [thick, line width = 1pt] (0, 2) -- (1, 2);
\draw [thin, line width = 0.5pt] (-2, 2) -- (-1, 2);
\draw [thin, line width = 0.5pt] (-2, 2) -- (0, 2);
\draw [densely dashed, line width = 1pt] (-3, 2) -- (-3,0.5);
\draw [thick, line width = 1pt] (-3, 0.5) -- (-3, -0.5);
\draw [densely dashed, line width = 1pt] (-1, 2) -- (-1, 0.5);
\draw [thick, line width = 1pt] (-1, 0.5) -- (-1, -0.5);
\draw [densely dashed, line width = 1pt] (-1, 2) -- (-0.5,0.5);
\draw [thick, line width = 1pt] (-0.5, 0.5) -- (-0.5, -0.5);
\filldraw (-3, 2) circle(.15);
\filldraw (-2, 2) circle(.15);
\filldraw[fill = white] (-1, 2) circle(.15);
\filldraw (0, 2) circle(.15);
\filldraw (1, 2) circle(.15);
\filldraw (-3,0.5) circle(.15);
\filldraw (-3,-0.5) circle(.15);
\filldraw (-1,0.5) circle(.15);
\filldraw (-1,-0.5) circle(.15);
\filldraw (-0.5,0.5) circle(.15);
\filldraw (-0.5,-0.5) circle(.15);

\draw [-latex, thick] (2, 0.5) to (6, 0.5);

\draw [thick, line width = 1pt] (8, 2) -- (9, 2);
\draw [thick, line width = 1pt] (11, 2) -- (12, 2);
\draw [thin, line width = 0.5pt] (9, 2) -- (10, 2);
\draw [thin, line width = 0.5pt] (10, 2) -- (11, 2);
\draw [densely dashed, line width = 1pt] (8, 2) -- (8, 0.5);
\draw [thick, line width = 1pt] (8, 0.5) -- (8, -0.5);
\draw [densely dashed, line width = 1pt] (9, 2) -- (9, 0.5);
\draw [thick, line width = 1pt] (9, 0.5) -- (9, -0.5);
\draw [thick, line width = 1pt] (9.5, 0.5) -- (9.5, -0.5);
\filldraw (8, 2) circle(.15);
\filldraw (9, 2) circle(.15);
\filldraw[fill = white] (10, 2) circle(.15);
\filldraw (11, 2) circle(.15);
\filldraw (12, 2) circle(.15);
\filldraw (8, 0.5) circle(.15);
\filldraw (8, -0.5) circle(.15);
\filldraw (9, 0.5) circle(.15);
\filldraw (9, -0.5) circle(.15);
\filldraw (9.5, 0.5) circle(.15);
\filldraw (9.5, -0.5) circle(.15);

\draw [thick, line width = 1pt, color=green] (9.7, 0.5) -- (13, 0.5);

\draw [thick, line width = 1pt] (13, 2) -- (14, 2);
\draw [thick, line width = 1pt] (16, 2) -- (17, 2);
\draw [thin, line width = 0.5pt] (14, 2) -- (15, 2);
\draw [thin, line width = 0.5pt] (15, 2) -- (16, 2);
\draw [densely dashed, line width = 1pt] (15, 2) -- (15, 0.5);
\draw [thick, line width = 1pt] (15, 0.5) -- (15, -0.5);
\draw [densely dashed, line width = 1pt] (15, 2) -- (15.5, 0.5);
\draw [thick, line width = 1pt] (15.5, 0.5) -- (15.5, -0.5);
\draw [thick, line width = 1pt] (13, 0.5) -- (13, -0.5);
\filldraw (13, 2) circle(.15);
\filldraw (14, 2) circle(.15);
\filldraw[fill = white] (15, 2) circle(.15);
\filldraw (16, 2) circle(.15);
\filldraw (17, 2) circle(.15);
\filldraw (15, 0.5) circle(.15);
\filldraw (15, -0.5) circle(.15);
\filldraw (15.5, 0.5) circle(.15);
\filldraw (15.5, -0.5) circle(.15);
\filldraw (13, 0.5) circle(.15);
\filldraw (13, -0.5) circle(.15);

\end{tikzpicture}
\end{minipage}}
\end{center}
\captionsetup{width=1.0\linewidth}
\caption{A representative possible case in Operation~\ref{op03} and the green edge is the edge $\{v,v'\}$. 
\label{fig05}}
\end{figure}

\begin{lemma}
\label{lemma15}
If $G$ has an edge $\{v,v'\}$ such that $v$ is in a critical satellite-element $S$ of $H+C$ and 
$v'\notin V(S)$ is neither a $2$-anchor nor a responsible $1$-anchor in $H+C$,
then one of Operations~\ref{op01}, \ref{op02}, and~\ref{op03} is applicable.
\end{lemma}
\begin{proof} 
Suppose $v'$ belongs to a satellite element $S'$ of $K'$.
If $K'_c$ is an edge or a star and $S'$ is the unique satellite element of $K'$,
then Operation~\ref{op02} is applicable.
Otherwise, $K'_c$ is a $5$-path or $K'_c$ is an edge or a star and $K'$ has at least two satellite elements by Lemma~\ref{lemma07}.
So, Operation~\ref{op03} is applicable.

We next assume $v'$ belongs to the center element $K'_c$ of $K'$.
By the second statement of Lemma~\ref{lemma07} and Definition~\ref{def07}, $v'$ is an anchor of $K'_c$.
Note that $v'$ is neither a $2$-anchor nor a responsible $1$-anchor.
So, $v'$ is a $0$-anchor or $v'$ is a $1$-anchor such that moving $S$ to $K'$ will not make $K'$ critical.
It follows that Operation~\ref{op01} is applicable again, which completes the proof.
\end{proof}

\begin{lemma}
\label{lemma16}
If we repeatedly perform Operations~\ref{op01},~\ref{op02}, and~\ref{op03} until none of them is applicable, 
then the number of repetitions is bounded by $O(n)$. 
\end{lemma}
\begin{proof} 
Let $n_0$ be the number of $0$-anchors in $H+C$, $n_c$ be the number of components in $H+C$, and 
$n_{cc}$ be the number of critical components in $H+C$. 
Let $g = n_0 + n_{cc} - 3n_c$. 
Recall that Operation~\ref{op01} decreases $n_0$ or $n_{cc}$ by~$1$ but does not change $n_c$ and 
Operation~\ref{op02} decreases $n_{cc}$ by~$1$, does not change $n_c$, and does not increase $n_0$. 
As for Operation~\ref{op03}, it increases $n_c$ by~$1$, increases $n_0$ at most $2$, but does not 
increase $n_{cc}$. 
So, each time we perform one of the operations, we decrease the value of $g$ by at least~$1$. 
Note that $2n \ge g \ge -3n$. 
Therefore, if we repeatedly perform the operations until none of them is applicable, 
then the number of repetitions is at most $O(n)$.
\end{proof}

\subsection{Bounding $opt(G)$}\label{subsec:bound}
\begin{notation}
\label{nota08}
Let $R$ denote the set of vertices $v \in V(H)$ such that $v$ is a $2$-anchor or a responsible $1$-anchor in $H+C$.
\end{notation}

Obviously, if $K$ is an isolated $5$-path of $H+C$, then each vertex of $K$ is a $0$-anchor and hence $|R \cap V(K)| = 0$. 
Moreover, for each composite component $K$ of $H+C$, $|R \cap V(K)|$ is bounded by the number 
of anchors in $K$. 
Thus, if the center element $K_c$ of $K$ is an edge, then $|R \cap V(K)| \in \{0, 1, 2\}$. 
If $K_c$ is a star, then $|R \cap V(K)| \in \{0, 1\}$.
Similarly, if $K_c$ is a $5$-path, then $|R \cap V(K)| \in \{0, 1, 2, 3, 4, 5\}$. 

\begin{notation}
\label{nota09}
For the components in $H+C$, we define the notations as follows.
\begin{itemize}
\parskip=0pt
\item Let ${\cal K}$ be the set of components of $H+C$ that are not isolated bad components of $H$.

\item
For each $i\in\{0,1,2,3,4,5\}$, let ${\cal K}_i \subseteq {\cal K}$ be a subset of ${\cal K}$
such that $|R \cap V(K)| = i$. 

\item
	For each $i\in\{1,2\}$, let ${\cal K}_{i,c}$ be the set of critical components in ${\cal K}_i$. 

	({\em Comment}: By Figure~\ref{fig01}, each critical component $K$ of $H+C$ has one or two $2$-anchors and no responsible $1$-anchor since $K$ is not responsible.) 

\item
Let $R_c$ be the set of $2$-anchors in the critical components of $H+C$. 

\item
$U_c = \bigcup_{v\in R_c} \{w\in V(H) \mid w$ appears in a critical satellite-element whose rescue-anchor is $v\}$. 

\item Let $G_c = G[V(G) \setminus (R_c\cup U_c)]$. 
\end{itemize}
\end{notation}

\begin{lemma}
\label{lemma17} 
$\displaystyle opt(G) \le opt(G_c) + 7\sum_{i=1}^5 i|{\cal K}_i|.$
\end{lemma}
\begin{proof}
Consider a critical satellite-element $S$ whose rescue-anchor is in $R_c$. 
Since $S$ is a bad component of $H$ by Lemma~\ref{lemma07}, $G[V(S)]$ contains no $4^+$-path. 
Moreover, if a vertex $v\not\in V(S)$ is adjacent to $S$ in $G$, then $v \in R$ by Lemma~\ref{lemma15} and none of Operations~\ref{op01},~\ref{op02} and~\ref{op03} is applicable.
Thus, removing the vertices of $R_c\cup U_c$ from $G$ destroys at most $|R|$ paths of $OPT(G)$. 
By Fact \ref{fact01}, each path has at most seven vertices. 
Moreover, each un-destroyed path of $OPT(G)$ still has at least four vertices. 
Hence, $opt(G_c) \ge opt(G) - 7|R|$. 
Because $|R| = \sum_{i=1}^5 i|{\cal K}_i|$, the lemma holds. 
\end{proof}

\subsection{Summary of the algorithm}\label{subsec:algo}
Let $r = \frac {15 + \sqrt{505}}{20} \approx 1.874$ be the positive root to the quadratic equation $10r^2 - 15r - 7 = 0$.
Our algorithm proceeds as follows.

\begin{enumerate}
\parskip=0pt
\setcounter{enumi}{-1}
\item
If $|V(G)| \le 4$, then find an optimal solution by brute-force search, output it, and then halt.
\item Construct the graph $H$ as follows:
\begin{enumerate}
\parskip=0pt
\item
Compute a maximum matching $M$ in $G$ and initialize $H$ to be the graph $(V(M), M)$. 
\item
Modify $M$ and $H$ by performing Steps~1.1 and~1.2 in Section~\ref{subsec:k=4}. 
\end{enumerate}

\item Compute a maximum path-cycle cover $C$ and modify it as follows:
\begin{enumerate}
\parskip=0pt
\item 
Perform Steps~2.1, 2.2, and~2.3 in Section~\ref{subsec:rescue} to compute a maximum path-cycle cover $C$ of edges $\{v,w\}\in E(G)$ 
such that $v$ and $w$ are in different components of $H$ and at least one of them is a bad component.
\end{enumerate}

\item 
Repeatedly perform Operations~\ref{op01},~\ref{op02}, and~\ref{op03} in Section~\ref{subsec:op} to modify $C$, 
until none of them is applicable. 

\item
If no component of $H+C$ is critical, or 
$\frac {\sum_{i=1}^5 i|{\cal K}_i|}{|{\cal K}_{1,c}|+2|{\cal K}_{2,c}|} > \frac {5}7r$, then
\begin{enumerate}
\parskip=0pt
\item
	compute $OPT(K)$ for each component $K$ of $H+C$ that is not an isolated bad component of $H$;
\item
	output their union as a solution for $G$, and then halt. 
\end{enumerate}

\item Otherwise, there exists at least one critical component in $H+C$ and 
$\frac {\sum_{i=1}^5 i|{\cal K}_i|}{|{\cal K}_{1,c}|+2|{\cal K}_{2,c}|} \le \frac {5}7r$.
\begin{enumerate}
\parskip=0pt
\item 
Recursively call the algorithm on the graph $G_c$ to obtain a solution $ALG(G_c)$. 

\item 
For each $v \in R_c$, compute a $5^+$-path $P_v$ since $v$ is an $2$-anchor.

\item 
Output the union of $ALG(G_c)$ and $\cup_{v\in R_c} P_v$, and halt.
\end{enumerate}
\end{enumerate}

\section{Analyzing the performance}\label{sec:ana4}
In this section, we show that the approximation ratio achieved by our algorithm is at most 
$r = \frac {15 + \sqrt{505}}{20} \approx 1.874$. 
For brevity, we first define several notations. 

\begin{notation}
\label{nota10}
We define the notation $\preceq$ as follows:
\begin{itemize}
\parskip=0pt
\item
	For any $K\in{\cal K}$ and a rational fraction $\frac{a}{b}$, we write $\frac{s(K)}{opt(K)} \preceq \frac{a}{b}$,
	whenever $opt(K) \ge b$ and $s(K) \le a$.
	
	({\em Comment:} If $\frac {s(K)}{opt(K)}\preceq \frac ab$, then $\frac {s(K)}{opt(K)}\le \frac ab$, but not vice versa.)

\item For a $K\in{\cal K}$ and a set ${\cal F}$ of rational fractions, we write 
$\frac{s(K)}{opt(K)} \preceq {\cal F}$ if $\frac{s(K)}{opt(K)} \preceq \frac{a}{b}$ for all $\frac{a}{b} \in {\cal F}$.

\item For a rational fraction $\frac{a}{b}$ and a set ${\cal F}$ of rational fractions, we write 
${\cal F} \preceq \frac{a}{b}$ if $\frac{c}{d} \le \frac{a}{b}$ for all $\frac{c}{d}\in{\cal F}$.
\end{itemize}
\end{notation}
Obviously, ${\cal K} = \bigcup^5_{i=0}{\cal K}_i$. 
The next fact follows from Lemmas~\ref{lemma10}--\ref{lemma14}.

\begin{fact}
\label{fact03}
The following statements hold:
\begin{enumerate}
\parskip=0pt
\item
For each $K \in {\cal K}_0$, $\frac {s(K)}{opt(K)} < \frac {14}{11}$.

\item 
For each $K \in {\cal K}_{1,c}$, $\frac {s(K)}{opt(K)} \preceq \left\{\frac 85, \frac {10}7\right\}$;
while for each $K \in {\cal K}_1 \setminus {\cal K}_{1,c}$, 
$\frac {s(K)}{opt(K)} \preceq \left\{\frac 65, \frac 87, \frac {10}8, \frac {12}{10}, \frac {14}{12}, \frac {16}{13}\right\}$.

\item 
For each $K \in {\cal K}_{2,c}$, $\frac {s(K)}{opt(K)} \preceq \left\{\frac {16}{12}, \frac {18}{13}, \frac {14}{11}\right\}$;
while for each $K \in {\cal K}_2 \setminus {\cal K}_{2,c}$, 
$\frac {s(K)}{opt(K)} \preceq \left\{\frac 66, \frac {12}{10}, \frac {14}{12}, \frac {16}{13}, \frac {18}{15}\right\}$.

\item 
For each $K \in {\cal K}_3$, $\frac {s(K)}{opt(K)} \preceq \left\{\frac {18}{15}, \frac {20}{17}, \frac {12}{11}\right\}$.

\item 
For each $K \in {\cal K}_4$, $\frac {s(K)}{opt(K)} \preceq \frac {22}{20}$.
\item
For each $K \in {\cal K}_5$, $\frac {s(K)}{opt(K)} \preceq \frac {24}{25}$.
\end{enumerate}
\end{fact}
\begin{proof}
If $K\in {\cal K}_0$, then $K$ has no $2$-anchor.
By Lemma~\ref{lemma10}, $K$ is not critical and hence $\frac {s(K)}{opt(K)}<\frac {14}{11}$.
If $K\in {\cal K}_{1,c}$, then $K$ is critical and $K$ has exactly one $2$-anchor (See the first line of Figure~\ref{fig01}).
So, by Lemma~\ref{lemma13}, $\frac {s(K)}{opt(K)} \preceq \left\{\frac 85, \frac {10}7\right\}$.
If $K\in {\cal K}\setminus {\cal K}_{1,c}$, then $K$ is not critical and $K$ has one $2$-anchor or one responsible $1$-anchor (See the first line of Figure~\ref{fig02}).
So, $s(K)\in \{6, 8, 10, 12, 14, 16\}$.
Since $K$ is not critical, $\frac {s(K)}{opt(K)}<\frac {14}{11}$ and hence 
$\frac {s(K)}{opt(K)} \preceq \left\{\frac 65, \frac 87, \frac {10}8, \frac {12}{10}, \frac {14}{12}, \frac {16}{13}\right\}$.

If $K$ is in ${\cal K}_{2,c}$, then $K$ is critical and $K$ has exactly two $2$-anchors (See Figure~\ref{fig01} except the first line).
So, by Lemma~\ref{lemma12}, $\frac {s(K)}{opt(K)} \preceq \left\{\frac {16}{12}, \frac {18}{13}, \frac {14}{11}\right\}$.
If $K\in {\cal K}\setminus {\cal K}_{2,c}$, then $K$ is not critical.
Furthermore, $K$ has two responsible $1$-anchors; one responsible $1$-anchor and a $2$-anchor or two $2$-anchors.
By the second structure of Figure~\ref{fig02}, $\frac {s(K)}{opt(K)}\le \frac 66$ if $K$ has two responsible $1$-anchors.
If $K$ has a responsible $1$-anchor and a $2$-anchor, then by Figure~\ref{fig02},
$\frac {s(K)}{opt(K)} \preceq \left\{\frac {12}{10}, \frac {14}{12}, \frac {16}{13}\right\}$.
Lastly, $K$ has two $2$-anchors and hence $s(K) \in \{12, 14, 16, 18\}$.
Since $K$ is not critical, we know $\frac {s(K)}{opt(K)} \preceq \left\{\frac {12}{10}, \frac {14}{12}, \frac {16}{13}, \frac {18}{15}\right\}$,
which completes the proof of the third statement.

If $K\in {\cal K}_3$, then by Figure~\ref{fig02}, $K$ has three $2$-anchors or one $2$-anchor and two responsible $1$-anchors.
By Lemma~\ref{lemma11} and the second structure of the fourth line in Figure~\ref{fig02}, $\frac {s(K)}{opt(K)} \preceq \left\{\frac {12}{11}, \frac {18}{15}, \frac {20}{17}\right\}$.
If $K \in {\cal K}_4\cup {\cal K}_5$, then $K$ has no responsible $1$-anchor.
By Lemma~\ref{lemma11} again, the lemma is proved.
\end{proof}

\begin{lemma}
\label{lemma18}
Suppose that no component of $H+C$ is critical. 
Then, $opt(G) < \frac {35}{22} alg(G)$.
\end{lemma}
\begin{proof}
By Step~4 of the algorithm, $alg(G) = \sum_{K\in{\cal K}} opt(K)$. 
Moreover, Definition~\ref{def08} implies $\frac {s(K)}{opt(K)} < \frac {14}{11}$. 
So, $alg(G) = \sum_{K \in {\cal K} } opt(K) \ge \frac {11}{14} \sum_K s(K) =  \frac {11}{14} |V(M_C)|$.
By Lemma \ref{lemma05}, the lemma is proved.
\end{proof}

\begin{lemma}
\label{lemma19}
Suppose that at least one component of $H+C$ is critical 
and $\frac {\sum_{i=1}^5 i|{\cal K}_i|}{|{\cal K}_{1,c}|+2|{\cal K}_{2,c}|} > \frac {5}7 r$. 
Then, $opt(G) \le r \times alg(G)$.
\end{lemma}
\begin{proof}
By Step~4 of the algorithm, $alg(G) = \sum_K opt(K)$, where $K$ ranges over 
all components of $H+C$ that is not a bad component of $H$. 
We can rewrite 
\begin{eqnarray}
\sum_{K\in{\cal K}}  opt(K) = \sum_{i\in\{0,3,4,5\} }\sum_{K\in {\cal K}_i} opt(K) + 
\sum_{i\in\{1,2\}} \sum_{K\in {\cal K}_i \setminus {\cal K}_{i,c} } opt(K) + 
\sum_{i\in\{1,2\}} \sum_{K\in {\cal K}_{i,c} } opt(K). \label{eq41}
\end{eqnarray}

By Lemma \ref{lemma05}, it suffices to show that $\frac{|V(M_C)|}{alg(G)} \le \frac 45 r$. 
We can rewrite 
\begin{eqnarray}
|V(M_C)| &=& \sum_{i\in\{0,3,4,5\} }\sum_{K\in {\cal K}_i} s(K) + 
\sum_{i\in\{1,2\}} \sum_{K\in {\cal K}_i \setminus {\cal K}_{i,c} } s(K) + 
\sum_{i\in\{1,2\}} \sum_{K\in {\cal K}_{i,c} } s(K) \nonumber \\
&\le& \sum_{i\in\{0,3,4,5\} }\sum_{K\in {\cal K}_i} \left(s(K) + (4r-6)i\right) + 
\sum_{i\in\{1,2\}} \sum_{K\in {\cal K}_i \setminus {\cal K}_{i,c} } \left(s(K) + (4r-6)i\right) \nonumber \\
&& + \sum_{i\in\{1,2\}} \sum_{K\in {\cal K}_{i,c} } \left(s(K) - \frac {(4r - 6)(5r - 7)}7 i\right) \nonumber \\
&=& \sum_{i\in\{0,3,4,5\} }\sum_{K\in {\cal K}_i} \left(s(K) + (4r-6)i\right) + 
\sum_{i\in\{1,2\}} \sum_{K\in {\cal K}_i \setminus {\cal K}_{i,c} } \left(s(K) + (4r-6)i\right) \nonumber \\
&& + \sum_{i\in\{1,2\}} \sum_{K\in {\cal K}_{i,c} } \left(s(K) - (8-4r) i\right), \label{eq42}
\end{eqnarray}
where the last equality holds because $10r^2 - 15r - 7 = 0$, while 
the inequality holds because 
\begin{eqnarray*}
\sum_{i\in\{0,3,4,5\} }\sum_{K\in {\cal K}_i} (4r-6)i + 
\sum_{i\in\{1,2\}} \sum_{K\in {\cal K}_i \setminus {\cal K}_{i,c} } (4r-6)i
&=& (4r - 6)\left( \sum_{i=1}^5 i |{\cal K}_i|-|{\cal K}_{1,c}|-2|{\cal K}_{2,c}|\right) \\
&\ge& \frac {(4r - 6)(5r - 7)}7 (|{\cal K}_{1,c}|+2|{\cal K}_{2,c}|).
\end{eqnarray*}

For each $i\in\{0,3,4,5\}$ and each $K\in{\cal K}_i$, we define $s'(K) = s(K) + (4r-6)i$. 
Similarly, for each $i\in\{1,2\}$ and each $K\in{\cal K}_i \setminus {\cal K}_{i,c}$, we define $s'(K) = s(K) + (4r-6)i$. 
Moreover, for each $i\in\{1,2\}$ and each $K\in {\cal K}_{i,c}$, we define $s'(K) = s(K) - (8 - 4r)i$. 
Then, by Eqs.~(\ref{eq41}, \ref{eq42}), 
it suffices to show that $\frac{s'(K)}{opt(K)} \le \frac{4}{5}r$ for all $K\in{\cal K}$, 
in order to show that $\frac{|V(M_C)|}{alg(G)} \le \frac {4}{5}r$. 

Consider a $K \in {\cal K}$. If $K \in {\cal K}_0$, then similarly to Lemma \ref{lemma18}, 
$\frac{s'(K)}{opt(K)} \le \frac{4}{5}r$. 
If $K \in {\cal K}_1 \setminus {\cal K}_{1,c}$, then by the second statement in Fact~\ref{fact03}, we have
\[
\frac {s'(K)}{opt(K)} \preceq \left\{\frac {4r}5, \frac {4r+2}7, \frac {4r+4}8, \frac {4r+6}{10}, \frac {4r+8}{12}, 
\frac {4r+10}{13}\right\} \preceq \frac {4r}5.
\]
If $K \in {\cal K}_2 \setminus {\cal K}_{2,c}$, then by the third statement in Fact~\ref{fact03}, we have
\[
\frac {s'(K)}{opt(K)} \preceq  \left\{\frac {8r-6}6, \frac {8r}{10}, \frac {8r+2}{12}, \frac {8r+4}{13}, \frac {8r+6}{15}\right\} 
\preceq \frac {4r}5.
\]
If $K \in {\cal K}_3 \cup {\cal K}_4 \cup {\cal K}_5$, then by the last two statements in Fact \ref{fact03}, we have
\[
\frac {s'(K)}{opt(K)} \preceq  \left\{\frac {12r}{15}, \frac {12r+2}{17}, \frac {12r-6}{11}, \frac {16r}{20}, \frac {20r-6}{25}\right\} 
\preceq \frac {4r}5.
\]
If $K \in {\cal K}_{1,c}$, then by the second statement in Fact~\ref{fact03}, we have
\[
\frac {s'(K)}{opt(K)} \preceq  \left\{ \frac {4r}5, \frac {4r+2}7\right\} \preceq \frac {4r}5.
\]
If $K \in {\cal K}_{2,c}$, by the third statement in Fact~\ref{fact03}, we have
\[
\frac {s'(K)}{opt(K)} \preceq  \left\{ \frac {8r-2}{11}, \frac {8r}{12}, \frac {8r+2}{13} \right\} \preceq \frac {4r}5.
\]
This completes the proof.
\end{proof}

\begin{theorem}
\label{thm01}
The running time of the algorithm is bounded by $O(\min\{m^2n^2, n^5\})$ and
the approximation ratio is at most $r = \frac {15 + \sqrt{505}}{20} \approx 1.874$.
\end{theorem}
\begin{proof}
It is easy to see that each of Operations~\ref{op01}, \ref{op02}, and~\ref{op03} can be done in $O(m)$. 
Since they are executed at most $O(n)$ repetitions in Step~3 of the algorithm by Lemma~\ref{lemma16}, it takes $O(nm)$ time. 
By Lemmas~\ref{lemma03} and~\ref{lemma04}, Steps~1 and~2 can be done in $O(\min\{m^2n, n^4\})$ time.
It is easy to verify that compared to these three steps, the other steps take less time. 
Since the recursion depth is $O(n)$, the algorithm takes $O(\min\{m^2n^2, n^5\})$ time in total.

We next prove that the approximation ratio is at most~$r$. 
The proof is done by induction on $n$. 
In the base case, $n \le 4$ and the algorithm outputs the optimal solution for $G$ and hence we are done. 
Now suppose that $n \ge 5$.
By Lemmas~\ref{lemma18} and~\ref{lemma19}, we only need to consider the case 
where there exists a critical component in $H+C$ and $\frac {\sum_{i=1}^5 i|{\cal K}_i|}{|{\cal K}_{1,c}|+2|{\cal K}_{2,c}|} \le \frac {5r}7$.
In this case, we have $alg(G) \ge 5(|{\cal K}_{1,c}|+2|{\cal K}_{2,c}|) + alg(G_c)$ by the last step of the algorithm.
By the inductive hypothesis, $opt(G_c) \le r\times alg(G_c)$.
According to Lemma \ref{lemma17}, we finally obtain that
\[
\frac {opt(G)}{alg(G)} \le \frac {7\sum_{i=1}^5 i |{\cal K}_i|+opt(G_c)}{5(|{\cal K}_{1,c}|+2|{\cal K}_{2,c}|)+alg(G_c)} \le r.
\]
This completes the proof.
\end{proof}

\section{Conclusions}
In this paper, we investigated the problem $MPC_v^{4+}$ to find a collection of vertex-disjoint paths,
each containing at least $4$ vertices, such that the number of vertices in these paths is maximized.
In~\cite{GFL22a}, the authors design an $O(n^8)$-time $2$-approximation algorithm based on several local improvement operations,
where $n$ is the number of vertices in the input graph $G = (V, E)$.
They asked whether better approximation algorithms are possible by a completely different method.

{\color{black}
We answered this open question affirmatively in this paper, to construct a solution on top of a maximum matching $M$.
The key observation to this success is that the maximum matching $M$ can be proven to contain at least $4/5$ of the vertices in the optimal solution.
The subsequent construction involves extending edges of $M$ into $5$-paths,
computing a maximum-weight path-cycle cover $C$ of an auxiliary graph,
and three operations to modify the achieved subgraph.
We not only reduce the running time to $O(\min\{m^2 n^2, n^5\})$, where $m$ is the number of edges in the input graph,
but also prove a better approximation ratio of $1.874$ for our algorithm.

Our design idea can be extended to $MPC_v^{5+}$ with some effort. However, new ideas are needed for the general $k \ge 6$. 
It is possible that the local operations we designed inside our algorithm can be developed into local search algorithms with better performance ratios.
Currently we do not have any inapproximability result for $MPC_v^{k+}$ for any fixed constant $k \ge 4$.
Such negative results can be interesting to pursue.
}

\section*{Acknowledgments}
ZZC is supported in part by the Grant-in-Aid for Scientific Research of the Ministry of Education, Science, Sports and Culture of Japan, under Grant No. 18K11183.
GL is supported by the NSERC Canada.
ZZ is supported by National Science Foundation of China (NSFC: 61972329),
GRF grants for Hong Kong Special Administrative Region, China (CityU 11210119, CityU 11206120, CityU11218821),
and a grant from City University of Hong Kong (CityU 11214522).


\bibliography{mfcs2023}

\begin{thebibliography}{10}

\bibitem{AN07}
K.~Asdre and S.~D. Nikolopoulos.
\newblock A linear-time algorithm for the $k$-fixed-endpoint path cover problem
  on cographs.
\newblock {\em Networks}, 50:231--240, 2007.

\bibitem{AN10}
K.~Asdre and S.~D. Nikolopoulos.
\newblock A polynomial solution to the $k$-fixed-endpoint path cover problem on
  proper interval graphs.
\newblock {\em Theoretical Computer Science}, 411:967--975, 2010.

\bibitem{BK06}
P.~Berman and M.~Karpinski.
\newblock 8/7-approximation algorithm for (1,2)-{TSP}.
\newblock In {\em Proceedings of ACM-SIAM SODA'06}, pages 641--648, 2006.

\bibitem{CCC18}
Y.~Cai, G.~Chen, Y.~Chen, R.~Goebel, G.~Lin, L.~Liu, and An~Zhang.
\newblock Approximation algorithms for two-machine flow-shop scheduling with a
  conflict graph.
\newblock In {\em Proceedings of COCOON 2018}, LNCS 10976, pages 205--217,
  2018.

\bibitem{CCL21}
Y.~Chen, Y.~Cai, L.~Liu, G.~Chen, R.~Goebel, G.~Lin, B.~Su, and A.~Zhang.
\newblock Path cover with minimum nontrivial paths and its application in
  two-machine flow-shop scheduling with a conflict graph.
\newblock {\em Journal of Combinatorial Optimization}, 43:571--588, 2022.

\bibitem{CKM10}
Z.-Z. Chen, S.~Konno, and Y.~Matsushita.
\newblock Approximating maximum edge 2-coloring in simple graphs.
\newblock {\em Discrete Applied Mathematics}, 158:1894--1901, 2010.

\bibitem{Gab83}
H.~N. Gabow.
\newblock An efficient reduction technique for degree-constrained subgraph and
  bidirected network flow problems.
\newblock In {\em Proceedings of ACM STOC'83}, pages 448--456, 1983.

\bibitem{GW20}
R.~Gomez and Y.~Wakabayashi.
\newblock Nontrivial path covers of graphs: Existence, minimization and
  maximization.
\newblock {\em Journal of Combinatorial Optimization}, 39:437--456, 2020.

\bibitem{GFL22a}
M.~Gong, J.~Fan, G.~Lin, and E.~Miyano.
\newblock Approximation algorithms for covering vertices by long paths.
\newblock In {\em Proceedings of MFCS 2022}, LIPIcs 241, pages 53:1--53:14,
  2022.

\bibitem{H83}
D.~S. Hochbaum.
\newblock Efficient bounds for the stable set, vertex cover and set packing
  problems.
\newblock {\em Discrete Applied Mathematics}, 6:243--254, 1983.

\bibitem{KLM22}
K.~Kobayashi, G.~Lin, E.~Miyano, T.~Saitoh, A.~Suzuki, T.~Utashima, and
  T.~Yagita.
\newblock Path cover problems with length cost.
\newblock In {\em Proceedings of WALCOM 2022}, LNCS 13174, pages 396--408,
  2022.

\bibitem{K09}
A.~Kosowski.
\newblock Approximating the maximum $2$- and $3$-edge-colorable subgraph
  problems.
\newblock {\em Discrete Applied Mathematics}, 157:3593--3600, 2009.

\bibitem{MV80}
S.~Micali and V.~V. Vazirani.
\newblock An ${O}(\sqrt{|{V}|} |{E}|)$ algorithm for finding maximum matching
  in general graphs.
\newblock In {\em Proceedings of IEEE FOCS'80}, pages 17--27, 1980.

\bibitem{N21}
M.~Neuwohner.
\newblock An improved approximation algorithm for the maximum weight
  independent set problem in $d$-claw free graphs.
\newblock In {\em Proceedings of STACS 2021}, pages 53:1--53:20, 2021.

\bibitem{PH08}
L.~L. Pao and C.~H. Hong.
\newblock The two-equal-disjoint path cover problem of matching composition
  network.
\newblock {\em Information Processing Letters}, 107:18--23, 2008.

\bibitem{RT14}
R.~Rizzi, A.~I. Tomescu, and V.~M{\"a}kinen.
\newblock On the complexity of minimum path cover with subpath constraints for
  multi-assembly.
\newblock {\em BMC Bioinformatics}, 15:S5, 2014.

\end{thebibliography}

\end{document}